\documentclass[11pt]{article}
\usepackage[ruled,vlined,linesnumbered]{algorithm2e}
\usepackage{color}
\usepackage{amsfonts}
\usepackage{amssymb}
\usepackage{amstext}
\usepackage{amsmath}
\usepackage{amsthm}
\usepackage{authblk}
\usepackage{enumitem}
\usepackage{geometry}
\usepackage{mathrsfs}
\usepackage{tikz}
\usepackage{float}
\usepackage{accents}
\usepackage[backref=page]{hyperref}
\usepackage{cleveref}
\usepackage{caption}
\usepackage{microtype}
\usepackage[colorinlistoftodos,prependcaption]{todonotes}
\usepackage{tcolorbox}
\usepackage{thmtools, thm-restate}
\usepackage{kbordermatrix}
\usetikzlibrary{arrows.meta}

\newtheorem{theorem}{Theorem}[section]
\newtheorem{lemma}[theorem]{Lemma}

\newtheorem{proposition}[theorem]{Proposition}
\newtheorem{claim}[theorem]{Claim}

\theoremstyle{definition}
\newtheorem{definition}[theorem]{Definition}

\newtheorem{property}[theorem]{Property}

\DeclareMathAlphabet{\mathbbd}{U}{bbold}{m}{n}

\newcommand{\focus}{\textsc{Focus}\xspace}
\newcommand{\pminsubset}{\textsc{ExtractAndFocus}\xspace}

\newcommand{\cut}{\operatorname{cut}}
\newcommand{\shore}{\operatorname{shore}}
\newcommand{\mincut}{\operatorname{mincut}}

\newcommand{\congestion}{\mathsf{cong}}

\newcommand{\OPT}{\mathtt{OPT}}
\newcommand{\R}{\mathbb R}

\newcommand{\N}{\mathbb N}

\newcommand{\0}{\mathbbd 0}
\newcommand{\1}{\mathbbd 1}

\newcommand{\ceil}[1]{\ensuremath{\left\lceil#1\right\rceil}}

\newcommand{\ang}[1]{\ensuremath{\left<#1\right>}}

\newcommand{\cC}{\mathcal C}

\newcommand{\mTSP}{$\mathsf{Metric\text{-}TSP}$}
\newcommand{\TSP}{$\mathsf{TSP}$}

\newcommand{\kECSM}{$k\text{-}\mathsf{ECSM}$}
\newcommand{\kECSS}{$k\text{-}\mathsf{ECSS}$}

\DeclareMathOperator{\eps}{\varepsilon}

\DeclareMathOperator*{\argmin}{arg\,min}

\DeclareMathOperator{\poly}{poly}

\DeclareMathOperator{\supp}{supp}

\tikzstyle{vertex}=[circle,draw,align=center]
\tikzstyle{node}=[circle,draw,fill=white!100,inner sep=0pt,minimum size=0.15cm,align=center]
\tikzstyle{weight} = [font=\small, black]
\tikzstyle{edge} = [draw,-]
\tikzstyle{selected edge} = [draw,line width=3pt,-,gray!70]
\tikzstyle{matched edge} = [draw,line width=3pt,-]
\tikzstyle{dashed edge} = [draw,dashed]
\tikzstyle{selected vertex} = [node, fill=gray!70, text=white]

\tikzstyle{remove selected vertex} = [vertex, fill=white!100]
\tikzstyle{blue colored edge} = [draw,line width=3pt,-,blue!80]
\tikzstyle{blue colored curved edge left} = [draw,line width=3pt,-,blue!80,bend left]
\tikzstyle{blue colored curved edge right} = [draw,line width=3pt,-,blue!80,bend right]
\tikzstyle{green colored edge} = [draw,line width=3pt,-,green!80]
\tikzstyle{remove selected edge} = [draw,line width=4pt,-,white!100]
\tikzstyle{remove matched edge} = [draw,line width=4pt,-,white!80,bend right]

\SetArgSty{textup}
\DontPrintSemicolon
\SetKw{Continue}{continue}
\SetKwInOut{Input}{input}
\SetKwInOut{Output}{output}

\hypersetup{
    colorlinks = true,
    linkcolor={blue},
    citecolor={olive},
    urlcolor={cyan}
}

\usepackage{setspace}
\newcounter{mycomment}
\newcommand{\Comment}[2]{
\refstepcounter{mycomment}
{
\setstretch{0.7}
\todo[inline,
backgroundcolor=black!30!white
, size=\small]{
\textbf{Comment[#1\themycomment]:} \\ #2}
}}

\newcommand{\sorrachai}[1]{
    \Comment{SY}{#1}
}

\title{Approximating the Held--Karp Bound for Metric TSP in \\ Nearly Linear Work and Polylogarithmic Depth\thanks{This project has received funding from the European Research Council (ERC) under the European Union’s Horizon 2020 research and innovation programme (grant agreement nos. 759557--ALGOCom and 805241--QIP), and the ERC Starting Grant (CODY 101039914).
It is partially supported by Dr.~Max R\"ossler, the Walter Haefner Foundation and the ETH Z\"urich Foundation.}}

\author[1]{Zhuan Khye Koh\thanks{Part of this work was done while the author was at Centrum Wiskunde \& Informatica.}}
\author[2]{Omri Weinstein}
\author[2,3]{Sorrachai Yingchareonthawornchai}

\affil[1]{Department of Computer Science, Boston University, USA.}
\affil[2]{Department of Computer Science, Hebrew University of Jerusalem, Israel.}
\affil[3]{Institute for Theoretical Studies, ETH Zurich, Switzerland.}
\date{\tt{zkkoh@bu.edu, omriwe@cs.huji.ac.il, sorrachai.cp@gmail.com}}

\begin{document}

\maketitle
\thispagestyle{empty}

\begin{abstract}
We present a nearly linear work parallel algorithm for approximating the Held--Karp bound for \mTSP.
Given an edge-weighted undirected graph $G=(V,E)$ on $m$ edges and $\varepsilon>0$, it returns a $(1+\varepsilon)$-approximation to the Held--Karp bound with high probability, in $\tilde{O}(m/\varepsilon^4)$ work and $\tilde{O}(1/\varepsilon^4)$ depth\footnote{The soft-O notation $\tilde O (\cdot)$ hides polylogarithmic factors in $m $ and $\varepsilon^{-1}$.}.
While a nearly linear time \emph{sequential} algorithm was known for almost a decade (Chekuri and Quanrud~'17), it was not known how to simultaneously achieve nearly linear work alongside  polylogarithmic depth.
Using a reduction by Chalermsook et al.~'22, we also give a parallel algorithm for computing a $(1+\varepsilon)$-approximate fractional solution to the $k$-edge-connected spanning subgraph (\kECSS) problem, with similar complexity.

To obtain these results, we introduce a notion of \emph{core-sequences} for the parallel Multiplicative Weights Update (MWU) framework (Luby--Nisan~'93, Young~'01).
For \mTSP~and \kECSS, core-sequences enable us to exploit the structure of approximate minimum cuts to reduce the cost per iteration and/or the number of iterations.
The acceleration technique via core-sequences is generic and of independent interest.
In particular, it improves the best-known iteration complexity of MWU algorithms for packing/covering LPs from $\poly(\log \mathsf{nnz}(A))$ to polylogarithmic in the product of cardinalities of the core-sequence sets, where $A$ is the constraint matrix of the LP.
For certain implicitly defined LPs such as the \kECSS~LP, this yields an exponential improvement in depth.

\end{abstract}

\newpage
\tableofcontents
\thispagestyle{empty}

\newpage
\pagenumbering{arabic}

\section{Introduction}

The Traveling Salesman Problem (\TSP) is among the most well-studied problems in combinatorial optimization and theoretical computer science, constituting almost an entire field of research on its own \cite{LawlerLRS91, ApplegateBCC06, GuntinP07, Cook12}. The input to \TSP~is a graph $G=(V, E)$ with positive edge costs $c\in \R^E_{>0}$, and the goal is to find a minimum cost Hamiltonian cycle in $G$.
It is well known that the problem is inapproximable already on undirected graphs, by a reduction from the Hamiltonian cycle problem.
This impossibility result, as well as
many practical applications, motivated the study of \mTSP.
Given $(G,c)$, the goal now is to find a minimum-cost \emph{tour} in $G$, where a tour is a closed walk which visits all the vertices.
So, every vertex can be visited multiple times.
By considering the \emph{metric completion} of $(G,c)$, denoted by $(\hat G := (V,{V\choose 2}), \hat c)$ where
\[\hat c_{\{u,v\}} := \min\left\{\sum_{e\in E(P)} c_e: P \text{ is a $u$-$v$ path in $G$}\right\},\]
one sees that \mTSP~is a special case of \TSP~with the triangle inequality, i.e., $\hat c_{u v} \leq \hat c_{u w}+ \hat c_{w v}$ for all $u, v, w \in V$.
Note that $(G,c)$ is an implicit representation of its metric completion $(\hat G, \hat c)$.

Having the triangle inequality makes \TSP~substantially more tractable.
While \mTSP~remains $\mathsf{APX}$-Hard (Lampis~\cite{Lampis14} showed that there is no $185/184$-approximation unless $\bf P = NP$),  the landmark algorithm of Christofides~\cite{Christofides76} and Serdyukov~\cite{Serdyukov78} gives an elegant $3/2$-approximation.
Notably, in the special case of Euclidean \TSP, Arora \cite{Arora98} and Mitchel \cite{Mitchell99} gave polynomial-time approximation schemes.

Improving the $3/2$-approximation for \mTSP~was a longstanding open problem in theoretical computer science
until 2021, when Karlin, Klein and Oveis Gharan \cite{KarlinKO21} gave a slightly improved $(3/2-10^{-36})$-approximation algorithm, using the theory of stable polynomials.
This was recently improved to $3/2-10^{-34}$ by Gurvits, Klein and Leake \cite{GurvitsKL24}.
Like many algorithms for \mTSP, their method exploits an optimal solution to the following linear program (LP) relaxation, also known as the \emph{Subtour Elimination LP}~\cite{DantzigFJ54}:

\begin{equation}\label{def_SE_LP}
\begin{aligned}
\operatorname{SE}(\hat G,\hat c) := &\min\;  \sum_{u,v}\hat c_{\{u,v\}}y_{\{u,v\}}  \\
&\text { s.t. } \sum_{u} y_{\{u,v\}}=2  \qquad\;\;\, \forall\, v \in V \\
&\quad\; \sum_{u\in S, v\notin S} y_{\{u,v\}} \geq 2 \qquad \forall\, \emptyset \subsetneq S \subsetneq V \\
&\qquad\; y_{\{u,v\}} \geq 0 \qquad\qquad\;\; \forall\, u,v\in V.
\end{aligned}
\end{equation}

Notice that \eqref{def_SE_LP} is defined on the metric completion of $(G,c)$.
We have a variable $y_{\{u,v\}}$ for every pair of vertices $u,v\in V$.
The first set of constraints (degree constraints) forces each vertex to be incident to exactly two edges.
The second set of constraints (subtour elimination constraints) forces connectivity.
Observe that they imply the inequality $y_{\{u,v\}}\leq 1$ for all $u,v\in V$.
Clearly, $\operatorname{SE}(\hat G,\hat c)$ is a lower bound on the length of an optimal tour.

The optimal value of
 the Subtour Elimination LP \eqref{def_SE_LP} is also called the \emph{Held--Karp bound}, as it coincides (by Lagrange duality) with another lower bound given by Held and Karp \cite{HeldK70} based on the notion of \emph{1-trees}.
The well-known `4/3 conjecture' \cite{Goemans95} postulates that the integrality gap of \eqref{def_SE_LP} is at most $4/3$.
Wolsey~\cite{Wolsey1980} and Shmoys and Williamson~\cite{ShmoysW90} gave an upper bound of $3/2$, which was recently improved by Karlin et al.~\cite{KarlinKO22} to $3/2-10^{-36}$.

\paragraph{Importance of Solving the Subtour Elimination LP.}
Apart from being a crucial component in the breakthrough result of \cite{KarlinKO21}, the importance of solving \eqref{def_SE_LP} has been recognized since the dawn of mathematical programming.
It inspired the \emph{cutting plane} method, introduced by Dantzig, Fulkerson and Johnson \cite{DantzigFJ54} to solve \TSP~exactly.
Applegate, Bixby, Chvátal and Cook \cite{ApplegateBCC03} implemented their method into the Concorde solver, which is capable of solving very large real-world instances.
The ellipsoid method can solve \eqref{def_SE_LP} using a minimum cut separation oracle, but it is impractical for large graphs.
Likewise, it is possible to reformulate the subtour-elimination constraints in \eqref{def_SE_LP} as a flow-based extended formulation \cite{CarrL02}, but the number of variables and constraints becomes cubic in $|V|$.
For more context, we refer the reader to \cite{CQ17} and the references therein.

\paragraph{\kECSM~and the Cut Covering LP.}
A problem closely related to \mTSP~is the \emph{k-edge-connected spanning multi-subgraph} problem (\kECSM).
In \kECSM, given an undirected graph $G=(V,E)$ with positive edge costs $c\in \R^E_{>0}$ and an integer $k\geq 1$, the goal is to find a minimum-cost $k$-edge-connected multi-subgraph of $G$ which spans $V$.
A \emph{multi-subgraph} of $G$ is a subgraph of $G$ with the exception that every edge can be taken multiple times,
but every copy needs to be payed for.
The canonical LP relaxation for \kECSM~is the following \emph{Cut Covering LP}:

\begin{equation}\label{lp:cut_cover}
\begin{aligned}
\operatorname{CC}(G,c,k) := &\min\; c^\top y  \\
&\text { s.t. } \sum_{e\in \delta(S)} y_e\geq k  \qquad \forall\, \emptyset \subsetneq S \subsetneq V \\
&\qquad\; y_e \geq 0 \qquad\qquad\;\; \forall\, e\in E.
\end{aligned}
\end{equation}

Clearly, $\operatorname{CC}(G,c,k) = k\operatorname{CC}(G,c,1)$ for all $k\geq 0$.
Cunningham~\cite{MonmaMP90} and Goemans and Bertsimas~\cite{GoemansB93} showed that for any graph $G$ with edge costs $c\in \R^E_{>0}$, the optimal value of the Subtour Elimination LP \eqref{def_SE_LP} for the metric completion $(\hat G, \hat c)$ coincides with the optimal value of \eqref{lp:cut_cover} for $(G,c)$ with $k=2$, i.e., $\operatorname{SE}(\hat G, \hat c) = \operatorname{CC}(G,c,2)$.
So, it suffices to solve \eqref{lp:cut_cover} in order to compute the Held--Karp bound.

The Cut Covering LP is perceived as easier to solve than the Subtour Elimination LP for the following 2 reasons.
Firstly, \eqref{lp:cut_cover} is a \emph{covering LP} because it only has $\geq$ constraints.
Secondly, \eqref{lp:cut_cover} only has $m$ variables, whereas \eqref{def_SE_LP} has $n^2$ variables.
This opens up the possibility of a fast algorithm for computing the Held--Karp bound.
A fast algorithm has several implications ranging from approximation algorithms to exact algorithms for \TSP.

\paragraph{Combinatorial Algorithms for the Cut Covering LP.}
The inefficiency of general LP solvers for \eqref{def_SE_LP} and \eqref{lp:cut_cover} motivated the development of \emph{combinatorial} algorithms which exploit the underlying graph structure of the LPs.
Held and Karp \cite{HeldK70} proposed a simple iterative procedure for approximating $\operatorname{CC}(G,c,2)$, based on repeated \emph{minimum spanning tree} computations.
Even though it provides good estimates in practice, there are no provable guarantees on the convergence rate.

The next line of development was based on the \emph{multiplicative weights update (MWU)} method \cite{AHK12} for approximately solving packing and covering LPs. In their influential work~\cite{PST95}, Plotkin, Shmoys and Tardos gave a $(1+\varepsilon)$-approximation for the Held--Karp bound in $O(n^4\log^6 n/\varepsilon^2)$ time.
Garg and Khandekar improved it to $O(m^2\log^2 m/\varepsilon^2)$~\cite{Khandekar04}.
This series of work culminated in the nearly linear time algorithm of Chekuri and Quanrud~\cite{CQ17}, running in $O(m\log^4 n/\varepsilon^2)$. \\

The aforementioned algorithms are inherently \emph{sequential}, provably requiring $\Omega(m)$ MWU iterations.
Our main contribution is a \emph{parallel} algorithm for approximately solving \eqref{lp:cut_cover} in nearly linear work.
For $k=2$, this yields a $(1+\varepsilon)$-approximation to the Held--Karp bound:

\begin{theorem}[Main]\label{thm:parallel_TSP}
Let $G$ be an undirected graph with $n$ nodes, $m$ edges and edge costs $c\in \R^m_{>0}$.
For every $0<\varepsilon < 0.5$, there is a randomized parallel algorithm that computes a $(1+\varepsilon)$-approximation to the Held--Karp bound with high probability.
The algorithm runs in $\tilde O(m/\varepsilon^4)$ work and $\tilde O(1/\varepsilon^4)$ depth.
\end{theorem}

Using the reduction by~\cite{ChalermsookHNSS22}, we extend our algorithm to solve the LP relaxation of the \emph{$k$-edge-connected spanning subgraph} problem (\kECSS). In \kECSS, given an undirected graph $G = (V,E)$ with nonnegative edge costs $c\in \R^E_{\geq 0}$
and an integer $k\geq 1$, the goal is to find a minimum-cost $k$-edge-connected subgraph of $G$ which spans $V$.
In other words, it is obtained from \kECSM~by imposing the extra condition that every edge can only be taken at most once.
The canonical LP relaxation for \kECSS~is given by \eqref{lp:cut_cover} with the additional upper bounds $y_e\leq 1$ for all $e\in E$.
\begin{theorem}\label{thm:parallel_kECSS}
Let $G$ be an undirected graph with $n$ nodes, $m$ edges and edge costs $c\in \R^m_{\geq0}$.
For every $0<\varepsilon < 0.5$, there is a randomized parallel algorithm that computes a $(1+\varepsilon)$-approximate solution to the \kECSS~LP with high probability.
The algorithm runs in $\tilde O(m/\varepsilon^4)$ work and $\tilde O(1/\varepsilon^4)$ depth\footnote{This hides a $\log \left(\frac{\max_{e\in E} c_e}{\min_{e\in E} c_e}\right)$ factor.}.
\end{theorem}

\subsection{The MWU Framework for Packing/Covering LPs}
\label{sec:mwu}

Since our algorithm follows the  MWU approach, we now provide a brief overview of this framework and our innovation within it.

Given a nonnegative matrix $A\in \R^{m\times N}_{\geq 0}$, a \emph{packing LP} is of the form
\begin{equation}\label{lp:pack}
\max_{x\geq \0} \{\ang{\1,x}:Ax \leq \1\}.
\end{equation}
Its dual is a \emph{covering LP}
\begin{equation}\label{lp:cover}
\min_{y\geq \0} \{\ang{\1,y}:A^{\top}y \geq \1\}.
\end{equation}
Clearly, \eqref{lp:cut_cover} can be converted into \eqref{lp:cover} by scaling the rows and columns.

\paragraph{Width-Independent MWU.}
Since the seminal work of Plotkin, Shmoys and Tardos \cite{PST95}, the MWU method \cite{AHK12} has become the main tool in designing low-accuracy solvers for positive linear programs.
For this exposition, let us focus on packing LPs (covering LPs are analogous).
Given $\varepsilon>0$, the algorithm of~\cite{PST95} computes a $(1-\varepsilon)$-approximate solution to \eqref{lp:pack} by iteratively calling a \emph{linear minimization oracle} on a \emph{weighted average} of the constraints.
Given a convex domain $\mathcal{D}$ and weights $w^{(t)}\in \R^m_{\geq 0}$, the oracle returns $g^{(t)} := \argmin_{x\in \mathcal{D}} \langle A^\top w^{(t)}, x\rangle$.
Based on $g^{(t)}$, the weights $w^{(t)}$ are updated multiplicatively.
Unfortunately, the updates in \cite{PST95} can be very slow, as they need to be scaled by
$1/\rho$, where $\rho:= \max_t \|Ag^{(t)}\|_\infty$ is
the \emph{width} of the oracle.

This drawback was overcome in the  subsequent influential work of Garg and Könemann~\cite{journals/siamcomp/GargK07}, who gave the first \emph{width-independent} MWU algorithm for packing/covering LPs.
Their algorithm requires a similar oracle,  which solves the following subproblem in every iteration $t$: Given weights $w^{(t)}\in \R^m_{\geq 0}$, find a minimum weight column of $A$
\begin{equation}\label{sys:gk_sub}
    j^*\in \argmin_{j\in [N]} (A^\top w^{(t)})_j.
\end{equation}
The main innovation of \cite{journals/siamcomp/GargK07} was to \emph{adaptively} scale the oracle response in order to achieve width independence.
Specifically, they set $g^{(t)} := \alpha e_{j^*}$ where $\alpha$ is chosen such that $\|Ag^{(t)}\|_\infty = 1$.
Denoting $x^{(t)} := \sum_{s=1}^{t} g^{(s)}$, they showed that their algorithm can be stopped as soon as $\|Ax^{(t)}\|_\infty \geq \Omega(\log m/\varepsilon^2)$, at which point $x^{(t)}/\|Ax^{(t)}\|_\infty$ is a $(1-\varepsilon)$-approximate solution to the packing LP.
Since $A$ has $m$ rows, the number of iterations is $O(m\log m/\varepsilon^2)$.

\paragraph{Epoch-Based MWU.}
To implement \cite{journals/siamcomp/GargK07} more efficiently, Fleischer~\cite{journals/siamdm/Fleischer00} introduced the notion of \emph{epochs}.
In every iteration $t$, the algorithm maintains an extra parameter $\lambda^{(t)}$, which lower bounds the minimum weight of a column of $A$ with respect to $w^{(t)}$.
The oracle is modified to return a column with weight less than $(1+\varepsilon)\lambda^{(t)}$, i.e., any coordinate in
\begin{equation}\label{sys:fleischer_sub}
    B^{(t)}:=\{j\in [N]: (A^\top w^{(t)})_j < (1+\varepsilon)\lambda^{(t)}\},
\end{equation}
or it concludes that $B^{(t)} = \emptyset$.
In the latter case, $\lambda^{(t)}$ is multiplied by $1+\varepsilon$.
An \emph{epoch} is
a maximal sequence of consecutive iterations
with the same value of $\lambda^{(t)}$.
The iteration bound remains the same as~\cite{journals/siamcomp/GargK07}, while the number of epochs is $O(\log m/\varepsilon^2)$.
Using this idea, Fleischer developed faster algorithms for multicommodity flow~\cite{journals/siamdm/Fleischer00}.

\paragraph{Clearing an Epoch in the Case of \mTSP.}  Since the weights $w^{(t)}$ are nondecreasing, the set $B^{(t)}$ is nonincreasing during an epoch, and the epoch ends when $B^{(t)} = \emptyset$.
We refer to this process as \emph{clearing an epoch}.
For the Cut Covering LP, given $\lambda\in \R$, clearing an epoch means to iteratively apply MWU on cuts with weight less than $(1+\varepsilon)\lambda$, until the minimum cut has weight at least $(1+\varepsilon)\lambda$.

Recall that in the case of \mTSP, the goal is to solve the Cut Covering LP \eqref{lp:cut_cover} of an input graph $G$ with $n$ nodes and $m$ edges.
The subproblems \eqref{sys:gk_sub} and \eqref{sys:fleischer_sub} correspond to finding an (approximate) minimum cut in $G$ with edge weights $w^{(t)}$.
A minimum cut can be computed in $\tilde O(m)$ time \cite{Karger00,conf/soda/HenzingerLRW24}, while updating the edge weights can be done in $O(m)$ time.
So, a naive implementation runs in $\tilde O(m^2/\varepsilon^2)$ time.
The key idea of Chekuri and Quanrud~\cite{CQ17} was to exploit the `correlation'
between these $\tilde O(m/\varepsilon^2)$ minimum cuts, so as to design a $\tilde O(m/\varepsilon^2)$-time algorithm for maintaining an (approximate) minimum cut under increasing edge weights.
They achieved this by designing clever data structures for the incremental minimum cut problem, as well as for updating the edge weights in a lazy fashion.
While the work of \cite{CQ17} led to a $(3/2+\varepsilon)$-approximation algorithm for \mTSP~running in $\tilde O(m/\varepsilon^2 + n^{1.5}/\varepsilon^3)$ time~\cite{arxiv/ChekuriQ18}, it is inherently sequential as the MWU methods of \cite{journals/siamcomp/GargK07,journals/siamdm/Fleischer00} may require $\tilde \Theta(m/\varepsilon^2)$ iterations. As such, a prerequisite for parallelizing this result
is a width-independent \emph{parallel} MWU method, which we discuss next.

\paragraph{Parallel Algorithms for Clearing an Epoch.}

The basic idea for parallelizing the MWU methods of \cite{journals/siamcomp/GargK07,journals/siamdm/Fleischer00}, originating from the work of Luby and Nisan~\cite{LN93}, is to update \emph{all} the coordinates in $B^{(t)}$.
We describe the simplified version given by Young~\cite{arxiv/Young14}.
In every iteration $t$, if $B^{(t)} \neq \emptyset$, then $g^{(t)}$ is set as
\begin{equation}\label{eq:young_update}
g^{(t)}_j := \begin{cases} \alpha x^{(t-1)}_j, &\text{ if }j\in B^{(t)} \\
0, &\text{ otherwise,}
\end{cases}
\end{equation}
where $\alpha$ is again chosen such that $\|Ag^{(t)}\|_\infty = 1$.
Otherwise, $\lambda^{(t)}$ is multiplied by $1+\varepsilon$.
Note that the variables in $B^{(t)}$ are incremented multiplicatively.
With this modification, Young showed that every epoch has $O(\log m\log(N\log(m)/\varepsilon)/\varepsilon^2)$ iterations.
Since there are $O(\log m/\varepsilon^2)$ epochs, the total number of iterations is $\poly(\log(mN)/\varepsilon)$.
This result was extended to mixed packing and covering LPs~\cite{conf/focs/Young01,arxiv/Young14}.

Unlike sequential MWU methods, parallel MWU methods do not readily apply to implicit LPs $(N\gg m)$.
For implicit LPs, $N$ is usually exponential in $m$, so the iteration bound becomes linear in $m$, losing its advantage over~\cite{journals/siamcomp/GargK07,journals/siamdm/Fleischer00}.
We remark that one can modify the initialization in \cite{conf/focs/Young01,arxiv/Young14} such that the number of iterations is proportional to $\log |B^{(t)}|$ instead of $\log N$ (see \Cref{sec:parallel_mwu_young}).
Even so, $B^{(t)}$ can still be very large.
To make matters worse, these large sets make each iteration prohibitively expensive.
For the Cut Covering LP, by a result of Henzinger and Williamson~\cite{HenzingerW96} on the number of approximate minimum cuts in a a graph, we know that $B^{(t)} = O(n^2)$ as long as $\varepsilon<1/2$.
However, this still precludes a nearly linear work implementation because we may have to update $\Theta(n^2)$ coordinates.

\subsection{Our Approach}

\paragraph{Clearing an Epoch Using a Core-Sequence.} We present a general framework for clearing an epoch, by introducing the notion of \emph{core-sequence}.
The basic idea is as follows.
In every iteration $t$, instead of updating all the coordinates in $B^{(t)}$ as in~\eqref{eq:young_update}, we only update
a fixed subset $\tilde B_1\subseteq B^{(t)}$.
In particular, we keep updating the variables in $ \tilde B_1\cap B^{(t)}$ until it becomes empty.
When this happens, we say that the set $\tilde B_1$ is \emph{cleared}.
Then, we pick another fixed subset $\tilde B_2\subseteq B^{(t)}$ to update.
This process is repeated until the epoch is cleared, i.e., $B^{(t)} = \emptyset$.

\begin{definition}
\label{def:core-sequence}
Fix an epoch and let $t_0$ be its first iteration.
Let $\tilde{\mathcal B} = (\tilde B_1, \tilde B_2, \dots, \tilde B_\ell)$ be a sequence of sets from $B^{(t_0)}$.
In every iteration $t\geq t_0$, suppose that we set
\begin{equation}\label{eq:core_sequence_update}
g^{(t)}_j := \begin{cases} \alpha x^{(t-1)}_j, &\text{ if }j\in \tilde B_{i(t)}\cap B^{(t)} \\
0, &\text{ otherwise,}
\end{cases}
\end{equation}
where $i(t)$ denotes the smallest index such that $\tilde B_i\cap B^{(t)} \neq \emptyset$, and $\alpha$ is chosen such that $\|Ag^{(t)}\|_\infty = 1$.
Let $t_1$ be the first iteration when $(\cup_{i=1}^\ell \tilde B_i) \cap B^{(t_1)} = \emptyset$.
If $B^{(t_1)} = \emptyset$, then $\tilde{\mathcal B}$ is called a \emph{core-sequence} of the epoch.
\end{definition}

Core-sequences capture the aforementioned epoch-based MWU methods.
If we choose $\tilde{\mathcal B}$ such that $|\tilde B_i| = 1$ for all $i\in [\ell]$, then we obtain an instantiation of Fleischer's sequential MWU method~\cite{journals/siamdm/Fleischer00}.
On the other hand, if we choose $\tilde{\mathcal B} = (B^{(t_0)})$, then we recover the parallel MWU method of \cite{conf/focs/Young01,arxiv/Young14}.
The general guarantee can be informally stated as follows.

\begin{theorem}[MWU with Core-Sequence] \label{thm:mwu tradeoff}
    Suppose that \Cref{eq:core_sequence_update} can be computed using $f(|\tilde B_{i(t)}|)$ work and $\tilde O(1)$ depth. Given a core-sequence $\tilde{\mathcal B} = (\tilde B_1, \tilde B_2, \ldots, \tilde B_{\ell})$ of an epoch, the epoch can be cleared using $\tilde O(\sum_{i=1}^\ell f(|\tilde B_i|)\log(|\tilde B_i|\log m/\varepsilon)/\varepsilon^2)$ work and $\tilde O(\sum_{i=1}^\ell\log(|\tilde B_i|\log m/\varepsilon)/\varepsilon^2)$ depth.
\end{theorem}

The formal proof is given in \Cref{sec:parallel MWU}. It is simple and unifies the results of \cite{journals/siamcomp/GargK07,journals/siamdm/Fleischer00,LN93,conf/focs/Young01,arxiv/Young14}.
We may assume that the function $f$ in \Cref{thm:mwu tradeoff} satisfies $f(x) =\Omega(x)$.
This is because in the first iteration $t_0$ of an epoch, $g^{(t_0)}$ already has support size $|\tilde B_1|$.
Hence, the work per iteration depends at least linearly on the size of each set in the core-sequence.
On the other hand, the depth (number of iterations) depends linearly on the length of the core-sequence, and logarithmically on the size of each constituent set.

\Cref{thm:mwu tradeoff} gives a generic tool for reducing both the
work per iteration and number of iterations, assuming that we can find a \emph{short} core-sequence which consists of \emph{small} sets.
For explicit LPs, it is unclear whether such core-sequences exist.
However, it is conceivable that they may exist for implicit LPs, because a lot of the coordinates in $[N]$ are related.

\paragraph{Core-Sequence for the Cut Covering LP.}

For the Cut Covering LP, we prove the existence of a short core-sequence consisting of small sets.
Furthermore, it can be computed efficiently.

\begin{theorem}\label{thm:core_sequence}
    Let $G$ be an undirected graph with $n$ nodes, $m$ edges and edge costs $c\in \R^m_{>0}$.
    When running an epoch-based MWU algorithm on the Cut Covering LP of $(G,c)$, every epoch has a core-sequence $\tilde{\mathcal{B}} = (\tilde B_1, \tilde B_2, \dots, \tilde B_{\ell})$ such that $|\tilde B_i| \leq \tilde O(n)$ for all $i\in [\ell]$ and $\ell = \tilde O(1)$. The core-sequence can be computed using $\tilde O(m/\varepsilon^2)$ work and $\tilde O(1/\varepsilon^2)$ depth.
\end{theorem}

Recall that the update \eqref{eq:young_update} touches $\Theta(n^2)$ coordinates for the Cut Covering LP.
Using the core-sequence given by \Cref{thm:core_sequence}, the update \eqref{eq:core_sequence_update} now only touches $\tilde{O}(n)$ coordinates, at the cost of increasing the depth by a factor of $\tilde O(1)$.
In the language of graphs, \Cref{thm:core_sequence} can be interpreted as follows.
Given edge weights $w^{(t)}$ and a lower bound $\lambda^{(t)}$ on the minimum weight of a cut, there exists a sequence of $\tilde O(1)$ sets with $\tilde O(n)$ cuts each, such that clearing them in the order of the sequence ensures that every cut has weight at least $(1+\varepsilon)\lambda^{(t)}$.
The key insight stems from the observation that \emph{updating a carefully chosen sequence of cuts can increase the weight of all approximate minimum cuts}.
This intuition is formalized using submodularity and posimodularity of the cut function, as we explain in \Cref{sec:tech_overview}.

With Theorems \ref{thm:mwu tradeoff} and \ref{thm:core_sequence}, a naive computation of \eqref{eq:core_sequence_update} for the Cut Covering LP takes $\tilde O(mn)$ work.
By leveraging the canonical cut data structure of~\cite{CQ17}, we show that \eqref{eq:core_sequence_update} can be computed in $\tilde O(m)$ work.
This is the final ingredient for obtaining a nearly linear work parallel algorithm for approximating the Cut Covering LP.

\paragraph{Core-Sequence for the \kECSS~LP.}

The \kECSS~LP is not a covering LP due to the upper bounds $y_e\leq 1$ for all $e\in E$.
However, it can be transformed into a covering LP by replacing the upper bounds with Knapsack Cover (KC) constraints~\cite{CarrFLP00}.
Unfortunately, $|B^{(t)}|$ can be as large as $\Theta(m^k)$ for this LP.
Hence, from the previous discussion, the standard parallel MWU method~\cite{conf/focs/Young01,arxiv/Young14} terminates in $\tilde O(k/\varepsilon^4)$ iterations.
Furthermore, an iteration can take $\Omega(m^k)$ work.

By leveraging the connection between this LP and the Cut Covering LP~\cite{ChalermsookHNSS22}, we prove the existence of a short core-sequence with small sets.
Moreover, it can be computed efficiently.

\begin{theorem}\label{thm:core_sequence_kecss}
    Let $G$ be an undirected graph with $n$ nodes, $m$ edges and edge costs $c\in \R^m_{\geq0}$.
    When running an epoch-based MWU algorithm on the \kECSS~LP of $(G,c)$ with KC constraints, every epoch has a core-sequence $\tilde{\mathcal{B}} = (\tilde B_1, \tilde B_2, \dots, \tilde B_{\ell})$ such that $|\tilde B_i| \leq \tilde O(n)$ for all $i\in [\ell]$ and $\ell = \tilde O(1)$. The core-sequence can be computed using $\tilde O(m/\varepsilon^2)$ work and $\tilde O(1/\varepsilon^2)$ depth\footnote{This hides a $\log \left(\frac{\max_{e\in E} c_e}{\min_{e\in E} c_e}\right)$ factor.}.
\end{theorem}

We remark that the core-sequence in \Cref{thm:core_sequence_kecss} is not defined with respect to \eqref{eq:core_sequence_update}.
Instead, we consider a different update rule that exploits the structure of the \kECSS~LP, in order to construct a shorter core-sequence than what we could achieve with \eqref{eq:core_sequence_update}.

The advantage of this core-sequence is twofold.
It gives rise to an MWU algorithm which terminates in $\tilde O (1/\varepsilon^4)$ iterations.
This represents an \emph{exponential} improvement in depth when $k = \Omega(m)$.
Additionally, we show that it enables a nearly linear work implementation because every set in the core-sequence is small.

\subsection{Related Work}

The parallel MWU framework for positive LPs has received a lot of attention since the work of Luby and Nisan \cite{LN93, conf/focs/Young01, AK2008, BBR97, BB05, AzO16_ParallelMWU}.
This line of work culminated in the algorithm of \cite{conf/icalp/MahoneyRWZ16}, which achieves an iteration complexity of $\tilde O(\log^2 (\mathsf{nnz}(A))/\eps^2)$ for packing/covering LPs.
They have also extended this to \emph{mixed} packing-covering LPs with an extra $1/\varepsilon$ factor.
In the sequential setting, Allen-Zhu and Orrechia \cite{AO15} combined width-independence with Nesterov-like acceleration \cite{Nes05} to get a randomized MWU algorithm with running time $\tilde O(\mathsf{nnz}(A)/\varepsilon)$.
The main focus of these works was to improve the dependence of $\eps$.
As they are not epoch-based, the notion of clearing an epoch does not apply.

 The \kECSS~problem and its special cases have been studied extensively.
 When $k=1$, it is the minimum spanning tree problem.
 When $k \geq 2$, it is APX-hard~\cite{fernandes1998better} already on bounded-degree graphs~\cite{csaba2002approximability}, or when the edge costs are binary~\cite{Pritchard10}.
 Frederickson and Jaja~\cite{FredericksonJ81} introduced the first $3$-approximation algorithm for \kECSS. This was later improved to a $2$-approximation algorithm by Kuller and Vishkin~\cite{KhullerV94}, which runs in $\tilde O(mnk)$ time. In more recent work, Chalermsook et al.~\cite{ChalermsookHNSS22} proposed a $(2+\varepsilon)$-approximation algorithm running in $\tilde O((m+k^2n^{1.5})/\varepsilon^2)$ time. While a factor $2$ approximation for the general \kECSS~problem has not been surpassed in more than $30$ years, there are several special cases in which improved approximation ratios have been obtained (see e.g.~\cite{grandoni2018improved,fiorini2018approximating,adjiashvili2018beating}). One such case is the unit-cost \kECSS~(where $c_e=1$ for all $e \in E$), which admits a $(1 + O(1/k))$ approximation algorithm \cite{gabow2009approximating,laekhanukit2012rounding}. Additionally, there has been significant progress on \kECSS~in specific graph classes. In Euclidean graphs, Czumaj and Lingas developed a nearly linear time approximation scheme for fixed values of $k$ \cite{czumaj2000fast,czumaj1999approximability}. The problem is also solvable in nearly linear time when both $k$ and the treewidth are fixed \cite{berger2007minimum,chalermsook2018survivable}. For planar graphs, 2-\textsf{ECSS}, 2-\textsf{ECSM} and 3-\textsf{ECSM} have polynomial-time approximation schemes~\cite{czumaj2004approximation,borradaile2014polynomial}.

For the \kECSM~problem, Frederickson and Jaja \cite{FredericksonJ81,FredericksonJ82} gave a 3/2-approximation for even $k$, and a $(3/2+O(1/k))$-approximation for odd $k$.
This was improved by Karlin, Klein, Oveis Gharan and Zhang to $1+O(1/\sqrt{k})$ \cite{KarlinK0Z22}.
More recently, Hershkowitz, Klein, and Zenklusen gave a $(1+O(1/k))$-approximation, and showed that this is tight up to constant factors \cite{HershkowitzKZ24}.

\subsection{Paper Organization}
In Section \ref{sec:prelim}, we introduce notation and provide preliminaries.
In Section \ref{sec:parallel MWU}, we formally define and analyze the parallel
MWU framework with core-sequences (Theorem \ref{thm:mwu tradeoff}).
In Section \ref{sec:cut_covering_LP}, we first give an overview on finding good core-sequences for the Cut Covering LP.
Then, we develop the corresponding MWU algorithm (\Cref{thm:parallel_TSP}) and the data structures it uses.
In Section \ref{sec:kECSS_LP}, we extend our techniques to the \kECSS~LP
(Theorem \ref{thm:parallel_kECSS}).
Missing proofs can be found in \Cref{sec:missing_proofs}.

\section{Preliminaries} \label{sec:prelim}

\paragraph{Model of Computation.} We use the standard \emph{work-depth model}~\cite{ShiloachV82a, Blelloch96}. The \emph{work} of an algorithm is the total number of operations over all processors, similar to the time complexity in the sequential RAM model. The (parallel) depth is the length of the longest sequence of dependent operations. We assume concurrent read and write operations. It is well-known that a parallel algorithm with work $W$ and depth $D$ implies a parallel algorithm that runs in $O(W/p+D)$ time when there are $p$ processors.

\paragraph{Graphs and Cuts.}
Let $G = (V,E)$ be an undirected graph.
By default, we denote $n$ as the number of vertices and $m$ as the number of edges in $G$.
A \emph{cut} in $G$ is $\delta_G(S)$ for some $\emptyset \subseteq S \subseteq V$, where $\delta_G(S)$ denotes the set of edges in $E$ having exactly one endpoint in $S$.
When the graph is clear from context, we will drop the subscript and write $\delta(S)$.
Given nonnegative edge weights $w\in \R^m_{\geq 0}$ and a subset $F\subseteq E$, we write $w(F) := \sum_{e\in F}w(e)$.

It is well-known that the cut function is \emph{submodular} and \emph{posi-modular} (see e.g. \cite{NagamochiI00}).

\begin{proposition} \label{lem:posimod cut}
Let $w\in \R^m_{\geq 0}$ be nonnegative edge weights.
For every pair of subsets $X,Y\subseteq V$, the following are true
\begin{itemize}
  \item Submodularity: $w(\delta(X)) + w(\delta(Y)) \geq w(\delta(X \cap Y)) + w(\delta(X \cup Y))$
  \item Posi-modularity: $w(\delta(X))+ w(\delta(Y)) \geq w(\delta(X \setminus Y)) + w(\delta(Y \setminus X))$.
\end{itemize}
\end{proposition}

Let $\OPT_w:= \min_{\emptyset \subsetneq S\subsetneq V} w(\delta(S))$ denote the value of a minimum cut in $G$.
For $\alpha\geq 1$, an \emph{$\alpha$-minimum cut} is a cut $\delta(S)$ which satisfies $w(\delta(S))\leq \alpha \cdot \OPT_w$.

\paragraph{Minimum Cuts via Tree-Packing.}
Karger~\cite{Karger00} provided the first nearly linear time randomized sequential minimum cut algorithm.
It is based on
the relationship between a maximum packing of spanning trees and a minimum cut given by Nash--Williams~\cite{NW61}.

\begin{definition}
Let $T$ be a spanning tree of $G$.
We say that a cut $C$ in $G$ \emph{$k$-respects} $T$ if $|C \cap E(T)| = k$.
We also say that the cut $1$-or-$2$-\textit{respects} $T$ if $|C \cap E(T)| \leq 2$.
\end{definition}

Karger observed that in an approximately maximum packing $\mathcal{T}'$ of spanning trees, every approximately minimum cut 1-or-2-respects some tree $T\in \mathcal{T}'$.
Hence, if we have $\mathcal{T}'$, then finding a minimum cut in $G$ reduces to finding a minimum cut among all 1-or-2 respecting cuts of $T$ for all $T\in \mathcal{T}'$.

In order to achieve nearly linear time, he did not compute $\mathcal{T}'$.
Instead, he gave a randomized algorithm for computing a tree packing such that this property holds with high probability.

\begin{theorem} [Fast Parallel Tree-Packing~\cite{Karger00}] \label{thm:tree packing}
    Given a graph $G$ with edge weights $w\in \mathbb{R}^m_{\geq 0}$, there is a randomized algorithm that outputs a set $\mathcal{T}$ of $O(\log n)$ spanning trees such that with high probability, every $(1+\varepsilon)$-minimum cut $1$-or-$2$-respects some tree in $\mathcal{T}$ for $\varepsilon<0.5$. The algorithm runs in $\tilde O(m)$ work and $\tilde O(1)$ depth.
\end{theorem}

Since there are $O(\log n)$ spanning trees in $\mathcal{T}$, focusing on one tree at a time suffices.
The problem of finding a minimum 1-or-2-respecting cut of a tree was first solved using dynamic programming and graph data structures~\cite{Karger00}. The algorithm has been subsequently simplified using various techniques such as more advanced data structures including top-trees~\cite{BhardwajLS20}, and exploiting structural properties of the cut function~\cite{conf/stoc/MukhopadhyayN20,conf/sosa/GawrychowskiMW21}.

We will use \Cref{thm:tree packing} in our algorithm.
Given a spanning tree, we identify the 1-respecting cut corresponding to a tree edge and the 2-respecting cut corresponding to a pair of tree edges as follows.

\begin{definition}
Let $T$ be a spanning tree of $G$.
For a tree edge $e\in E(T)$, we define $\shore_T(\{e\})$ as the vertex set of one of the two components in $T\setminus e$ (breaking ties arbitrarily).
We denote $\cut_T(\{e\})$
as the set of edges in $G$ with exactly one endpoint in $\shore_T(\{e\})$.

For a pair of distinct tree edges $e,f\in E(T)$, let $X,Y,Z$ be the three components in $T\setminus \{e,f\}$, where $e$ connects $X$ and $Y$, and $f$ connects $Y$ and $Z$.
We define $\shore_T(\{e,f\}):= V(Y)$ as the vertex set of the \emph{middle} component $Y$.
We denote $\cut_T(\{e,f\})$
as the set of edges in $G$ with exactly one endpoint in $\shore_T(\{e,f\})$.
\end{definition}

Conversely, the set of all 1-or-2-respecting cuts of a tree can be represented succinctly by referring to a tree edge or a pair of tree edges.

\begin{definition} [1-or-2-Respecting Cuts of a Tree]
    Let $T$ be a spanning tree of $G$.
    We denote $\mathcal{C}_T := \{ F \subseteq E(T) \colon 1\leq |F| \leq 2\}$ as the set of $1$-or-$2$-respecting cuts of $T$.
    Let $E^1(T) := \{ \{e\} \colon e \in E(T)\}$ be the set of 1-respecting cuts in $T$, and let $E^2(T) := \{ F \subseteq E(T) \colon |F| = 2\}$ be the set of 2-respecting cuts in $T$.
    We denote $\mincut_w(T) := \min_{s \in \mathcal{C}_T} w(\cut_T(s))$ as the value of a minimum 1-or-2-respecting cut of $T$.
\end{definition}

\section{Parallel MWU Framework} \label{sec:parallel MWU}

In this section, we develop a MWU framework which is compatible with core-sequences.
We start by giving a variant of Young's parallel MWU method~\cite{conf/focs/Young01,arxiv/Young14} for packing/covering LPs.
Then, we will modify it to work with core-sequences.

\subsection{Parallel MWU with On-The-Fly Initialization}
\label{sec:parallel_mwu_young}

Given a nonnegative matrix $A\in \R^{m\times N}_{\geq 0}$, Young's parallel MWU method~\cite{conf/focs/Young01,arxiv/Young14} returns a $(1-\varepsilon)$-approximate solution to the packing LP~\eqref{lp:pack} and a $(1+\varepsilon)$-approximate solution to the covering LP~\eqref{lp:cover} in $O(\log^2 (m) \log(N\log(m)/\varepsilon)/\varepsilon^4)$ iterations.
Unfortunately, it is not suitable for implicitly defined LPs because $N\gg m$.
For example, $N = \Omega(2^m)$ for the Cut Covering LP \eqref{lp:cut_cover} when $m = O(n)$, so the iteration bound becomes linear in $m$.

The dependence on $N$ in the iteration bound is due to how the packing variables are initialized.
They are set as $x_j := \min_{i\in [m]}1/(NA_{ij})$ for all $j\in [N]$ at the beginning of the algorithm.
If we instead initialize them `on the fly', then this dependence can be improved.
In particular, we only initialize the `relevant' variables at the start of every epoch.
More formally, let $t$ be the first iteration of an epoch.
Recall that $x^{(t)}$ and $w^{(t)}$ are the packing variables and weights at the start of iteration $t$ respectively, while $\lambda^{(t)}$ is the lower bound on the minimum weight of a column.
We only initialize the coordinates in the following set
\begin{equation}
    B_0^{(t)} := \{j\in [N]\setminus \supp(x^{(t)}): (A^\top w^{(t)})_j < (1+\varepsilon)\lambda^{(t)}\}.
\end{equation}
Note that $B_0^{(t)}\subseteq B^{(t)}$.
In particular, they are set as $x^{(t)}_j := \min_{i\in [m]}1/(|B^{(t)}_0|A_{ij})$ for all $j\in B^{(t)}_0$.
See \Cref{alg:mwu_parallel} for a pseudocode.

With this change, the iteration bound improves to $O(\log^2 (m) \log(\max_t|B^{(t)}|\log(m)/\varepsilon)/\varepsilon^4)$.
Of course, $B^{(t)}$ could still be very large.
However, this is already useful for the Cut Covering LP \eqref{lp:cut_cover}.
For this LP, $B^{(t)}$ corresponds to a subset of $(1+\varepsilon)$-approximate minimum cuts with respect to the edge weights $w^{(t)}$.
It is a well-known fact~\cite{HenzingerW96} that as long as $\varepsilon<0.5$, the number of $(1+\varepsilon)$-approximate minimum cuts in an edge-weighted graph $G=(V,E)$ is $O(|V|^2)$.
Hence, $|B^{(t)}| = O(|V|^2) \ll N$.

\begin{algorithm}[htb!]
  \caption{\textsc{Parallel MWU}}
  \label{alg:mwu_parallel}
  \SetKwInOut{Input}{Input}
  \SetKwInOut{Output}{Output}
  \SetKwComment{Comment}{$\triangleright$\ }{}
  \SetKw{And}{\textbf{and}}
  \SetKw{Or}{\textbf{or}}
  \Input{Nonnegative matrix $A\in \R^{m\times N}_{\geq 0}$, accuracy parameter $\varepsilon>0$}
  \Output{A $(1-O(\varepsilon))$-optimal solution to \eqref{lp:pack}, and a $(1+O(\varepsilon))$-optimal solution to \eqref{lp:cover}}
  $\eta \gets \ln(m)/\varepsilon$\;
  $x \gets \0_N$, $w\gets \1_m$, $\lambda \gets \min_{j\in [N]} (A^\top w)_j$, $y\gets w/\lambda$\;
  \While{$\|Ax\|_\infty < \eta$}{
    \uIf{$\min_{j\in [N]} (A^\top w)_j < (1+\varepsilon)\lambda$}{
        $B \gets \{j\in [N]:(A^\top w)_j < (1+\varepsilon)\lambda\}$ \;
        $B_0 \gets \{j\in B: x_j = 0\}$\;
        \uIf(\Comment*[f]{Only happens at the start of an epoch}){$B_0\neq \emptyset$}{
          Set $g_j \gets \varepsilon/(|B_0| \max_{i\in [m]} A_{i,j})$ for all $j\in B_0$, and $g_j \gets 0$ for all $j\notin B_0$ \;
        }
        \Else{
          Set $g_j \gets \delta x_j$ for all $j\in B$, and $g_j \gets 0$ for all $j\notin B$, where $\delta$ is chosen such that $\|Ag\|_\infty = \varepsilon$ \;
        }
        $x \gets x + g$\;
        $w \gets w\circ (\1+Ag)$ \label{line:weight_update}\;}
    \Else{
        $\lambda \gets (1+\varepsilon)\lambda$ \Comment*{New epoch}
        \If{$\ang{\1,w}/\lambda< \ang{\1,y}$}{
            $y\gets w/\lambda$\;
        }
    }
    }
  \Return $(x/\|Ax\|_\infty, y)$ \;
\end{algorithm}

For every iteration $t\geq 0$, let $x^{(t)}, w^{(t)}, \lambda^{(t)}, y^{(t)}$ denote the corresponding values at the start of the iteration, and let $B^{(t)}, B^{(t)}_0, g^{(t)}$ denote the corresponding values computed during the iteration.
We remark that on \Cref{line:weight_update}, we update the weights as $w^{(t+1)}_i\gets w^{(t)}_i(1+(Ag^{(t)})_i)$ like in \cite{journals/siamcomp/GargK07,journals/siamdm/Fleischer00}, instead of $w^{(t+1)}_i\gets w^{(t)}_ie^{(Ag^{(t)})_i}$ in \cite{conf/focs/Young01}.
This allows us to provide a simple and self-contained correctness proof, without needing to go through the log-sum-exp function.

First, we upper bound the total weight in every iteration.

\begin{restatable}{lemma}{ubound}
\label{lem:weight-upper-bound}
For every iteration $t\geq 0$, we have
\[\ang{\1,w^{(t)}} \leq m \exp \left((1+\varepsilon) \sum_{s=0}^{t-1} \frac{\ang{\1,g^{(s)}}}{\ang{\1,w^{(s)}}/\lambda^{(s)}}\right)\]
\end{restatable}

Next, we lower bound each individual weight by the congestion of the corresponding row.

\begin{restatable}{lemma}{lbound}
\label{lem:weight-lower-bound}
For every iteration $t\geq 0$ and row $i\in [m]$, we have $w^{(t)}_i \geq e^{(1-\varepsilon)(Ax^{(t)})_i}.$
\end{restatable}

Using these two lemmas, we prove the correctness of the algorithm.

\begin{restatable}{theorem}{correctness}
\label{thm:mwu_correctness}
\Cref{alg:mwu_parallel} returns a $(1-O(\varepsilon))$-optimal solution to \eqref{lp:pack}, and a $(1+O(\varepsilon))$-optimal solution to \eqref{lp:cover}.
\end{restatable}

It is left to upper bound the number of iterations.
The proof is similar to \cite{conf/focs/Young01,arxiv/Young14}.
It proceeds by bounding the number of epochs, followed by the number of iterations per epoch.
Recall that an \emph{epoch} is a maximal sequence of consecutive iterations with the same value of $\lambda^{(t)}$.

\begin{restatable}{lemma}{epochs}
\label{lem:epochs}
There are at most $O(\log m/\varepsilon^2)$ epochs in Algorithm~\ref{alg:mwu_parallel}.
\end{restatable}

\begin{restatable}{theorem}{iterations}
\label{thm:mwu_runtime}
The number of iterations in \Cref{alg:mwu_parallel} is at most
\[O\left(\frac{\log^2(m)\log(\eta\max_t|B^{(t)}|/\varepsilon)}{\varepsilon^4}\right).\]
\end{restatable}

If we apply \Cref{alg:mwu_parallel} to the Cut Covering LP of an edge-weighted graph with $n$ vertices and $m$ edges, it converges in $O(\log^2(n)\log(\eta n/\varepsilon)/\varepsilon^4)$ iterations because $|B^{(t)}|\leq O(n^2)$.
However, we may need to update $|B^{(t)}| = \Theta(n^2)$ coordinates in an iteration, so it does not lead to a nearly linear work algorithm.
In the next subsection, we remedy this issue by modifying \Cref{alg:mwu_parallel} to work with core-sequences.

\subsection{MWU with Core-Sequences}
In this subsection, we state a general form of \Cref{alg:mwu_parallel} (\Cref{alg:mwu_parallel_general}).
Fix an epoch.
Let $t_0$ be the first iteration of this epoch, and let $\lambda$ be the lower bound used in this epoch.
For every iteration $t\geq t_0$, instead of updating all the coordinates in $B^{(t)}$, we allow the algorithm to \emph{focus} on a \emph{fixed} subset $\tilde B_1\subseteq B^{(t_0)}$.
In particular, the algorithm only updates the coordinates in $\tilde B_1\cap B^{(t)}$ for $t\geq t_0$ until this set becomes empty.
\Cref{alg:mwu_focus} implements this {\sc Focus} procedure.

Let $t_1$ be the first iteration when $\tilde B_1\cap B^{(t_1)} = \emptyset$.
If $\min_{j\in [n]}(A^\top w^{(t_1)})_j \geq (1+\varepsilon)\lambda$, then the epoch is \emph{cleared}.
Otherwise, the algorithm selects a new fixed subset $\tilde B_2\subseteq B^{(t_1)}$ to focus on.
Note that $B^{(t_1)}\subseteq B^{(t_0)}\setminus \tilde B_1$ because $w^{(t)}$ is nondecreasing while $\lambda$ remains unchanged.
This process is repeated until the epoch is cleared.

Another change that we make in \Cref{alg:mwu_parallel_general} is to not maintain the packing variables $x$ explicitly.
This will also help us in implementing every iteration cheaply for \mTSP~and \kECSS.
For these applications, we are only interested in a solution to the covering LP~\eqref{lp:cover}.
However, we still need to keep track of the \emph{congestion} vector $\congestion:= Ax$, as it forms our termination criterion.
Whenever \textsc{Focus} (\Cref{alg:mwu_focus}) is invoked, instead of taking $x$ as input, it takes $\congestion$.
It also initializes a temporary vector $\tilde x$ of packing variables for local use.
Upon termination, $\tilde x$ is forgotten and not passed to the main algorithm (\Cref{alg:mwu_parallel_general}).
Instead, the updated congestion vector $\congestion$ is passed.

\begin{algorithm}[htb!]
  \caption{\textsc{Focus}$_{A,\lambda,\varepsilon}(\tilde B,w,\congestion)$}
  \label{alg:mwu_focus}
  \SetKwInOut{Input}{Input}
  \SetKwInOut{Output}{Output}
  \SetKwComment{Comment}{$\triangleright$\ }{}
  \SetKw{And}{\textbf{and}}
  \SetKw{Or}{\textbf{or}}
  \Input{Nonnegative matrix $A\in \R^{m\times N}_{\geq 0}$, lower bound $\lambda \geq 0$, accuracy parameter $\varepsilon>0$, subset $\tilde B \subseteq \{ j \in [N] \colon (A^\top w)_j < (1+\varepsilon)\lambda\}$, weights $w\in \R^m_{\geq 0}$, congestion $\congestion \in \R^m_{\geq 0}$}
  \Output{Weights $w'\in \R^m_{\geq 0}$ and congestion $\congestion' \in \R^m_{\geq 0}$}
    $\eta \gets \ln(m)/\varepsilon$\;
    $x \gets \0_N$\;
 \While{$\|\congestion\|_\infty < \eta$ \textbf{and} $\tilde B \neq \emptyset$}{
   $\tilde B_0 \gets \{j\in \tilde B: x_j =0\}$\;
    \uIf(\Comment*[f]{Only happens in the first iteration}){$\tilde B_0\neq \emptyset$}{
      Set $g_j \gets \varepsilon/(|\tilde B_0| \max_{i\in [m]} A_{i,j})$ for all $j\in \tilde B_0$, and $g_j \gets 0$ for all $j\notin \tilde B_0$ \;
    }
    \Else{
      Set $g_j \gets \delta x_j$ for all $j\in \tilde B$, and $g_j \gets 0$ for all $j\notin \tilde B$, where $\delta$ is chosen such that $\|Ag\|_\infty = \varepsilon$ \;
    }
    $x \gets x + g$\;
    $w \gets w\circ (\1+Ag)$\;
    $\congestion \gets \congestion + Ag$\;
    $\tilde B \gets \{ j \in \tilde B  \colon (A^\top w)_j < (1+\varepsilon)\lambda\}$\;
    }
    \Return $(w,\congestion)$\;

\end{algorithm}

\begin{algorithm}[htb!]
  \caption{\textsc{MWU with Focus}}
  \label{alg:mwu_parallel_general}
  \SetKwInOut{Input}{Input}
  \SetKwInOut{Output}{Output}
  \SetKwComment{Comment}{$\triangleright$\ }{}
  \SetKw{And}{\textbf{and}}
  \SetKw{Or}{\textbf{or}}
  \Input{Nonnegative matrix $A\in \R^{m\times N}_{\geq 0}$, accuracy parameter $\varepsilon>0$}
  \Output{A $(1+O(\varepsilon))$-optimal solution to \eqref{lp:cover}}
  $\eta \gets \ln(m)/\varepsilon$\;
  $w \gets \1_m$, $\congestion \gets \0_m$, $\lambda \gets \min_{j\in [N]} (A^\top w)_j$, $y\gets w/\lambda$\;
  \While{$\|\congestion\|_\infty < \eta$}{
    \uIf{$\min_{j\in [N]} (A^\top w)_j < (1+\varepsilon)\lambda$}{
        Select a subset $\tilde B \subseteq \{j\in [N]:(A^\top w)_j < (1+\varepsilon)\lambda\}$\;
        $(w,\congestion) \gets \textsc{Focus}_{A,\lambda,\varepsilon}(\tilde B,w,\congestion)$\;
    }
    \Else{
        $\lambda \gets (1+\varepsilon)\lambda$ \Comment*{New epoch}
        \If{$\ang{\1,w}/\lambda < \ang{\1,y}$}{
            $y\gets w/\lambda$\;
        }
    }
    }
  \Return $y$ \;
\end{algorithm}

It is easy to check that \Cref{lem:weight-upper-bound}, \Cref{lem:weight-lower-bound} and \Cref{thm:mwu_correctness} apply to \Cref{alg:mwu_parallel_general}.
\Cref{lem:epochs} also applies to \Cref{alg:mwu_parallel_general} because the number of epochs remain the same.

The next lemma bounds the number of iterations carried out by \Cref{alg:mwu_focus}.
Its proof is identical to the proof of \Cref{thm:mwu_runtime}.

\begin{lemma}\label{lem:iters_focus}
\Cref{alg:mwu_focus} terminates in $O(\log(m)\log(\eta|\tilde B|/\varepsilon)/\varepsilon^2)$ iterations.
\end{lemma}

Our goal is to apply \Cref{alg:mwu_focus} on a sequence $\tilde{\mathcal{B}} = (\tilde B_1, \tilde B_2, \dots, \tilde B_\ell)$ of sets such that the epoch is cleared.
Such a sequence is called a \emph{core-sequence} of the epoch (\Cref{def:core-sequence}).
We say that \Cref{alg:mwu_parallel_general} \emph{follows} $\tilde{\mathcal B}$ during this epoch.
By \Cref{lem:iters_focus}, the total number of iterations depends not only on the length of the sequence, but also the size of each constituent set.
The following theorem is an immediately consequence of \Cref{lem:epochs} and \Cref{lem:iters_focus}.

\begin{theorem}\label{thm:mwu_focus_runtime}
If \Cref{alg:mwu_parallel_general} always follows a core-sequence of length at most $\ell$ and with sets of size at most $n$, then the total number of iterations is
\[O\left(\frac{\ell\log^2(m)\log(\eta n/\varepsilon)}{\varepsilon^4}\right).\]
\end{theorem}

\Cref{alg:mwu_parallel_general} is a generalization of \Cref{alg:mwu_parallel}.
To see this, let $t$ be the first iteration of an epoch.
If it follows the core sequence $\tilde{\mathcal B} = (B^{(t)})$, then it specializes to \Cref{alg:mwu_parallel}.
On the other hand, if it follows a core sequence $\tilde{\mathcal B} = (\tilde B_1, \tilde B_2, \dots, \tilde B_\ell)$ where $|\tilde B_i| = 1$ for all $i\in [\ell]$, then it becomes an instantiation of Fleischer's sequential MWU algorithm \cite{journals/siamdm/Fleischer00}.

Consider the Cut Covering LP on a graph with $n$ vertices and $m$ edges.
Since there are $O(n^2)$ $(1+\varepsilon)$-approximate minimum cuts for $\varepsilon<0.5$, we can trivially construct a core-sequence of length $O(n)$, in which every set has size $O(n)$.
This certainly reduces the number of coordinates that need to be updated per iteration, but the number of iterations becomes $\tilde O(n/\varepsilon^4)$  by \Cref{thm:mwu_focus_runtime}.
In the next section, we show that there exists a shorter core-sequence with length $\poly\log(n)$, in which every set has size $\tilde O(n)$.
This enables us to keep the number of iterations polylogarithmic.

\section{Approximating the Cut Covering LP}
\label{sec:cut_covering_LP}

In this section, we apply \Cref{alg:mwu_parallel_general} to the Cut Covering LP \eqref{lp:cut_cover}, and give a nearly linear work and polylogarithmic depth implementation.
By \Cref{lem:epochs}, it suffices to implement every epoch of \Cref{alg:mwu_parallel_general} in nearly linear work and polylogarithmic depth.
The initial value of $\lambda$ can be computed using the nearly linear work parallel algorithm of \cite{conf/spaa/GeissmannG18} for minimum cut.

\paragraph{Epoch Algorithm.}  In every epoch, the input is a graph $G$ with edge weights $w\in \R^m_{> 0}$, edge congestion $\congestion\in \R^m_{\geq 0}$ and scalars $0<\lambda\leq \OPT_w,0<\varepsilon<0.5$.
We iterate between computing a subset $\tilde B$ of cuts with weight less than $(1+\varepsilon)\lambda$ and applying \Cref{alg:mwu_focus} on $\tilde B$, until we obtain a pair $(w,\congestion)$ which satisfies $\OPT_w \geq (1+\varepsilon)\lambda$ or $\|\congestion\|_\infty \geq \eta$.
The challenge is to compute a good sequence of $\tilde B$'s so that the total work is $\tilde O(m/\varepsilon^2)$ while the parallel depth is $\tilde O(1/\varepsilon^2)$.
We now formulate the algorithm for an epoch.

\begin{definition}[\pminsubset operation] \label{def:extract_focus op}
    Given a graph $G$ with edge weights $w\in \R^m_{\geq 0}$, the \pminsubset operation computes a subset $\tilde B$ of cuts whose weights are less than $(1+\varepsilon)\lambda$, and applies $\focus$ (\Cref{alg:mwu_focus}) on $\tilde B$.
    We call $\tilde B$ a \emph{focus set}.
\end{definition}

\begin{theorem} [Epoch Algorithm]\label{thm:epoch algorithm}
    Let $G$ be a graph with edge weights $w\in \R^m_{>0}$ and edge congestion $\congestion\in \R^m_{\geq 0}$.
    Given scalars $0<\lambda \leq \OPT_w$ and $0< \varepsilon < 0.5$, there is a randomized algorithm that iteratively applies $\pminsubset$ so that $\OPT_{w} \geq (1+\varepsilon)\lambda$ or $\|\congestion\|_\infty\geq \ln(m)/\varepsilon$
     with high probability.
    The algorithm runs in $\tilde O(m/\varepsilon^2)$ work and $\tilde O(1/\varepsilon^2)$ depth.
\end{theorem}

Let $\tilde{\mathcal{B}} = (\tilde B_1, \tilde B_2, \dots, \tilde B_\ell)$ be the sequence of cuts found by \pminsubset in \Cref{thm:epoch algorithm}.
If $\OPT_{w}\geq (1+\varepsilon)\lambda$, then $\tilde{\mathcal{B}}$ is a core-sequence of the epoch.
To prove \Cref{thm:epoch algorithm}, we implement a sequence of \pminsubset operations such that it runs in $\tilde O(\sum_{i=1}^\ell(m + |\tilde B_i|)\log(|\tilde B_i|)/\varepsilon^2)$ work and $\tilde O(\sum_{i=1}^\ell \log(|\tilde B_i|)/\varepsilon^2)$ depth.
Our key technical contribution is to find a core-sequence such that $\ell = O(\log^2 n)$ and $|\tilde B_i| = \tilde O(n)$ for all $i\in [\ell]$.
This then gives the desired work and depth bounds.

\subsection{Finding Good Core-Sequences}
\label{sec:tech_overview}

In this section, we give an overview on finding \emph{good} core-sequences, i.e., having short length and small constituent sets, for the Cut Covering LP.

Like the nearly linear time sequential MWU algorithm~\cite{CQ17}, we also rely on Karger's tree-packing theorem \cite{Karger00}.
Let $\mathcal{T}$ be the set of spanning trees given by \Cref{thm:tree packing}.
Let $\lambda$ be a lower bound on the value of a minimum cut and let $\varepsilon>0$.
A tree $T\in \mathcal{T}$ is said to be \emph{cleared} if every cut which 1-or-2-respects $T$ has weight at least $(1+\varepsilon)\lambda$.
By \Cref{thm:tree packing}, if every tree in $\mathcal{T}$ is cleared, then the epoch is cleared with high probability.
Since $|\mathcal{T}| = O(\log n)$, it suffices to find a \emph{core-sequence} for each tree in $T\in \mathcal{T}$, i.e., a sequence $\tilde{\mathcal{B}} = (\tilde B_1, \tilde B_2, \dots, \tilde B_\ell)$ such that applying the update \eqref{eq:core_sequence_update} on $\tilde{\mathcal{B}}$ clears $T$.
Even though each $\tilde B_i$ can be chosen as any subset of cuts with weight less than $(1+\varepsilon)\lambda$, we will only choose from those which 1-or-2-respect $T$.

Fix a tree $T\in \mathcal{T}$.
In the rest of this section, we outline how to find a good core-sequence for $T$ when $T$ is a \emph{path}. This case already captures the main idea of the algorithm, since general trees can be reduced to paths using classical path-decomposition techniques~\cite{conf/stoc/MukhopadhyayN20}.
When $T$ is a path, we denote its edges as $e_1 < e_2 < \dots < e_{n-1}$, where $e_i$ and $e_{i+1}$ share a common vertex for all $1\leq i\leq n-2$.
Then, every cut in the graph which 2-respects $T$ can be represented as an open interval $(e_i,e_j)$.
The cut \textit{induced} by an open interval $(e_i,e_j)$ where $i<j$ corresponds to the middle subpath $P$ after deleting $e_i, e_j$ from $T$, i.e., $\delta_G(V(P))$.
The \textit{weight} of the interval $w_G(e_i,e_j)$ is the weight of the cut induced by $(e_i,e_j)$, i.e., $w_G(e_i,e_j) := w(\delta_G(V(P))$.
We remark that empty intervals $(e_i,e_i)$ are not considered as they do not correspond to cuts.
For a subpath $P$ of $T$, let $S_P$ be the set of open intervals in $P$ with weight less than $(1+\varepsilon)\lambda$, i.e.,
\[S_P := \left\{(e_i,e_j) \colon e_i,e_j \in E(P) \text{ and } w_G(e_i,e_j) < (1+\varepsilon)\lambda\right\}.\]

For the sake of simplicity, let us assume that every cut which 1-respects $T$ has weight at least $(1+\varepsilon)\lambda$.
We will choose each $\tilde B_i$ as a subset of $S_T$.
Observe that $(S_T)$ is trivially a core-sequence.
However, there could be as many as ${n-1 \choose 2} = \Theta (n^2)$ intervals in $S_T$.
So, an $\tilde O(m)$-work parallel algorithm cannot afford to update (or even enumerate) all of them.

The key idea for constructing a good core-sequence is to exploit the \emph{correlation} between the intervals.
In particular, we choose each $\tilde B_i$ in such a way that increasing their weights also increases the weight of \emph{many other} intervals.
This is formalized in the following divide-and-conquer approach.
Here, we only focus on proving the \emph{existence} of a core-sequence of length $\tilde O(1)$ with $\tilde O(n)$ cuts in total. Efficiently computing the core-sequence requires dynamic data structures, which we present in later sections.

\paragraph{Path Algorithm.}  For convenience, let us assume that $n = 2^k+1$ for some integer $k\geq 1$.

\begin{enumerate}
    \item In iteration 0, we clear the set of 1-respecting cuts of $T$:
    \begin{enumerate}
        \item We set $\tilde B_0 := \{e\in E(T):w(\cut_T(\{e\}))<(1+\varepsilon)\lambda\}$.
        \item We apply \textsc{Focus} on $\tilde B_0$. When it terminates, the weight of every 1-respecting cut of $T$ is at least $(1+\varepsilon)\lambda$.
    \end{enumerate}

    \item
    In iteration $1\leq i\leq k$, we form paths of length $2^i$ by concatenating pairs of consecutive paths of length $2^{i-1}$ obtained in iteration $i-1$. Let $\mathcal{P}_i$ be the set of these $2^{k-i}$ paths.
     \begin{enumerate}
        \item We set $\tilde B_i := \bigcup_{P \in \mathcal{P}_i} S_P$.
        \item We apply \textsc{Focus} on $\tilde B_i$. When it terminates, the weight of every interval in $\tilde B_i$ is at least $(1+\varepsilon)\lambda$.
    \end{enumerate}
\end{enumerate}

\medskip

At first sight, it is unclear whether this algorithm actually constructs a good core-sequence because the size of $\tilde B_i$ could be $\Theta(n^2)$ for some $i\in [k]$. Surprisingly, we prove that $|\tilde B_i| \leq n$ for all $i \in [k]$. This immediately implies the existence of a core-sequence of $T$ of length $O(\log n)$ with total size $ O(n\log n)$.

We say that an open interval $(e_i,e_j)$ \emph{contains} a vertex $v$ if $v\in V(P)$, where $P$ is the middle subpath after deleting $e_i,e_j$ from $T$.
It is denoted as $v \in (e_i,e_j)$.
The key insight is the following lemma.

\begin{lemma} \label{claim: small SP}
Let $P$ be a subpath of $T$. If there exists a vertex $r \in V(P)$ such that $w_G(e_i,e_j) \geq (1+\varepsilon)\lambda$ for every open interval $(e_i,e_j)$ in $P$ which does not contain $r$, then $|S_P| \leq |E(P)|$.
\end{lemma}

For every $1\leq i\leq k$, we can bound the set $\tilde B_i$ by
\[|\tilde B_i| = \sum_{P \in \mathcal{P}_i} |S_P| \leq  \sum_{P \in \mathcal{P}_i} |E(P)| \leq n.\]
The first inequality trivially holds when $i = 1$.
For $i>1$, recall that every path $P\in \mathcal{P}_i$ is obtained by concatenating two subpaths $P_L$ and $P_R$ from iteration $i-1$.
Let $r$ be the common vertex of $P_L$ and $P_R$. By induction on the previous iteration, every open interval $(e_i,e_j)$ in $P_L$ and $P_R$ has weight at least $(1+\varepsilon)\lambda$.
Therefore, the condition of \Cref{claim: small SP} applies to $P$ using node $r$.

The proof of \Cref{claim: small SP} is very simple when expressed in the language of forbidden matrix theory.

\paragraph{Forbidden Matrix Theory.} Let $A, B$ be binary matrices. We say that $A$ \textit{contains} $B$ if $B$ can be obtained from $A$ by deleting rows, deleting columns, and turning 1's to 0's. Otherwise, we say that $A$ \textit{avoids} $B$. We denote $|A|$ as the number of 1's in a binary matrix $A$. Given a binary matrix $B$, we define $Ex(B,n,m)$ as \textit{the maximum number} of 1's in an $n$-by-$m$ matrix which avoids $B$. We refer to \cite{FurediH92}
for background on the topic.

The following binary matrices are relevant to us:
    \[ Z_2:=\kbordermatrix{
          & &  \\
         && 1 \\
        & 1 &} \qquad Z_3 :=\kbordermatrix{
         & & &  \\
         & & & 1 \\
        & &1& \\
       & 1& &}. \]
For the sake of completeness, we prove extremal bounds for binary matrices that avoid these basic patterns.

\begin{claim} \label{lem:ex Z}
    $Ex(Z_2,n,m) \leq m+n-1$ and  $Ex(Z_3,n,m) = O(m+n)$.
\end{claim}
\begin{proof}
Let $Z_k$ be the anti-diagonal pattern of length $k$. We prove the stronger claim $Ex(Z_k,n,m) \leq (k-1)(m+n-1)$. For any matrix $A$ avoiding $Z_k$, every anti-diagonal line of $A$ contains at most $k-1$ nonzero entries. As the entries of $A$ can be partitioned into $n+m-1$ anti-diagonal lines, the claim follows.
\end{proof}

We now reformulate the setting of \Cref{claim: small SP} using forbidden matrices.
Recall that $P$ is a subpath of $T$.
Given a vertex $r\in V(P)$ that is not a leaf of $P$, we define the matrix $M_r$ as $M_r(e_i,e_j) := w_G(e_i,e_j)$ for every open interval $(e_i,e_j)\subseteq P$ containing $r$.
Alternatively, we can think of the rows of $M_r$ as being indexed by the subpath $P_L$ of $P$ to the left of $r$, and the columns of $M_r$ as being indexed by the subpath $P_R$ of $P$ to the right of $r$.
In particular, the $i$-th row corresponds to $i$-th edge in $P_L$ (counting from $r$), and the $j$-th column corresponds to $j$-th edge in $P_R$ (counting from $r$).
We also define the binary matrix $A_r$
 as $A_r(e_i,e_j) := 1$ if and only if $M_r(e_i,e_j) < (1+\varepsilon)\lambda$.

\begin{lemma}\label{lem:avoid}
Let $P$ be a subpath of $T$. If there exists a vertex $r\in V(P)$ such that $w_G(e_i,e_j)\geq (1+\varepsilon)\lambda$ for every open interval $(e_i,e_j)$ in $P$ which does not contain $r$, then $A_r$ avoids $Z_2$.
\end{lemma}

\begin{proof}
We may assume that $r$ is not a leaf of $P$, as otherwise $A_r = 0$.
For the purpose of contradiction,
let $e_1 > e_2$ and $f_1 < f_2$ such that $A_r(e_2,f_1) = A_r(e_1,f_2) = 1$. Let $X$ and $Y$
be the middle subpaths of $P$ obtained by removing $e_2,f_1$ and $e_1,f_2$ from $P$ respectively.
Then, $w(\delta(X)) = w_G(e_2,f_1)$ and $w(\delta(Y)) = w_G(e_1,f_2)$.
By the posi-modularity of the cut function (\Cref{lem:posimod cut}),
\[w(\delta(X\setminus Y)) + w(\delta(Y\setminus X)) \leq w(\delta(X)) + w(\delta(Y)) < 2(1+\varepsilon)\lambda.\]

Observe that the left endpoints of $X$ and $Y$ are different (see \Cref{fig:example} for an example).
Similarly, the right endpoints of $X$ and $Y$ are different.
Hence, $X\setminus Y$ and $Y\setminus X$ are nonempty.
It follows that $\min\{w(\delta(X\setminus Y)), w(\delta(Y\setminus X))\} < (1+\varepsilon)\lambda.$
Since $X\setminus Y$ and $Y\setminus X$ correspond to open intervals in $P$ which do not contain $r$, this contradicts our assumption on $r$.
\end{proof}

\begin{figure}[h]
    \centering
    \includegraphics[width=0.85\textwidth]{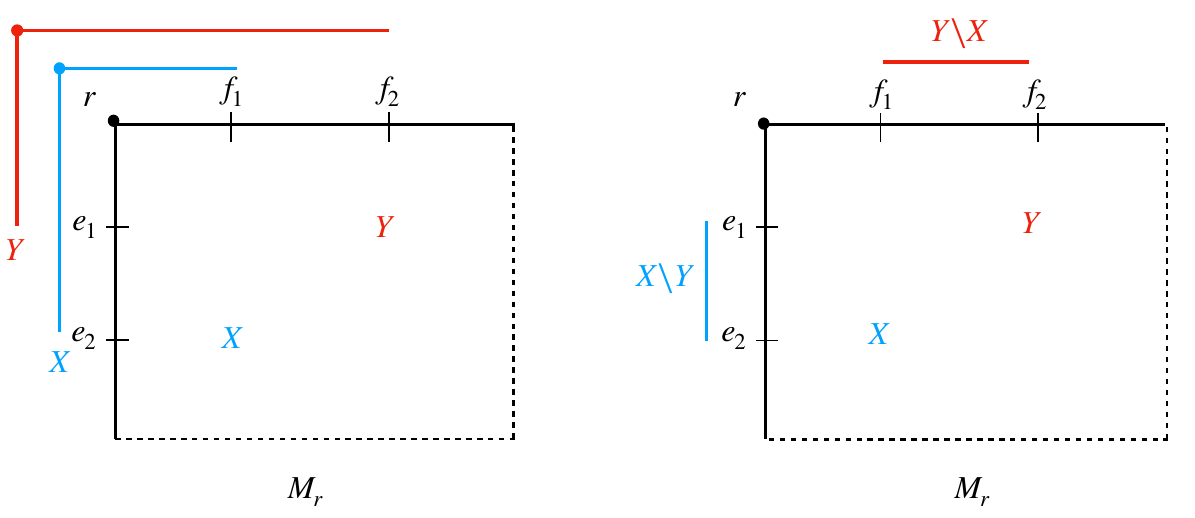}
    \caption{The matrix $M_r$ formed by the path $P$ given a node $r \in V(P)$. (Left) The intervals $A$ and $B$ (blue and red, respectively) contain $r$ and do not share the same endpoints. (Right) The intervals $A \setminus B$ and $B \setminus A$ are nonempty and do not contain $r$.}
    \label{fig:example}
\end{figure}

We are ready to prove \Cref{claim: small SP}.

\begin{proof}[Proof of \Cref{claim: small SP}]
Let $n'$ and $m'$ be the number of rows and columns of $A_r$, respectively.
Then, $|A_r| \leq n' + m' - 1 \leq |E(P)|$ by \Cref{lem:ex Z} and \Cref{lem:avoid}.
The proof is complete by noticing that $|A_r| = |S_P|$.
\end{proof}

We have established that for every $1\leq i\leq k$, there are at most $|E(P)|$ intervals in $S_P$ for every path $P\in \mathcal{P}_i$.
However, we also need to efficiently find them in order to form the set $\tilde B_i = \cup_{P\in \mathcal{P}_i} S_P$, which is fed into \textsc{Focus} (\Cref{alg:mwu_focus}).
Let $P\in \mathcal{P}_i$ and let $r\in V(P)$ be the common vertex of the two subpaths from $\mathcal{P}_{i-1}$ which form $P$.
Then, this task amounts to listing all the positive entries in $A_r$.

To achieve this in nearly linear work, we adapt the framework of~\cite{conf/stoc/MukhopadhyayN20}, which was originally designed to find a minimum-weight interval in $T$.
Let $n'$ and $m'$ be the number of rows and columns of $A_r$ respectively.
We identify a region of $A_r$ containing all the positive entries, and show that it avoids the pattern $Z_3$.
The proof is based on the submodularity of the cut function.
Then, by \Cref{lem:ex Z}, this region has size at most
\[ Ex(Z_3,n',m') = O(n'+m') = O(|E(P)|).\]
Since the region has linear size, we can afford to query every entry in the region, assuming that each query takes $\tilde O(1)$ work.
More details can be found in \Cref{sec:listing n cuts in a path}.

\subsection{The Epoch Algorithm}
\label{sec:epoch_algorithm}

In this section, we develop an algorithm for clearing a spanning tree.
The epoch algorithm is then obtained by sequentially clearing every tree in the packing $\mathcal{T}$.
Indeed, once a tree $T$ is cleared, the weight of every cut which 1-or-2-respects $T$ will remain at least $(1+\varepsilon)\lambda$ because the edge weights are nondecreasing throughout \Cref{alg:mwu_parallel_general}.
This only increases the depth of the algorithm by a factor of $|\mathcal{T}| = O(\log n)$.

Fix a tree $T\in \mathcal{T}$.
In order to implement \textsc{ExtractAndFocus} (\Cref{def:extract_focus op}) efficiently, we need a data structure that enables a fast execution of \Cref{alg:mwu_focus} given any subset of 1-or-2-respecting cuts of $T$.
This will be developed in \Cref{lem:cut oracle ds} (MWU Cut Oracle).
Assuming its availability, our goal is to compute a good core-sequence $\mathcal{\tilde B} = (\tilde B_1, \tilde B_2, \dots, \tilde B_\ell)$ for $T$.

\begin{definition}
    Given a spanning tree $T$ of $G$ and $F\subseteq E(T)$, let $T/F$ be the tree obtained from $T$ by contracting the edges in $F$. If $T/F$ is a path, we call it a \textit{path minor} of $T$.
\end{definition}

Note that $T/F$ is a spanning tree of $G/F$.
In particular, every 1-or-2-respecting cut of $T/F$ is a 1-or-2-respecting cut of $T$ with the same value.

Below is a high-level description of the algorithm for clearing a tree.
We remark that the path minors in the algorithm need not be formed explicitly.
They are only used to illustrate the connection to the path case in the previous section.

\paragraph{Tree Algorithm.} Given a spanning tree $T$ of $G$,

\begin{enumerate}
    \item \label{step:cut_oracle}
    We construct the \textit{MWU Cut Oracle} $\mathcal{D}_T$ (\Cref{lem:cut oracle ds}). It is used to query the weight of a 1-or-2-respecting cut of $T$, and to apply $\textsc{Focus}$ (\Cref{alg:mwu_focus}) on a batch of 1-or-2-respecting cuts of $T$.

    \item \textbf{Decompose into paths.}
    We root $T$ at an arbitrary vertex and decompose $T$ into a collection of edge-disjoint paths $\mathcal{P}$ such that every root-to-leaf path in $T$ intersects at most $\log n$ paths in $\mathcal{P}$.
    This can be achieved, e.g., using the parallel bough decomposition of ~\cite{conf/spaa/GeissmannG18}.

    \item \label{step:paths}
    \textbf{Clear paths.}
    For every $P\in \mathcal{P}$, let $T_P$ be the path minor of $T$ obtained by contracting edges not in $E(P)$.
    We run the path algorithm on $T_P$ for all $P\in \mathcal{P}$ in parallel. This clears the 1-respecting cuts of $T$, and the 2-respecting cuts of $T$ that intersect 2 edges of some path in $\mathcal{P}$.

    \item \label{step:path_pairs}
    \textbf{Clear path pairs.} We compute the set $\mathcal{I}$ of \emph{interested path pairs} $(P,Q)$, along with their corresponding edges
    $E_{P\to Q}\subseteq E(P)$ and $E_{Q\to P}\subseteq E(Q)$ using the parallel algorithm of~\cite{conf/spaa/Lopez-MartinezM21}.
    \begin{enumerate}
        \item For every $(P,Q)\in \mathcal{I}$, let $T(P,Q)$ be the path minor of $T$ obtained by contracting edges not in $E_{P\to Q}\cup E_{Q\to P}$.
        \item We set $\tilde B := \cup_{(P,Q)\in \mathcal{I}} S(P,Q)$, where $S(P,Q)$ is the set of 2-respecting cuts of $T(P,Q)$ with weight less than $(1+\varepsilon)\lambda$.
        \item We apply \textsc{Focus} on $\tilde B$ using the MWU Cut Oracle.
    \end{enumerate}
    \end{enumerate}

\medskip

The tree algorithm relies on the following data structure, which we develop in \Cref{sec:mwu oracle}.

\begin{lemma} [MWU Cut Oracle] \label{lem:cut oracle ds}
Let $G$ be a graph with edge costs $c\in \R^m_{\geq 0}$, edge weights $w\in \R^m_{\geq 0}$ and edge congestion $\congestion\in \R^m_{\geq 0}$.
Let $\lambda,\varepsilon>0$.
Given a spanning tree $T$ of $G$, there is a data structure $\mathcal{D}_T$ which supports the following operations:
    \begin{itemize}
        \item $\mathcal{D}_T.\textsc{CutValue}(s)$: Given a 1-or-2-respecting cut $s\in \mathcal{C}_T$, return its weight $w(\cut_T(s))$ in $\poly( \log n)$ work and depth.
        \item $\mathcal{D}_T.\textsc{Focus}(B)$
        : Let $A^\top$ be the constraint matrix of \eqref{lp:cut_cover} when it is expressed as \eqref{lp:cover}, i.e.,
        \[A_{e,S} = \begin{cases} \frac{1}{kc_e}, &\text{ if }e\in \delta(S), \\ 0, &\text{ otherwise.} \end{cases}\]
        Given a set of 1-or-2-respecting cuts $B\subseteq \mathcal{C}_T$, implement $\textsc{Focus}_{A,\lambda,\varepsilon}(B,w,\congestion)$
        in $\tilde O((m+|B|)\log(|B|)/\varepsilon^2)$ work
        and $\tilde O(\log(|B|)/\varepsilon^2)$ depth.
    \end{itemize}
    The data structure $\mathcal{D}_T$ can be constructed in $\tilde O(m)$ work and $\tilde O(1)$ depth.
\end{lemma}

Equipped with the MWU Cut Oracle, the tree algorithm achieves the following guarantees.

\begin{restatable}[Tree Algorithm]{lemma}{focustree}
\label{lem:focus tree}
    Let $G$ be a graph with edge weights $w\in \R^m_{>0}$ and edge congestion $\congestion\in \R^m_{\geq 0}$.
    Given a spanning tree $T$ of $G$, the data structure $\mathcal{D}_T$ and scalars $0<\lambda\leq \OPT_w, 0<\varepsilon<0.5$, there is an algorithm called \textsc{ClearTree}$_{G,\lambda,\varepsilon}(T,\mathcal{D}_T,w,\congestion)$, which iteratively applies \pminsubset until $\mincut_w(T) \geq (1+\varepsilon)\lambda$ or $\|\congestion\|_\infty \geq \ln (m)/\varepsilon$.
    The algorithm runs in $\tilde O(m/\varepsilon^2)$ work and $\tilde O(1/\varepsilon^2)$ depth.
\end{restatable}

We divide the proof of \Cref{lem:focus tree} into three sections~(\Cref{sec:listing n cuts in a path,sec:extract_focus path,sec:extractandfocus tree}).
In \Cref{sec:listing n cuts in a path}, we show how to compute $S_P$ for a path $P$ that satisfies the condition in \Cref{claim: small SP}.
Then, we formally describe Step~\ref{step:paths} in \Cref{sec:extract_focus path}.
Finally, Step~\ref{step:path_pairs} in elaborated in \Cref{sec:extractandfocus tree}.

The proof of \Cref{thm:epoch algorithm} is straightforward given \Cref{lem:cut oracle ds,lem:focus tree}.

\begin{proof} [Proof of \Cref{thm:epoch algorithm}]
    We start by computing a tree packing $\mathcal{T}$ using \Cref{thm:tree packing}. For each tree $T \in \mathcal{T}$, we construct the MWU cut oracle $\mathcal{D}_T$ and run $\textsc{ClearTree}_{G,\lambda,\varepsilon}(T,\mathcal{D}_T,w,\congestion)$. The correctness and running time follow immediately from \Cref{thm:tree packing} and \Cref{lem:cut oracle ds,lem:focus tree}.
\end{proof}

\subsection{Extracting 2-Respecting Cuts on a Rooted Path} \label{sec:listing n cuts in a path}

Let $P$ be a path minor of a spanning tree $T$ of $G$, obtained by contracting a subset of edges $F\subseteq E(T)$.
Note that $P$ is a spanning tree of $G/F$.
Moreover, $\mathcal{C}_P\subseteq \mathcal{C}_T$.
For every $s\in \mathcal{C}_P$, we denote $\cut_P(s)$ as the set of edges in $G/F$ having exactly one endpoint in $\shore_P(s)$.
Then, $w(\cut_P(s)) = w(\cut_T(s))$ for all $s \in \cC_P$.

Fix a root $r\in V(P)$.
Given a 2-respecting cut $s=\{e,f\}\in E^2(P)$, recall that $\shore_P(s)$ is the vertex set of the middle subpath in $P\setminus \{e,f\}$.
We say that $s$ \emph{contains} $r$ if $r\in \shore_P(s)$.

The goal of this section is to prove the following lemma.

\begin{lemma} [Extracting 2-Respecting Cuts] \label{lem:small 2-respecting cuts}
Let $G$ be a graph with edge weights $w\in \R^m_{>0}$.
Let $T$ be a spanning tree of $G$ and $\lambda,\varepsilon>0$.
Given a path minor $P$ of $T$ with root $r\in V(P)$ and the data structure $\mathcal{D}_T$, there is an algorithm called $\textsc{ExtractCutsInPath}_{G,T,\lambda,\varepsilon}(P,r, \mathcal{D}_T)$ which outputs a set of 2-respecting cuts $S \subseteq E^2(P)$ satisfying the following properties:
\begin{enumerate}
    \item $|S| = O(|E(P)|)$;
    \item If every 2-respecting cut of $P$ that does not contain $r$ has weight at least $(1+\varepsilon)\lambda$, then $S = \{ s \in E^2(P) \colon w(\cut_P(s)) < (1+\varepsilon)\lambda\}$.
\end{enumerate}
The algorithm runs in $\tilde O(|E(P)|)$ work and $\poly\log(n)$ depth.
\end{lemma}

Let $e_1<e_2<\dots<e_\ell$ be the edges of $P$, where $e_i$ and $e_{i+1}$ share a common vertex for all $i<\ell$.
Given a non-leaf vertex $r\in V(P)$, let $e_k$ and $e_{k+1}$ be the edges incident to $r$.
Recall the matrix $M_r\in \R_{\geq 0}^{k\times (\ell-k)}$ given by $M_r(e_i,e_j) = w(\cut_T(\{e_i,e_j\}))$ for every open interval $(e_i,e_j)$ containing $r$.
The $i$-th row of $M_r$ corresponds to the edge $e_{k+1-i}$, while the $j$-th column of $M_r$ corresponds to the edge $e_{k+j}$.
Since we will mainly work with the matrix $M_r$, for the sake of convenience we overload its indexing as $M_r(i,j) := M_r(e_{k+1-i},e_{k+j})$.
We also define the binary matrix $A_r\in \{0,1\}^{k\times (\ell-k)}$ as $A_r(i,j) := 1$ if and only if $M_r(i,j)<(1+\varepsilon)\lambda$.

Let $S_P := \{s \in E^2(P) \colon w(\cut_P(s)) < (1+\varepsilon)\lambda\}$.
Although in general $|S_P|$ can be as large as $O(|E(P)|^2)$, we have shown that $|S_P| =  O(|E(P)|)$ whenever every 2-respecting cut of $P$ that does not contain $r$ has weight at least $(1+\varepsilon)\lambda$ (\Cref{claim: small SP,lem:avoid}).
We restate it here for convenience.

\begin{lemma} \label{lem:Q av Z}
 If every 2-respecting cut of $P$ that does not contain $r$ has weight at least $(1+\varepsilon)\lambda$, then $A_r$ avoids $Z_2$.
\end{lemma}

It is left to give an algorithm for computing $S_P$.
In the matrix $M_r$, for every column $j\in [\ell-k]$, let $i^*_j := \argmin_{i \in [k]}M_r(i,j)$ be the row containing the smallest entry of column $j$ (breaking ties by selecting the largest such row).
Let $L := \{(i^*_j,j):j\in [\ell-k]\}$ be the set of these coordinates.
We define the binary matrix $B_r\in \{0,1\}^{k\times (\ell-k)}$ as $B_r(i,j) := 1$ if and only if $(i,j)\in L$.

The following fact was shown by \cite{conf/stoc/MukhopadhyayN20}. We give an alternate proof of this fact using forbidden matrix theory.
\begin{lemma} [\cite{conf/stoc/MukhopadhyayN20}] \label{lem:compute W}
     The matrix $B_r$ avoids $Z_2$.
     Therefore,  $i^*_1 \leq i^*_2 \leq \ldots \leq i^*_{\ell-k}$. Moreover, $L$ can be computed in $\tilde O(|E(P)|)$ work and $\poly(\log n)$ depth, given the data structure $\mathcal{D}_T$.
\end{lemma}
\begin{proof}
For the purpose of contradiction, let $i_1<i_2$ and $j_1<j_2$ such that $B_r(i_2,j_1) = B_r(i_1,j_2) = 1$.
Let $X := \shore_P(\{e_{k+1-i_2},e_{k+j_1}\})$ and $Y := \shore_P(\{e_{k+1-i_1},e_{k+j_2}\})$.
Then, $M_r(i_2,j_1) = w(\delta(X))$ and $M_r(i_1,j_2) = w(\delta(Y))$.
By the submodularity of the cut function  (\Cref{lem:posimod cut}),
\begin{align*}
    M_r(i_2,j_1) + M_r(i_1,j_2) &= w(\delta(X)) + w(\delta(Y)) \\
    &\geq w(\delta(X\cap Y)) + w(\delta(X \cup Y)) = M_r(i_1,j_1) + M_r(i_2,j_2)
\end{align*}
where the last equality is due to $X\cap Y = \shore_P(\{e_{k+1-i_1},e_{k+j_1}\})$ and $X\cup Y = \shore_P(\{e_{k+1-i_2},e_{k+j_2}\})$.
Since $B_r(i_2,j_1) = 1$, we have $M_r(i_1,j_1) \geq M_r(i_2,j_1)$.
It follows that $M_r(i_2,j_2)\leq M_r(i_1,j_2)$.
However, this contradicts $B_r(i_1,j_2) = 1$ because $i_2>i_1$.

To compute $L$, we use the recursive algorithm of~\cite{conf/stoc/MukhopadhyayN20}, which has depth $\log(|P|)$.
At each level, it queries at most $|P|$ entries of the matrix $M_r$.
Each query involves returning the value of a 2-respecting cut of $P$, which can be performed in $\poly(\log n)$ work and depth by \Cref{lem:cut oracle ds}.
\end{proof}

We remark that the property of $B_r$ avoiding $Z_2$ does not rely on the assumption in \Cref{lem:Q av Z}.
We are ready to prove \Cref{lem:small 2-respecting cuts}.

\begin{proof} [Proof of \Cref{lem:small 2-respecting cuts}]
We present the algorithm and its analysis.

\paragraph{Algorithm.} We are given a path minor $P$ of a spanning tree $T$ of a graph $G = (V,E)$ with weight vector $w$, a root $r \in V(P)$, global parameters $\lambda, \varepsilon$ and a reference to the cut oracle $\mathcal{D}_T$ of $T$ as inputs.

\begin{enumerate}
    \item\label{ass:no_zero_columns} Compute the set $L$ using \Cref{lem:compute W}. Without loss of generality, we may assume that every column of $A_r$ is not a zero vector. Otherwise, we can drop such columns from consideration.
    \item For $2\leq j\leq \ell-k$, define $R_j := [i^*_{j-1},i^*_{j}] \times [j-1,j]$. It is the smallest rectangle containing $(i^*_{j-1},j-1)$ and $(i^*_j,j)$.
    We also define $R_1 := [1,i^*_1]\times [1,1]$ and $R_{\ell-k+1} := [i^*_{\ell-k},k]\times [\ell-k,\ell-k]$.
    Let $\mathcal{R} := \cup_{j=1}^{\ell-k+1} R_j$ be the union of these entries.
    \item Return $S := \{\{e_{k+1-i},e_{k+j}\} \colon (i,j)\in \mathcal R, A_r(i,j) = 1\}$.
    The set $S$ is computed by calling $\mathcal{D}_T.\textsc{CutValue}(e_{k+1-i},e_{k+j})$ for each $(i,j) \in \mathcal R$.
\end{enumerate}

\paragraph{Analysis.}

We analyze the key properties of the returned set $S$. First, we prove that the set $\mathcal R$ has a linear size.
Let $B'_r\in \{0,1\}^{k\times (\ell-k)}$ be the binary matrix given by $B'_r(i,j) := 1$ if and only if $(i,j)\in \mathcal R$.

\begin{claim}\label{claim:small S tilde P}
The matrix $B'_r$ avoids $Z_3$.
    Therefore, $|S| \leq |\mathcal R| = |B'_r| = O(|E(P)|)$.
\end{claim}
\begin{proof}
    For the purpose of contradiction, suppose that $B'_r$ contains $Z_3$. Let $(i,j)$ be an entry in $B'_r$ that corresponds to the center 1 in $Z_3$.
    Then, $1<j<\ell-k$.
    By construction, we have $i \in [i^*_{j-1},i^*_{j+1}]$.
    Let $(i',j')$ and $(i'',j'')$ be the entries in $B'_r$ that correspond to the bottom-left 1 and top-right 1 of $Z_3$ respectively.
    If $i \geq i^*_j$, then $i'> i \geq i^*_j$ and $j' < j$ imply that $(i',j') \not \in \mathcal R$.
    Similarly, if $i \leq i^*_j$, then $i'' < i \leq i^*_j$ and $j' > j$ imply that $(i'',j'') \not \in \mathcal R$.
    Either case contradicts $B'_r(i',j') = B'_r(i'',j'') =1$.
    Since $B'_r$ avoids $Z_3$, \Cref{lem:ex Z} yields $|B'_r| \leq Ex(Z_3,k,\ell-k) = O(k+(\ell-k)) = O(\ell) = O(|E(P)|)$.
\end{proof}

\begin{claim}  If every 2-respecting cut of $P$ that does not contain $r$ has weight at least $(1+\varepsilon)\lambda$, then $S$ is the set of 2-respecting cuts of $P$ with weight less than $(1+\varepsilon)\lambda$.
\end{claim}
\begin{proof}
Since every 2-respecting cut of $P$ that does not contain $r$ has weight at least $(1+\varepsilon)\lambda$, it suffices to prove that $A_r(i,j) = 0$ for all $(i,j)\notin \mathcal R$.
For the purpose of contradiction, suppose that $A_r(i',j') = 1$ for some $(i',j')\notin \mathcal R$.
We will show that $A_r$ contains $Z_2$, which contradicts \Cref{lem:Q av Z}.
First, observe that $A_r(i^*_j,j) = 1$ for all $j\in [\ell-k]$ by the assumption in Step \ref{ass:no_zero_columns}.
Next, since $(i',j') \not \in \mathcal R$, we have $i\notin [i^*_{j'-1},i^*_{j'+1}]$ (where $i^*_{j'-1} := 1$ if $j'=1$, and $i^*_{j'+1} := k$ if $j'=\ell-k$).
If $i' > i^*_{j'+1}$, then the pattern $Z_2$ is formed by $A_r(i',j') = A_r(i^*_{j'+1},j'+1) = 1$.
On the other hand, if $i' < i^*_{j'-1}$, then the pattern $Z_2$ is formed by $A_r(i',j') = A_r(i^*_{j'-1},j'-1) = 1$.
\end{proof}

Therefore, the correctness of the returned set $S$ and its size follows.

\paragraph{Running Time.} Finally, we establish the time complexity.
The first step takes $\tilde O(|E(P)|)$ work and $\poly(\log n)$ depth by \Cref{lem:compute W}.
The second and third steps can be implemented by calling $\mathcal{D}_T.\textsc{CutValue}(\{e_{k+1-i},e_{k+j}\})$ for each $(i,j) \in \mathcal R$.
Since $|\mathcal R| = O(|E(P)|)$ by \Cref{claim:small S tilde P}, this takes $\tilde O(|E(P)|)$ work and $\poly(\log n)$ depth by \Cref{lem:cut oracle ds}.
\end{proof}

\subsection{Clearing Paths} \label{sec:extract_focus path}

In the previous section, given a path minor $P$ of a spanning tree $T$ and a root $r\in V(P)$, we developed an algorithm for computing all 2-respecting cuts of $P$ with weight less than $(1+\varepsilon)\lambda$, assuming that every 2-respecting cut of $P$ that does not contain $r$ has weight at least $(1+\varepsilon)\lambda$.
In this section, we will use this algorithm to clear $P$, by implementing a sequence of \textsc{ExtractAndFocus} operations.
We will also show how it can be extended to clear multiple path minors simultaneously.
The formal guarantee is stated as follows.

\begin{lemma}[Clearing Paths] \label{lem:focus path}
    Let $G$ be a graph with edge weights $w\in \R^m_{>0}$ and edge congestion $\congestion\in \R^m_{\geq 0}$.
    Let $T$ be a spanning tree of $G$ and $\lambda,\varepsilon>0$.
    Given a family of path minors $\mathcal P$ of $T$ and the data structure $\mathcal{D}_T$, there is an algorithm called \textsc{ClearMultiplePaths}$_{G,T,\lambda,\varepsilon}(\mathcal{P},\mathcal{D}_T,w,\congestion)$ which iteratively applies \pminsubset until $\mincut_w(P) \geq (1+\varepsilon)\lambda$ for all $P\in \mathcal{P}$ or $\|\congestion\|_\infty \geq \ln (m)/\varepsilon$.
    The algorithm runs in $\tilde O((m+\sum_{P\in \mathcal P}|E(P)|)\log(\sum_{P\in \mathcal P}|E(P)|)/\varepsilon^2)$ work and $\tilde O(\log(\sum_{P\in \mathcal P}|E(P)|)/\varepsilon^2)$ depth.
\end{lemma}

Let us first consider the singleton case $\mathcal{P} = \{P\}$.
The extension to multiple paths is straightforward and will be addressed at the end of this section.
Let $e_1<e_2<\dots<e_{\ell}$ be the edges of $P$, where $e_i$ and $e_{i+1}$ share a common vertex for all $i< \ell$.
For $i\leq j$, we denote the subpath of $P$ that starts with $e_i$ and ends with $e_j$ as the closed interval $[e_i,e_j]$

\begin{definition} \label{def:recursion tree path instance}
    Given a closed interval $Q = [e_i,e_j]$ where $i\leq j$, we define the tree $\mathbb{T}_Q$ recursively as follows.
    First, create a node representing $Q$, and root the tree at this node.
    If $i=j$, then we are done.
    Otherwise, consider the subintervals $Q_L := [e_i ,e_{\text{mid}}]$ and $Q_R := [e_{\text{mid}+1},e_j]$, where $\text{mid} := \lfloor (i+j)/2 \rfloor$.
    Create two nodes representing $Q_L$ and $Q_R$, and make them the children of $Q$.
    The subtree of $\mathbb{T}_Q$ rooted at $Q_L$ and $Q_R$ are $\mathbb{T}_{Q_L}$ and $\mathbb{T}_{Q_R}$ respectively.
\end{definition}

Consider the tree $\mathbb{T}_P$, which has height $\ceil{\log (\ell)}$.
For the sake of convenience, if a node in $\mathbb{T}_P$ represents an interval $Q$, we label the node as $Q$ as well.
Note that the leaves of $\mathbb{T}_P$ represent the edges of $P$.
We say that a node $Q\in V(\mathbb{T}_P)$ is at \emph{level} $i$ if it is at depth $\ceil{\log (\ell)}-i$ from the root.
For $i=0,1,\dots,\ceil{\log(\ell)}$, let $\mathcal{L}_i$ denote the set of nodes at level $i$ in $\mathbb{T}_P$.
Note that $\mathcal{L}_{\ceil{\log(\ell)}} = \{P\}$.

  \begin{definition}[Rooted path minor] \label{def:rooted path instance}
  Given a node $Q$ in the tree $\mathbb{T}_P$, let $P_Q$ be the path minor of $P$ obtained by contracting edges not in $E(Q)$.
  If $Q$ is not a leaf in $\mathbb{T}_P$, we define the \emph{root} $r_Q$ of $P_Q$ as the common node of two subintervals represented by the children of $Q$.
  \end{definition}

We are ready to state the algorithm for clearing a single path minor $P$ (\Cref{alg:focus_path}).
It consists of $\ceil{\log(\ell)}+1$ iterations.
In iteration $i=0,1,\dots,\ceil{\log(\ell)}$, it does the following for every node $Q\in \mathcal{L}_i$ in parallel.
If $Q=[e_j,e_j]$ for some $j$, then it extracts the 1-respecting cut $\{e_j\}$ if $w(\cut_P(\{e_j\}))<(1+\varepsilon)\lambda$.
Otherwise, it calls \textsc{ExtractCutsInPath} on the path minor $P_Q$ with root $r_Q$ to extract the 2-respecting cuts of $P_Q$ with weight less than $(1+\varepsilon)\lambda$.
After that, these extracted cuts are fed to \textsc{Focus}.
We remark that the tree $\mathbb{T}_P$ need not be constructed explicitly.

\begin{algorithm}[htb!]
  \caption{$\textsc{ClearPath}_{G,T,\lambda,\varepsilon}(P,\mathcal{D}_T,w,\congestion)$}
  \label{alg:focus_path}
  \SetKwInOut{Input}{Input}
  \SetKwInOut{Output}{Output}
  \SetKwComment{Comment}{$\triangleright$\ }{}
  \SetKw{And}{\textbf{and}}
  \SetKw{Or}{\textbf{or}}
  \Input{A path minor $P$ of $T$, the cut oracle $\mathcal{D}_T$ for $T$, edge weights $w\in \R^m_{\geq 0}$, congestion $\congestion\in \R^m_{\geq 0}$.}
  \Output{A pair $(w,\congestion)$}

  \For{$i = 0,1,\dots,\ceil{\log(\ell)}$}{
  \ForEach{$Q \in \mathcal{L}_i$}{
    \If{$Q=[e_j,e_j]$ for some $j$ and $w(\cut_T(\{e_j\}))<(1+\varepsilon)\lambda$}{
        $S_Q \gets \{\{e_j\}\}$
    }
    \Else{
        $S_Q \gets \textsc{ExtractCutsInPath}_{G,T,\lambda,\varepsilon}(P_Q,r_Q, \mathcal{D}_T)$\;
    }

  }

  $\tilde B_i \gets \bigcup_{Q \in \mathcal{L}_i} S_Q$\;

  $\mathcal{D}_T.\textsc{Focus}_{G,T,\lambda,\varepsilon}(\tilde B_i, w,\congestion)$\;
  }
  \Return{$(w,\congestion)$.}
\end{algorithm}

We prove that \Cref{alg:focus_path} achieves the guarantees in \Cref{lem:focus path}.

\begin{proof}[Proof of \Cref{lem:focus path} for the singleton case $\mathcal{P} = \{P\}$.]
   We prove correctness and running time.

\paragraph{Correctness.}
We may assume that all the calls to \textsc{Focus} did not terminate due to $\|\congestion\|_\infty \geq \ln(m)/\varepsilon$.
By design, the algorithm iteratively runs the \pminsubset operation as long as every cut in $\tilde B_i$ has weight less than $(1+\varepsilon)\lambda$ when it is fed to \textsc{Focus}, i.e., $\tilde B_i$ is a focus set.
So, it remains to prove that $\tilde B_i$ is a focus set for all $i$, and $\mincut_w(P) \geq (1+\varepsilon)\lambda$ when the algorithm terminates.

We proceed by strong induction on the iterations $i\geq 0$.
In particular, it suffices to prove that for every $Q\in \mathcal{L}_i$, we have $w(\cut_T(s))<(1+\varepsilon)\lambda$ for all $s\in S_Q$ before calling \textsc{Focus}, and $\mincut_w(P_Q) \geq (1+\varepsilon) \lambda$ afterwards.
The base case $i=0$ is trivial because every $Q\in \mathcal{L}_0$ represents a single edge of $P$.
For the inductive step, assume that the statement is true for iterations $0,1,\dots,i$, and consider iteration $i+1$.
Fix a node $Q \in \mathcal{L}_{i+1}$.
We may assume that $Q$ is not a leaf of $\mathbb{T}_P$; otherwise it is again trivial.
Let $Q_L$ and $Q_R$ be the two children of $Q$.
Then, $r_Q$ is the common node of the subpaths represented by $Q_L$ and $Q_R$.
By the inductive hypothesis, we have
 \[ \mincut_w(P_{Q_L}) \geq (1+\varepsilon)\lambda \qquad \text{ and } \qquad \mincut_w(P_{Q_R}) \geq (1+\varepsilon)\lambda. \]
So, every 2-respecting cut of $P_Q$ that does not contain $r_Q$ has weight at least $(1+\varepsilon)\lambda.$
By \Cref{lem:small 2-respecting cuts}, $S_Q$ is precisely the set of 2-respecting cuts of $P_Q$ with weight less than $(1+\varepsilon)\lambda.$
Hence, after applying \textsc{Focus}, every 2-respecting cut of $P_Q$ has weight at least $(1+\varepsilon)\lambda$.
Every 1-respecting cut of $P_Q$ also has weight at least $(1+\varepsilon)\lambda$ by the induction hypothesis because they correspond to the leaves of the subtree $\mathbb{T}_Q$ rooted at $Q$.
It follows that $\mincut_w(P_Q) \geq (1+\varepsilon)\lambda$ as desired.

\paragraph{Running Time.}
For every $Q\in V(\mathbb{T}_P)$, if $Q$ is a leaf of $\mathbb{T}_P$, then extracting the 1-respecting cut of $P_Q$ amounts to calling $\mathcal{D}_T$.\textsc{CutValue}, which takes $\poly\log n$ work and depth by \Cref{lem:cut oracle ds}.
Otherwise, we extract the 2-respecting cuts of $P_Q$ by calling \textsc{ExtractCutsinPath}, which takes $\tilde O(|E(P_Q)|)$ work and $\poly\log n$ depth by \Cref{lem:small 2-respecting cuts}.
Summing over all the nodes in $\mathbb{T}_P$, we incur
\[\sum_{i=0}^{\ceil{\log(\ell)}} \sum_{Q \in \mathcal{L}_i} \tilde O(|E(P_Q)|) = \tilde O(|E(P)|)\]
work because the subpaths represented by nodes in $\mathcal{L}_i$ are edge-disjoint for every $i$.
On the other hand, the depth is $\poly\log n$ because there are $O(\log n)$ iterations.

Next, we calculate the work and depth incurred by \textsc{Focus}.
For every $Q\in V(\mathbb{T}_P)$, we have $|S_Q|\leq O(|E(P_Q)|)$ by \Cref{lem:small 2-respecting cuts},
Hence, for every iteration $i$,
\[|\tilde B_i| = \sum_{Q\in \mathcal{L}_i} |S_Q| \leq \sum_{Q\in \mathcal{L}_i} |E(P_Q)| = |E(P)|.\]
So, calling $\mathcal{D}_T.\textsc{Focus}(\tilde B_i)$ takes $\tilde O((m+|E(P)|)\log(|E(P)|)/\varepsilon^2)$ work and $\tilde O(\log(|E(P)|)/\varepsilon^2)$ depth by \Cref{lem:cut oracle ds}.
As there are $O(\log n)$ iterations, the work and depth bounds follow.
\end{proof}

\paragraph{Multiple Paths.}  To extend \Cref{alg:focus_path} to multiple path minors, we run every iteration of \Cref{alg:focus_path} for all $P \in \mathcal{P}$ in parallel.
For $i=0,1,\dots,\ceil{\log(n-1)}$ and $P\in \mathcal{P}$, let $\mathcal{L}_i(P)$ be the set of nodes at level $i$ in the tree $\mathbb{T}_P$.
Note that $\mathcal{L}_i(P) = \emptyset$ if $i$ is larger than the height of $P$.
In every iteration $i$, we extract the corresponding cuts for all $Q\in \mathcal{L}_i(P)$ and $P\in \mathcal{P}$ in parallel.
Then, we collect these cuts $\tilde B_i := \cup_{P\in \mathcal{P}} \cup_{Q\in \mathcal{L}_i(P)}S_Q$ and feed them to \textsc{Focus}.
It is easy to check that the same correctness proof applies.
To bound the running time, we know that calling \textsc{ExtractCutsInPath} on the nodes in $\cup_{P\in \mathcal{P}} V(\mathbb{T}_P)$ takes $\sum_{P\in \mathcal{P}} \tilde O(|E(P)|)$ work and $\poly\log n$ depth.
We also know that $|\tilde B_i| \leq \sum_{P\in \mathcal{P}} \sum_{Q\in \mathcal{L}_i} |S_Q| \leq \sum_{P\in \mathcal{P}} |E(P)|$.
Thus, applying \textsc{Focus} on $\tilde B_i$ takes $\tilde O((m+\sum_{P\in \mathcal P}|E(P)|)\log(\sum_{P\in \mathcal P}|E(P)|)/\varepsilon^2)$ work and $\tilde O(\log(\sum_{P\in \mathcal P}|E(P)|)/\varepsilon^2)$ depth.
As there are $O(\log n)$ iterations, the work and depth bounds follow.

\subsection{Clearing Path Pairs} \label{sec:extractandfocus tree}

In the previous section, we have developed the path algorithm as outlined in \Cref{sec:tech_overview}.
In this section, we will formally describe the tree algorithm as outlined in \Cref{sec:epoch_algorithm}, and explain the last step of clearing path pairs.
Fix a spanning tree $T$ of $G=(V,E)$.
We first root $T$ at an arbitrary vertex $r$, and decompose $T$ into a collection of edge-disjoint paths that satisfy the following property.

\begin{property}\label{ppt:bough}
Any root-to-leaf path in $T$ intersects at most $\log n$ paths in $\mathcal{P}$.
\end{property}

Geissmann and Gianinazzi gave a parallel algorithm \cite{conf/spaa/GeissmannG18} for computing such a decomposition (which they called a \emph{bough decomposition}).

\begin{lemma}[{\cite[Lemma 7]{conf/spaa/GeissmannG18}}]\label{lem:bough_decomposition}
A rooted tree with $n$ vertices can be decomposed w.h.p.~into a set of edge-disjoint paths $\mathcal{P}$ that satisfy \Cref{ppt:bough} using $O(n\log n)$ work and $O(\log^2 n)$ depth.\footnote{The original algorithm is Las Vegas and produces vertex-disjoint paths. It can be easily modified to produce edge-disjoint paths. Furthermore, it can be converted into a Monte Carlo algorithm using Markov's inequality.}
\end{lemma}

Next, we run \textsc{ClearMultiplePaths} on $\mathcal{P}$.
It clears the cuts that 1-respect $T$, as well as the cuts that 2-respect $T$ on the same path of $\mathcal{P}$.
By \Cref{lem:focus path}, in $\tilde{O}(m)$ work and $\tilde O(1/\varepsilon^2)$ depth, every such cut has weight at least $(1+\varepsilon)\lambda$.

It remains to clear the cuts which 2-respect $T$ on different paths of $\mathcal{P}$.
The number of possible path pairs is $O(n^2)$, so we cannot afford to check every pair if we are aiming for $\tilde{O}(m)$ work.
To overcome this, we rely on the notion of \emph{interested path pairs} by Mukhopadhyay and Nanongkai~\cite{conf/stoc/MukhopadhyayN20}.
For an edge $e\in E(T)$, let $T_e$ be the subtree of $T$ rooted at the lower (further from $r$) endpoint of $e$.
Let $w(T_e,T_f)$ be the total weight of edges between $T_e$ and $T_f$.
We also denote $w(T_e) := w(T_e,T\setminus T_e)$.
If $f\in E(T_e)$, we say that $f$ is a \emph{descendant} of $e$.
If $e$ and $f$ do not lie on the same root-to-leaf path, we say that $e$ and $f$ are \emph{unrelated}.

\begin{definition}\label{def:cross-int}
An edge $e\in E(T)$ is \emph{cross-interested} in an unrelated edge $f\in E(T)$ if
\[w(T_e)<2w(T_e, T_f).\]
\end{definition}
That is, $e$ is cross-interested in an unrelated edge $f$ if the edges between $T_e$ and $T_f$ account for at least half the value of the 1-respecting cut $T_e$.
Observe that if $w(\cut_T(e,f))<(1+\varepsilon)\lambda$, then $e$ must be cross-interested in $f$ (and vice versa), because otherwise $w(\cut_T(e,f)) = w(T_e) + w(T_f) - 2w(T_e,T_f) \geq w(T_f) \geq (1+\varepsilon)\lambda$, which is a contradiction.
This means that we do not have to check every pair of unrelated edges in $T$, but only the ones which are cross-interested in each other.

\begin{definition}\label{def:down-int}
An edge $e\in E(T)$ is \emph{down-interested} in a descendant $f\in E(T_e)$ if
\[w(T_e) < 2w(T_f,T\setminus T_e).\]
\end{definition}
That is, $e$ is down-interested in a descendant $f$ if the edges between $T_f$ and $T\setminus T_e$ account for at least half the value of the 1-respecting cut $T_e$.
Observe that if $w(\cut_T(e,f))<(1+\varepsilon)\lambda$, then $e$ must be down-interested in $f$, because otherwise $w(\cut_T(e,f)) = w(T_e) + w(T_f) - 2w(T_f,T\setminus T_e) \geq w(T_f) \geq (1+\varepsilon)\lambda$, which is again a contradiction.
This means that for every edge in $T$, we do not have to check all of its descendants, but only the ones it is down-interested in.

\begin{claim}[{\cite{conf/stoc/MukhopadhyayN20,conf/sosa/GawrychowskiMW21}}]\label{clm:interested}
For any edge $e\in E(T)$, the edges in which $e$ is cross-interested form a path in $T$ from $r$ to some node $c_e$.
Similarly, the edges in which $e$ is down-interested form a path in $T_e$ from the lower endpoint of $e$ to some node $d_e$.
\end{claim}

The following extends Definitions~\ref{def:cross-int} and \ref{def:down-int} to paths in $\mathcal{P}$.

\begin{definition}
An edge $e\in E(T)$ is \emph{cross-interested (down-interested)} in a path $P\in \mathcal{P}$ if $e$ is cross-interested (down-interested) in some edge $f\in E(P)$.
Given distinct paths $P,Q\in \mathcal{P}$, the ordered pair $(P,Q)$ is called a \emph{cross-interested path pair} if $P$ has an edge cross-interested in $Q$ and vice versa, and a \emph{down-interested path pair} if $P$ has an edge down-interested in $Q$.
An \emph{interested path pair} is a cross-interested or down-interested path pair.
We denote $\mathcal{I}$ as the set of interested path pairs.
\end{definition}

By \Cref{ppt:bough} and \Cref{clm:interested}, every edge in $T$ is cross/down-interested in $O(\log n)$ paths from $\mathcal{P}$.
Hence, there are $O(n\log n)$ interested path pairs.

\begin{definition}
For a cross-interested path pair $(P,Q)\in \mathcal{I}$, we denote $E_{P\to Q}$ as the set of edges in $P$ that are cross-interested in $Q$, and denote $E_{Q\to P}$ as the set of edges in $Q$ that are cross-interested in $P$.
For a down-interested path pair $(P,Q)\in \mathcal{I}$, we denote $E_{P\to Q}$ as the set of edges in $P$ that are down-interested in $Q$, and denote $E_{Q\to P}:=E(Q)$.
\end{definition}

Observe that the total number of edges in $E_{P\to Q}$ and $E_{Q\to P}$ over all interested path pairs $(P,Q)$ is $O(n\log n)$.
To see this, let $e$ be an edge in $T$ and let $P\in \mathcal{P}$ be the unique path which contains $e$.
By \Cref{ppt:bough} and \Cref{clm:interested}, $e$ is cross/down-interested in $O(\log n)$ paths from $\mathcal{P}$.
Furthermore, the root to $e$ path in $T$ intersects at most $\log n$ paths in $\mathcal{P}$.
Hence, it appears in $E_{P\to Q}$ and $E_{Q\to P}$ for $O(\log n)$ interested path pairs.

Gawrychowski, Mozes and Weimann~\cite{conf/sosa/GawrychowskiMW21} gave a sequential $O(m\log n + n\log^2n)$ algorithm for finding all interested path pairs $(P,Q)$, along with $E_{P\to Q}$ and $E_{Q\to P}$.
It led to a simplification and improvement of the minimum cut algorithm in \cite{conf/stoc/MukhopadhyayN20}.
This was subsequently parallelized by López-Martínez, Mukhopadhyay and Nanongkai \cite{conf/spaa/Lopez-MartinezM21}.
In particular, they gave an efficient parallel algorithm for finding all interested path pairs.

\begin{lemma}[{\cite{conf/spaa/Lopez-MartinezM21}}]\label{lem:interested_pairs}
Let $G=(V,E)$ be an edge-weighted graph.
Given a rooted spanning tree $T$ of $G$, let $\mathcal{P}$ be a path decomposition of $T$ satisfying \Cref{ppt:bough}.
All the interested path pairs $(P,Q)$ along with $E_{P\to Q}$ and $E_{Q\to P}$ can be computed in $O(m\log m + n\log^3 n)$ work and $O(\log^2 n)$ depth.
\end{lemma}

Let $(P,Q)$ be an interested path pair, along with its edge sets $E_{P\to Q}$ and $E_{Q\to P}$.
Let $T(P,Q)$ be the path minor of $T$ obtained by contracting edges not in $E_{P\to Q}\cup E_{Q\to P}$.
Let $v$ be the vertex which separates $E_{P\to Q}$ and $E_{Q\to P}$ in $T(P,Q)$.
We call \textsc{ExtractCutsInPath} on $T(P,Q)$ with root $v$ for all interested path pairs $(P,Q)\in \mathcal{I}$ in parallel.
For this, we remark that the path minors $T(P,Q)$ need not be formed explicitly.
Let $\tilde B$ be the set of all returned cuts.
Then, we feed $\tilde B$ into $\mathcal{D}_T.\textsc{Focus}$.
This finishes the description of the tree algorithm (see \Cref{alg:tree} for a pseudocode).

\begin{algorithm}[htb!]
  \caption{$\textsc{ClearTree}_{G,\lambda,\varepsilon}(T,\mathcal{D}_T,w,\congestion)$}
  \label{alg:tree}
  \SetKwInOut{Input}{Input}
  \SetKwInOut{Output}{Output}
  \SetKwComment{Comment}{$\triangleright$\ }{}
  \SetKw{And}{\textbf{and}}
  \SetKw{Or}{\textbf{or}}
  \Input{Spanning tree $T$ and its data structure $\mathcal{D}_T$, edge weights $w$, edge congestion $\congestion$.}
  \Output{A pair $(w,\congestion)$}
  Root $T$ at an arbitrary vertex $r\in V$\;
  Decompose $T$ into edge-disjoint paths $\mathcal{P}$ that satisfy \Cref{ppt:bough}\;
  $\textsc{ClearMultiplePaths}_{G,T,\lambda,\varepsilon}(\mathcal{P},\mathcal{D}_T,w,\congestion)$\;
  Compute all interested path pairs $(P,Q)$ along with $E_{P\to Q}$ and $E_{Q\to P}$ \;
  \ForEach(\Comment*[f]{in parallel}){interested path pair $(P,Q)$}{
    Let $T(P,Q)$ be the path minor of $T$ obtained by contracting edges not in $E_{P\to Q}\cup E_{Q\to P}$\;
    Let $v$ be the vertex separating $E_{P\to Q}$ and $E_{Q\to P}$ in $T(P,Q)$\;
    $S(P,Q) \gets \textsc{ExtractCutsInPath}_{G,T,\lambda,\varepsilon}(T(P,Q),v,\mathcal{D}_T)$\;
  }
   $\tilde B\gets \cup_{(P,Q)\in \mathcal{I}} S(P,Q)$\;
  $\mathcal{D}_T.\textsc{Focus}_{G,T,\lambda,\varepsilon}(\tilde B,w,\congestion)$\;
  \Return{$(w,\congestion)$.}
\end{algorithm}

The following lemma states the guarantee of \Cref{alg:tree}.

\focustree*

\begin{proof}
First, we prove correctness.
We may assume that all the calls to \textsc{Focus} did not terminate due to $\|\congestion\|_\infty\geq \ln(m)/\varepsilon$.
Then, we need to show that \Cref{alg:tree} iteratively applies the \textsc{ExtractAndFocus} operation, and $\mincut_w(T)\geq (1+\varepsilon)\lambda$ when it terminates.
By \Cref{lem:focus path}, \textsc{ClearMultiplePaths} iteratively applies \textsc{ExtractAndFocus} until $\mincut_w(P)\geq (1+\varepsilon)\lambda$ for all $P\in \mathcal{P}$.
At this point, every cut that 1-or-2-respects $T$ on some path in $\mathcal{P}$ weights at least $(1+\varepsilon)\lambda$.
So, it is left to show that $\tilde B$ is the set of cuts that 2-respect $T$ on different paths of $\mathcal{P}$ with weight less than $(1+\varepsilon)\lambda$.

Let $s$ be such a cut.
Let $\{e\} := \cut_T(s)\cap E(P)$ and $\{f\} := \cut_T(s)\cap E(Q)$ for some $P,Q\in \mathcal{P}$.
Then, either $e$ and $f$ are cross-interested in each other, $e$ is down-interested in $f$, or $f$ is down-interested in $e$.
Therefore, $(P,Q)$ is an interested path pair, $e\in E_{P\to Q}$ and $f\in E_{Q\to P}$.
Let $v$ be the root of the path minor $T(P,Q)$.
Because of \textsc{ClearMultiplePaths}, every 2-respecting cut of $T(P,Q)$ that does not contain $v$ has weight at least $(1+\varepsilon)\lambda$.
Hence, $s\in S(P,Q)\subseteq \tilde B$ by \Cref{lem:small 2-respecting cuts} as required.

Next, we bound the running time of \Cref{alg:tree}.
The path decomposition $\mathcal{P}$ of $T$ can be computed in $O(n\log n)$ work and $O(\log^2 n)$ depth using \Cref{lem:bough_decomposition}.
\textsc{ClearMultiplePaths} runs in $\tilde{O}(m/\varepsilon^2)$ work and $\tilde O(1/\varepsilon^2)$ depth by \Cref{lem:focus path}.
All interested path pairs $(P,Q)$ and their edge sets $E_{P\to Q}, E_{Q\to P}$ can be found in $O(m\log m + n\log^3 n)$ work and $O(\log^2 n)$ depth using \Cref{lem:interested_pairs}.
Recall that the total number of edges in these sets is $O(n\log n)$.
Hence, by \Cref{lem:small 2-respecting cuts}, the parallel runs of \textsc{ExtractCutsInPath} takes $\tilde O(n)$ work and $\poly\log(n)$ depth.
We remark that the path minors $T(P,Q)$ need not be explicitly formed.
From \Cref{lem:small 2-respecting cuts}, we also know that the number of returned cuts is $|\tilde B| = O(n\log n)$.
Therefore, applying $\mathcal{D}_T.\textsc{Focus}$ on $\tilde B$ takes $\tilde O(m/\varepsilon^2)$ work and $\tilde O(1/\varepsilon^2)$ depth according to \Cref{lem:cut oracle ds}.
\end{proof}

\subsection{MWU Cut Oracle} \label{sec:mwu oracle}

In this section, we prove \Cref{lem:cut oracle ds}. Given a spanning tree $T$ of a graph $G = (V,E)$, we construct the cut oracle using
a basic data structure which is called \textit{canonical cuts}. Canonical cuts are based on the standard techniques including range trees where the ordering of vertices is defined by the Euler tour of a spanning tree $T$. It is a standard fact that an Euler tour of a spanning tree can be computed in nearly linear work and $O(\log |V|)$ depth~\cite{AtallahV84}. The rest of the construction of canonical cuts follows from \cite{CQ17}.

\begin{lemma} [Canonical Cuts~\cite{CQ17}] \label{lem:canonical cuts}
    Given a spanning tree $T$ of a graph $G= (V,E)$ to preprocess, we can construct a data structure $\mathcal{B}_T$ that maintains a family of non-empty edge-set (called \textit{canonical cuts}) $\mathcal{K}_T = \{F_1,\ldots F_{\ell}\}$  where $F_i \subseteq E$ for all $i$ and every edge $e \in E$ is contained at most $O(\log ^2|V|)$ canonical cuts. The data structure $\mathcal{B}_T$ supports the following operation:
    \begin{itemize}
        \item $\mathcal{B}_T.\textsc{Decompose}(s \in \mathcal{K}_T$) where the input is a 1-or-2-respecting cut of $T$: It returns a disjoint union of edge-sets $F_1,\ldots F_k \subseteq E$ that form $\cut_T(s)$, i.e, $\cut_T(s) = \bigsqcup_{i \leq k}F_k$ and $k \leq \poly\log |V|$ where $\sqcup$ denotes the disjoint union operation. This operation runs in $\poly\log |V|$ work.
    \end{itemize}
    The processing algorithm takes $\tilde O(|E|)$ work and $\poly\log |V|$ depth.
\end{lemma}

We are ready to prove \Cref{lem:cut oracle ds}.
\begin{proof} [Proof of \Cref{lem:cut oracle ds}]
Given a spanning tree $T$ of a graph $G = (V,E)$, we describe the construction of the cut-oracle data structure $\mathcal{D}_T$ and the implementation of each operation.

\paragraph{Preprocessing.} We construct the data structure $\mathcal{B}_T$ for the canonical cuts $\mathcal{K}_T$
 using \Cref{lem:canonical cuts}. Given the set of canonical cuts, we construct a \textit{base graph} $H_T$, a \textit{cut-edge incident graph} defined as follows. The base graph $H_T = (X,Y,E_H)$ is a bipartite graph where the left partition $X := \mathcal{K}_T$  is the set of canonical cuts, and the right partition $Y := E$ is the set of edges. For all $F \in \mathcal{K}_T, e \in E$, we add an edge $(F,e) \in E_H$ if and only if $e \in F$. For each $F \in \mathcal{K}_T$ we store the weighted sum $w(F) = \sum_{e \in F}w(e)$ and we always update $w(F)$ whenever one of its edges' weight is changed.

\paragraph{Data Structure Operations.} We now describe the implementation of operations in the lemma and also the helper operation $\mathcal{D}_T.\textsc{Augment}(B)$.
\begin{itemize}
    \item $\mathcal{D}_T.\textsc{CutValue}(s \in \mathcal{C}_T):$ Given $s$, a 1-or-2-respecting cut of $T$, return $w(\cut_T(s))$.
    \begin{itemize}
        \item This operation can be implemented in $\poly\log |V|$ work as follows. Let $\cut_T(s) = \bigsqcup_{i \leq k}F_i$ obtained from $\mathcal{B}_T.\textsc{Decompose}(s)$ (\Cref{lem:canonical cuts}). We return $\sum_{i \leq k}w(F_i)$ which takes $\tilde O(k) \leq \poly\log |V|$ work and parallel depth.
    \end{itemize}
    \item $\mathcal{D}_T.\textsc{Augment}(B)$:   Given a set of 2-respecting cuts in $T$, $\textsc{Augment}(B)$ operation updates the base graph $H_T$ to an \textit{augmented graph} $\hat H_{T,B}$ as follows. We add $B$ as the set of new vertices to $H_T$. For each $s \in B$, let $\cut_T(s) = \bigsqcup_{i \leq k}F_i$ obtained from $\mathcal{B}_T.\textsc{Decompose}(s)$ (\Cref{lem:canonical cuts}), and  for all $j \leq k \leq \poly\log |V|$, we add an edge $(s,F_j)$ to $H_T$. Here, we can assume an efficient representation of the set $F_j$ for all $j$ given by the construction in \Cref{lem:canonical cuts}. \begin{itemize}
        \item Since $k \leq \poly\log |V|$, this operation can be done in $\tilde O(|B|)$ work and $\poly\log |V|$ depth, and also each node $s$ has degree $k \leq \poly\log |V|$.
    \end{itemize}

    \item $\mathcal{D}_T.\textsc{Focus}_{G,T,\lambda,\varepsilon}(B)$: implement $\textsc{Focus}_{A,\lambda,\varepsilon}(B,w,\congestion)$ operation (\Cref{alg:mwu_focus}) without returning the output.

    \begin{itemize}
        \item This is proved in \Cref{claim:efficient update pd}.
    \end{itemize}

\end{itemize}
\begin{claim} \label{claim:efficient update pd}
     $\mathcal{D}_T.\textsc{Focus}(B)$ takes $\tilde O((|E|+|B|)\log(|B|)/\varepsilon^2)$ work and $\tilde O(\log(|B|)/\varepsilon^2)$ depth.
\end{claim}
\begin{proof}
We first call $\mathcal{D}_T.\textsc{Augment}(B)$ to obtain the augmented graph $\hat H_{T,B}$.  In every iteration of \Cref{alg:mwu_focus}, the bottleneck is to compute $\delta$ that $\|Ag\|_\infty = \varepsilon$. It is enough to compute an edge with the highest congestion increase. For each edge $e \in E$, denote $\Gamma(e) = \{ s \in B \colon e \in \cut_T(s)\}$ be the set of 2-respecting cuts in $B$ whose cut-set contains $e$. Using the augmented graph, the set $\Gamma(e)$ is the set of vertices in $B$ that is reachable by $e$ in the augmented graph. So we set $\delta = \varepsilon \cdot  (\max_{e \in E} \sum_{s \in \Gamma(e)} x_s/c_e)^{-1}$ and obtain $\|Ag\|_\infty = \varepsilon$ as desired. Note that we can compute $e^* := \arg \max_{e \in E} \sum_{s \in \Gamma(e)} x_s/c_e$ in $\tilde O(|E|)$ work and $\poly \log |V|$ depth using the augmented graphs. The rest of the steps can be done in $\tilde O(|E|+|B|)$ work and $\poly\log (|V|)$ depth.

Finally, since \Cref{alg:mwu_focus} terminates in $\tilde O(\log(|B|)/\varepsilon^2)$ iterations, the work and depth bounds follow.
\end{proof}
\end{proof}

\section{Approximating the $k$-ECSS LP} \label{sec:kECSS_LP}

Given an undirected graph $G = (V,E)$ with nonnegative edge costs $c\in \R^m_{\geq 0}$ and an integer $k\geq 1$, the \emph{$k$-edge-connected spanning subgraph (\kECSS)} problem asks to find a subgraph $H\subseteq G$ which is $k$-edge-connected, spans $V$ and minimizes $c(H):= \sum_{e\in E(H)} c_e$.
The natural LP relaxation is given by
\begin{equation}\label{sys:kECSS_LP}
\begin{aligned}
    \min &\;\; c^\top y \\
    \text {s.t.} &\; \sum_{e \in \delta(S)} y_e \geq k \quad\; \forall\, \emptyset \subsetneq S \subsetneq V \\
    &\;\; 0\leq y_e \leq 1 \qquad \forall\, e \in E.
\end{aligned}
\end{equation}
This is not a covering LP due to the packing constraints $y\leq \1$.
However, they can be replaced with \emph{Knapsack Covering (KC)} constraints~\cite{CarrFLP00} to obtain a covering LP:

\begin{equation}\label{sys:kECSS_LP_KC}
\begin{aligned}
    \min &\;\;  c^\top y \\
    \text {s.t.} &\; \sum_{e\in \delta(S)\setminus F} y_e \geq k - |F| \quad \forall\, \emptyset \subsetneq S \subsetneq V, F\subseteq \delta(S), |F| < k \\
    &\;\; y_e\geq 0 \qquad\qquad\qquad\quad \forall\, e\in E.
\end{aligned}
\end{equation}

For every $\emptyset\subsetneq S\subsetneq V$ and $F\subseteq \delta(S)$ with $|F|<k$, the KC constraints enforce the solution to (fractionally) use at least $k-|F|$ of the remaining edges $\delta(S)\setminus F$.

\begin{lemma}[{\cite[Lemma 17]{ChalermsookHNSS22}}]
Every feasible solution to \eqref{sys:kECSS_LP_KC} is also feasible to \eqref{sys:kECSS_LP}.
Conversely, for every feasible solution $y$ to \eqref{sys:kECSS_LP}, there exists a feasible solution $y'$ to \eqref{sys:kECSS_LP_KC} such that $y'\leq y$.
\end{lemma}

Since $c\geq \0$, this immediately implies that \eqref{sys:kECSS_LP} and \eqref{sys:kECSS_LP_KC} are equivalent.
In particular, every optimal solution to \eqref{sys:kECSS_LP_KC} is also an optimal solution to \eqref{sys:kECSS_LP}.

We will apply the epoch-based MWU method to \eqref{sys:kECSS_LP_KC}.
Recall that for the Cut Covering LP, a minimum weight column of $A$ corresponds to a minimum cut.
For \eqref{sys:kECSS_LP_KC}, a minimum weight column of $A$ corresponds to a \emph{minimum normalized free cut}.

\begin{definition}
Let $G$ be a graph with edge weights $w\in \R^m_{\geq 0}$ and let $k\geq 1$ be an integer.
A \emph{free cut} is a pair $(S,F)$ consisting of a subset of vertices $\emptyset \subsetneq S \subsetneq V$ together with a subset of edges $F\subseteq \delta(S)$ where $|F|< k$.
The edges in $F$ are called the \emph{free edges} of $(S,F)$.
We denote the set of free cuts as $\mathcal{F}$.
A \emph{minimum normalized free cut} is a free cut which minimizes its \emph{normalized weight}:
\[\min_{(S,F)\in \mathcal{F}} \frac{w(\delta(S)\setminus F)}{k-|F|}.\]
\end{definition}

In an epoch-based MWU method, for every iteration $t$, the set $B^{(t)}$ corresponds to the set of free cuts with normalized weight less than $(1+\varepsilon)\lambda^{(t)}$.

For $\rho\in \R$, let $w_\rho$ denote the edge weights obtained by truncating the edge weights larger than $\rho$ to $\rho$, i.e.
\[w_\rho(e) = \min\{w(e), \rho\} \qquad \forall\, e\in E.\]
For $F\subseteq E$, we denote $E^{\geq \rho}_w(F):=\{e\in F:w(e)\geq \rho\}$ as the subset of edges in $F$ with weight at least $\rho$.
The following theorem establishes a connection between the normalized weight of a free cut with respect to $w$, and the weight of a cut with respect to $w_\rho$.

\begin{theorem}[{\cite[Range Mapping Theorem]{ChalermsookHNSS22}}]\label{thm:range_mapping}
Let $G$ be a graph with edge weights $w\in \R^m_{\geq 0}$ and let $k\geq 1$ be an integer.
Let $\lambda >0$ and $\rho = (1+\varepsilon)\lambda$.
\begin{enumerate}
    \item If the minimum normalized weight of a free cut lies in $[\lambda, \rho)$, then the value of a minimum cut in $(G,w_\rho)$ lies in $[k\lambda, k\rho)$.
    \item For any cut $\delta(S)$ where $w_\rho(\delta(S))< k\rho$, we have
    \[\frac{w(\delta(S)\setminus E^{\geq \rho}_w(\delta(S)))}{k-|E^{\geq \rho}_w(\delta(S))|} < \rho.\]
\end{enumerate}
\end{theorem}

\cite{ChalermsookHNSS22} used it within Fleischer's sequential MWU method to develop a nearly linear time algorithm for approximating \eqref{sys:kECSS_LP_KC}.
We will use it within \Cref{alg:mwu_parallel_general} to develop a parallel algorithm for approximating \eqref{sys:kECSS_LP_KC} in nearly linear work and polylogarithmic depth.

Given a spanning tree $T$ of $G$, let $s\in \mathcal{C}_T$  be a 1-or-2-respecting cut such that $w_\rho(\cut_T(s))<k\rho$.
Clearly, $|E^{\geq \rho}_w(\cut_T(s))|<k$.
We define
\[f^\rho_w(s) := (\shore_T(s),E^{\geq \rho}_w(\cut_T(s)))\]
as the free cut obtained by designating $E^{\geq \rho}_w(\cut_T(s))$ as the free edges.
We also denote $k^\rho_w(s) := k - |E^{\geq \rho}_w(\cut_T(s))|$ for convenience.
Similarly, for a subset $B\subseteq \mathcal{C}_T$ of 1-or-2-respecting cuts where $w_\rho(\cut_T(s))<k\rho$ for all $s\in B$, we define
\[f^{\rho}_w(B):=\{f^{\rho}_w(s):s\in B\}.\]
When $\rho$ is clear from context, we will drop the superscript and just write $f_w$ and $k_w$.

\subsection{Warm-up}
In this section, we present an algorithm with $O(k)$-factor overhead, and we explain how to improve to $O(\log k)$-factor in the next section. The following data structure is an analogue of the MWU Cut Oracle (\Cref{lem:cut oracle ds}).

\begin{lemma} [MWU Free Cut Oracle] \label{lem: free cut oracle ds}
    Let $G$ be a graph with edge costs $c\in \R^m_{\geq 0}$, edge weights $w\in \R^m_{\geq 0}$ and congestion $\congestion\in \R^m_{\geq 0}$.
    Let $k\geq 1$ be an integer and $\lambda,\varepsilon>0$.
    Given a spanning tree $T$ of $G$, there is a data structure $\mathcal{D}'_{T}$ which supports the following operations:
    \begin{itemize}
        \item $\mathcal{D}'_{T}.\textsc{CutValueTruncated}(s)$: Given a 1-or-2-respecting cut $s\in \mathcal{C}_T$, return its truncated weight $w_\rho(\cut_T(s))$ in $\poly(\log n)$ work and depth, where $\rho:=(1+\varepsilon)\lambda$.
        \item $\mathcal{D}'_{T}.\textsc{Focus}(B)$:
        Let $A^\top$ be the constraint matrix of \eqref{sys:kECSS_LP_KC} when it is expressed as \eqref{lp:cover}, i.e.,
        \[A_{e,(S,F)} = \begin{cases} \frac{1}{(k-|F|)c_e}, &\text{ if }e\in \delta(S)\setminus F, \\ 0, &\text{ otherwise.} \end{cases}\]
        Given a set of 1-or-2-respecting cuts $B\subseteq \mathcal{C}_T$, implement $\textsc{Focus}_{A,\lambda,\varepsilon}
        (f^{\rho}_w(B),w,\congestion)$
        in $\tilde O((m+|B|)\log(|B|)/\varepsilon^2)$ work and $\tilde O(\log(|B|)/\varepsilon^2)$ depth.
    \end{itemize}
The data structure $\mathcal{D}'_T$ can be constructed in $\tilde O(m)$ work and $\tilde O(1)$ depth.
\end{lemma}

\begin{proof}
First, we show how to construct $\mathcal{D}'_{T,\lambda,\varepsilon}$.
We start by constructing the data structure $\mathcal{B}_T$ for canonical cuts using \Cref{lem:canonical cuts}, which takes $\tilde O(m)$ work and $\tilde O(1)$ depth.
Then, we construct the MWU Cut Oracle $\mathcal{D}_T$ using \Cref{lem:cut oracle ds}, which also takes $\tilde O(m)$ work and $\tilde O(1)$ depth.
It constructs the base graph $H_T$ on the bipartition $(\mathcal{K}_T,E)$, where $\mathcal{K}_T$ is the set of canonical cuts given by $\mathcal{B}_T$, and $E$ is the edge set of our input graph $G$.
For every $K\in \mathcal{K}_T$ and $e\in E$, $Ke\in E(H_T)$ if and only if $e\in K$.

Let $\rho:= (1+\varepsilon)\lambda$.
Every edge $e\in E$ maintains its weight $w(e)$ and congestion $\congestion(e)$.
Every canonical cut $K\in \mathcal{K}_T$ maintains its
truncated weight $w_\rho(K) := \sum_{e\in K}w_\rho(e)$.
Since $d_{H_T}(e) = O(\log^2 n)$ for all $e\in E$, this can be done in $\tilde O(m)$ work and $\tilde O(1)$ depth.

Next, we show how to return the truncated weight of a 1-or-2-respecting cut of $T$.
Let $s\in \mathcal{C}_T$.
We call $\mathcal{B}_T.\textsc{Decompose}(s)$ to partition $\cut_T(s)$ into $\tilde O(1)$ canonical cuts $K_1,K_2,\dots, K_\ell$, which takes $\tilde O(1)$ work and depth.
Here, we can assume a succinct representation of the canonical cuts given by \Cref{lem:canonical cuts}.
Then, $w_\rho(\cut_T(s)) = \sum_{i=1}^\ell w_\rho(K_i)$.
Evaluating this sum takes $\tilde O(1)$ work and depth.

It is left to show how to implement $\textsc{Focus}_{A,\lambda,\varepsilon}(f_w(B),w,\congestion)$, given any subset of 1-or-2-respecting cuts $B\subseteq \mathcal{C}_T$.
We call $\mathcal{D}_T.\textsc{Augment}(B)$, which augments the base graph $H_T$ to the augmented graph $\hat H_{T,B}$ in $\tilde O(B)$ and $\tilde O(1)$ depth.
In the augmented graph, we have an additional node for every $s\in B$, and $sK\in E(\hat H_{T,B})$ if and only if $K$ belongs to the decomposition of $s$ given by $\mathcal{B}_T.\textsc{Decompose}(s)$.
For every $s\in B$, let $s' := f_w(s)$ be the corresponding free cut with respect to the current weights $w$, and let $F_s := E^{\geq \rho}_w(\cut_T(s))$ be the free edges of $s'$.
For every $K\in \mathcal{K}_T$, we also let $F_K:= E^{\geq \rho}_w(K)$ be the edges in $K$ whose current weight exceeds $\rho$.
Every $s\in B$ maintains the packing variable $x_{s'}$, which is initialized to 0 as per \Cref{alg:mwu_focus}.
Every $K\in \mathcal{K}_T$ maintains the partial weight $w(K\setminus F_K)$ and the constant $|F_K|$.
The latter can be (implicitly) done in $\tilde O(m)$ work and $\tilde O(1)$ depth because $d_{H_T}(e) = O(\log^2 n)$ for all $e\in E$.

In the first iteration of \Cref{alg:mwu_focus}, we set $g$ as
\[g_{s'} := \frac{k-|F_s|}{|B|}\cdot\varepsilon = \frac{k-\sum_{K:sK\in E(\hat H_{T,B})}|F_K|}{|B|}\cdot\varepsilon \qquad \forall s\in B.\]
It can be computed in $\tilde O(|B|)$ work and $\tilde O(1)$ depth using the augmented graph because $d_{\hat H_{T,B}}(s) = \poly\log n$ for all $s\in B$.

In subsequent iterations of \Cref{alg:mwu_focus}, we set $g_{s'}:=\delta x_{s'}$ for all $s\in B$ whose corresponding free cut $s'$ has normalized weight less than $\rho$, where $\delta$ is chosen such that $\|Ag\|_\infty = \varepsilon$.
For every $s\in B$, the normalized weight of $s'$ can be computed as
\[\frac{w(\cut_T(s)\setminus F_s)}{k - |F_s|} = \frac{\sum_{K:sK\in E(\hat H_{T,B})} w(K\setminus F_K)}{k - \sum_{K:sK\in E(\hat H_{T,B})}|F_K|}.\]
Evaluating the latter ratio takes $\tilde O(1)$ work and depth using the augmented graph because $d_{\hat H_{T,B}}(s) = \poly\log n$.
To compute $\delta$, it suffices to consider $E\setminus (\cup_{s\in B}F_s)$ because the congestion of edges in $\cup_{s\in B}F_s$ do not change.
For every edge $e\in E\setminus (\cup_{s\in B}F_s)$, let
\[\Gamma_w(e):=\left\{s\in B:e\in \cut_T(s),\frac{w(\cut_T(s)\setminus F_s)}{k-|F_s|} < \rho \right\}.\]
They correspond to $s\in B$ which can reach $e$ on a path of length $2$ in $\hat H_{T,B}$, and whose corresponding free cut $s'$ has normalized weight less than $\rho$.
Using this definition, the increase in congestion can be written as
\[(Ag)_e = \frac{\delta}{c_e} \sum_{s\in \Gamma_w(e)} \frac{x_{s'}}{k-|F_s|} =: \frac{\delta}{c_e} \cdot z_e\]
The values $\{z(e):e\in E\setminus (\cup_{s\in B}F_s)\}$ can be computed in $\tilde O(|B| + m)$ work and $\tilde O(1)$ depth using the augmented graph because $d_{\hat H_{T,B}}(s),d_{\hat H_{T,B}}(e) = \poly\log n$ for all $s\in B$ and $e\in E$.
It follows that $\delta = \varepsilon \min_{e\in E^{<\rho}_e} c_e/z_e$.

Once $\delta$ and $g$ is computed, we can easily update $x$, $w$ and $\congestion$ in $\tilde O(|B|+m)$ work and $\tilde O(1)$ depth.
We can also update the partial weight $w(K\setminus F_K)$ for all canonical cuts $K$ in $\tilde O(m)$ work and $\tilde O(1)$ depth.

Since \Cref{alg:mwu_focus} terminates in $\tilde O(\log(|B|)/\varepsilon^2)$ iterations by \Cref{lem:iters_focus}, the work and depth bounds follow.
Finally, when \Cref{alg:mwu_focus} terminates, we update the truncated weight $w_\rho(K)$ for all canonical cuts $K$ in $\tilde O(m)$ work and $\tilde O(1)$ depth.
\end{proof}

We now have the necessary tools to prove a weaker version of \Cref{thm:parallel_kECSS}.
Compared to \Cref{thm:parallel_kECSS}, it has an extra factor of $k$ in the work and depth bounds.

\begin{theorem}\label{thm:parallel_kECSS_weak}
Let $G$ be an undirected graph with $n$ nodes, $m$ edges and edge costs $c\in \R^m_{>0}$.
For every $0<\varepsilon < 0.5$, there is a randomized parallel algorithm that computes a $(1+\varepsilon)$-approximate solution to the \kECSS~LP with high probability.
The algorithm runs in $\tilde O(km/\varepsilon^4)$ work and $\tilde O(k/\varepsilon^4)$ depth.
\end{theorem}

\begin{proof}
Fix an epoch of \Cref{alg:mwu_parallel_general}.
By \Cref{lem:epochs}, it suffices to show how to clear the epoch in $\tilde O(km/\varepsilon^2)$ work and $\tilde O(k/\varepsilon^2)$ depth.
Let $\lambda$ be the lower bound used in this epoch, and denote $\rho := (1+\varepsilon)\lambda$.
Our goal is to apply \textsc{Focus} iteratively until the minimum normalized weight of a free cut is at least $\rho$.

At the start of the epoch, we invoke \Cref{thm:tree packing} to obtain a set $\mathcal{T}$ of $O(\log n)$ spanning trees such that with high probability, every $(1+\varepsilon)$-minimum cut with respect to  $w_\rho$ 1-or 2-respects some tree in $\mathcal{T}$.
Since $\lambda$ is a lower bound on the minimum normalized weight of a free cut with respect to $w$, by \Cref{thm:range_mapping}, $k\lambda$ is a lower bound on the minimum weight of a cut with respect to $w_\rho$.
It follows that with high probability, every cut whose weight with respect to $w_\rho$ lies in $[k\lambda, k\rho)$ 1-or-2-respects some tree in $\mathcal{T}$.
Hence, our goal reduces to applying \textsc{Focus} iteratively until $\mincut_{w_\rho}(T) \geq k\rho$ for all $T\in \mathcal{T}$.
This is because by \Cref{thm:range_mapping}, the minimum normalized weight of a free cut with respect to $w$ is at least $\rho$ with high probability.

Fix a tree $T\in \mathcal{T}$.
We construct the data structure $\mathcal{D}'_T$ and follow the template of the tree algorithm in the previous section, with the following two differences.
First, instead of feeding cuts $B\subseteq \mathcal{C}_T$ where $w(\cut_T(s))<\rho$ for all $s\in B$ to $\mathcal{D}_T.\textsc{Focus}$, we feed cuts $B\subseteq \mathcal{C}_T$ where $w_\rho(\cut_T(s))<k\rho$ for all $s\in B$ to $\mathcal{D}'_T.\textsc{Focus}$.
According to \Cref{thm:range_mapping}, for every cut $s\in B$, the corresponding free cut $f_w(s)$ has normalized weight less than $\rho$.
Hence, $f_w(B)\subseteq B^{(t)}$.

When $\mathcal{D}'_T.\textsc{Focus}(B)$ terminates, let $B_1\subseteq B$ be the subset of cuts whose truncated weight is still below $k\rho$, i.e.,

\[B_1 := \left\{s\in B:w_\rho(\cut_T(s))<k\rho\right\}.\]
Observe that for every $s\in B_1$, the set $E^{\geq \rho}_w(\cut_T(s))$ must have grown.
Otherwise, the free cut $f_w(s)$ remains the same, so it has normalized weight at least $\rho$ by the definition of \textsc{Focus}, which contradicts \Cref{thm:range_mapping}.
We feed $B_1$ to $\mathcal{D}'_T.\textsc{Focus}$, define $B_2$ analogously, and repeat this process until $B_i = \emptyset$.
Note that $i< k$ because $|E^{\geq \rho}_w(s)|<k$ for all $s\in B_i$.

The correctness of the algorithm follows from \Cref{lem: free cut oracle ds} and the correctness of the tree algorithm.
The work and depth bounds also follow with an extra factor of $k$ because $\mathcal{D}_T$ and $\mathcal{D}'_T$ have identical complexity.
\end{proof}

Note that this implementation suffers a factor $O(k)$ overhead in work and depth because we need to call $\mathcal{D}'_T.\textsc{Focus}$ repeatedly for at most $k$ iterations. In the next section, we present an improved algorithm that reduces the iteration complexity to approximately $O(\log k)$.

\subsection{Accelerating \textsc{Focus}}
In this section, we prove \Cref{thm:parallel_kECSS}.
The main idea is to modify \Cref{alg:mwu_focus} to exploit the structure of the LP~\eqref{sys:kECSS_LP_KC}. Let $G$ be a graph and let $A^\top$ be the constraint matrix of \eqref{sys:kECSS_LP_KC} when it is expressed as \eqref{lp:cover}, i.e.,
\[A_{e,(S,F)} = \begin{cases} \frac{1}{(k-|F|)c_e}, &\text{ if }e\in \delta(S)\setminus F, \\ 0, &\text{ otherwise.} \end{cases}\]
Fix a spanning tree $T$ of $G$ and scalars $\lambda \geq 0$, $\varepsilon>0$, $\rho=(1+\varepsilon)\lambda$.

Suppose are given a set of 1-or-2-respecting cuts $\tilde C\subseteq \mathcal{C}_T$ such that $w_\rho(\cut_T(s))<k\rho$ for all $s\in \tilde C$. Our goal is to apply MWU on the free cuts associated with $\tilde C$ until $w_\rho(\cut_T(s))\geq k\rho$ for all $s\in \tilde C$. In the previous section, we achieved this by applying \textsc{Focus} on the free cuts $f_w(\tilde C)$ for at most $k$ times.

To accelerate \textsc{Focus}, we make the following 3 modifications.
\begin{enumerate}
    \item For every $s\in \tilde C$, the corresponding free cut $f_w(s)$ is updated in every iteration.
    \item The algorithm is divided into $\ceil{\log k}$ phases, based on the value of $k_w(s)$. At the start of every phase, the packing variables $x_{f_w(s)}$ with value 0 are initialized. For $q = 1,2,\dots,\ceil{\log k}$, phase $q$ ends when $k_w(s)<k/2^q$ for all $s\in \tilde C$.
    \item  Given two free cuts $(S_1, F_1)$ and $(S_2,F_2)$, we denote $(S_1,F_1)\subseteq (S_2,F_2)$ if $S_1 = S_2$ and $F_1\supseteq F_2$. In every iteration, the increment $g_{f_w(s)}$ is calculated relative to $\sum_{\ell\in B_{\text{prev}}:  f_w(s)\subseteq \ell}x_{\ell}$, instead of $x_{f_w(s)}$. The terms in this sum correspond to the free cuts associated with $s$ that were previously encountered during this phase.
    \begin{itemize}
    \item  The intuition is as follows. For any fixed cut $S$, instead of treating the free-cut variables $\{(S,F): F \subseteq \delta(S), |F| <k\}$ separately, we view them as if they form a single variable. In this way, we avoid initializing the variables $O(k)$ times.
    \item Intuitively speaking, for any fixed $S$, we may define a new variable $z_S := \sum_{F \subseteq \delta (S): |F| <k} x_{S,F}$, and simulate the update $\delta \cdot z_S = \delta \cdot \sum_{F \subseteq \delta (S): |F| <k} x_{S,F}$. We cannot implement this directly, however, because the dimension is too large, and we do not know in advance which support of $x$ will be used.
    \item Instead, for a fixed $S$, the sequence of free-cuts $(S,F_1), (S,F_2), \ldots (S,F_b)$ is produced update. When $(S,F_j)$ is first introduced, we use the sum $\sum_{i < j} x_{S,F_i}$ as its initial value. If $(S,F_j)$ is involved in the subsequent update, then we use the sum $\sum_{i \leq j} x_{S,F_i}$.

    \end{itemize}
\end{enumerate}

The algorithm is summarized in \Cref{alg:fast_focus_new}.  By subdividing into $O(\log k)$ phases and the fact the increment vector on $(S,F)$ uses cumulative sum of the variables $x_{(S,F')}$ over $F'$, we improve the iteration complexity to roughly $O(\log k)$. In each phase $q$, \Cref{alg:kECSSFocus} is similar to the $\textsc{Focus}$ algorithm, but we restrict the set of cuts to ``clear" to the set of 1-or-2-respecting cuts $s \in \tilde C$ whose free-cut $f_w(s)$ has $k_w(s) \geq k/2^q.$

\begin{algorithm}[htb!]
  \caption{\textsc{FastFocus}$_{G,\lambda,\varepsilon,k}(\tilde C,w,\congestion,\rho)$}
  \label{alg:fast_focus_new}
  \SetKwInOut{Input}{Input}
  \SetKwInOut{Output}{Output}
  \SetKwComment{Comment}{$\triangleright$\ }{}
  \SetKw{And}{\textbf{and}}
  \SetKw{Or}{\textbf{or}}
  \Input{Graph $G$ with edge weights $w\in \R^m_{\geq 0}$ and edge congestion $\congestion \in \R^m_{\geq 0}$, integer $k\in \N$, lower bound $\lambda \geq 0$, accuracy parameter $\varepsilon>0$, upper bound $\rho=(1+\varepsilon)\lambda$, subset $\tilde C\subseteq \{s\in  \mathcal{C}_T: w_\rho(\cut_T(s))<k\rho\}$ where $T$ is a spanning tree of $G$.}
  \Output{Edge weights $w'\in \R^m_{\geq 0}$ and edge congestion $\congestion' \in \R^m_{\geq 0}$}
    $\eta \gets \ln(m)/\varepsilon$\;

    \For{$q = 1, 2, \dots, \ceil{\log k}$}{

    \uIf{$\|\congestion\|_\infty < \eta$} {
    $\tilde B \gets \{f_w(s):s\in \tilde C, k_w(s) \geq k/2^q\}$\;
    $(w,\congestion) \gets \textsc{kECSSFocus}(\tilde B,\tilde C, w,\congestion,q)$\;
    $\tilde C \gets \{s\in \tilde C: w_\rho(\cut_T(s))<k\rho\}$\;
    }

    }
    \Return $(w,\congestion)$\;

\end{algorithm}

\begin{algorithm}[htb!]
  \caption{\textsc{kECSSFocus}$(\tilde B, \tilde C, w,\congestion,q,\rho)$}
  \label{alg:kECSSFocus}
  \SetKwInOut{Input}{Input}
  \SetKwInOut{Output}{Output}
  \SetKwComment{Comment}{$\triangleright$\ }{}
  \SetKw{And}{\textbf{and}}
  \SetKw{Or}{\textbf{or}}
  \Input{ }
  \Output{}
    $x \gets \0_N $ \Comment*[r]{$N$ is the number of free cuts in $G$}

    $\tilde B_{\text{prev}} \gets \emptyset$\;
    \While{$\|\congestion\|_\infty < \eta \mbox{ and } \tilde B \neq \emptyset$}{
        $\tilde B_0 \gets \{j\in \tilde B: x_j =0 \mbox{ and } \shore_T(j) \not\in \bigl \{{ \shore_T(i) \colon i \in \tilde B_{\text{prev}}}\}\bigr\}$\;
        \uIf(\Comment*[f]{Only happens in the 1st iteration}){$\tilde B_0\neq \emptyset$}{
          Set $g_j \gets \varepsilon/(|\tilde B_0| \max_{i\in [m]} A_{i,j})$ for all $j\in \tilde B_0$, and $g_j \gets 0$ for all $j\notin \tilde B_0$ \;
        }
        \Else{
          Set $g_j \gets \delta \sum_{\ell\in \tilde B_{\text{prev}}:j\subseteq \ell}x_\ell$ for all $j\in \tilde B$, and $g_j \gets 0$ for all $j\notin \tilde B$, where $\delta$ is chosen such that $\|Ag\|_\infty = \varepsilon$ \;
        }
        $x \gets x + g$\;
        $w \gets w\circ (\1+Ag)$\;
        $\congestion \gets \congestion + Ag$\;
        $\tilde B_{\text{prev}} \gets \tilde B_{\text{prev}} \cup \tilde B$\;
        $\tilde C \gets \{s\in \tilde C: w_\rho(\cut_T(s))<k\rho\}$\;
        $\tilde B \gets \{f_w(s): s \in \tilde C, k_w(s)\geq k/2^q\}$\;
    }
    \Return  $(w,\congestion).$
\end{algorithm}

For correctness, observe that \Cref{alg:fast_focus_new} computes $O(\log k)$ iterations of core-sequences in \Cref{alg:mwu_focus}.
\begin{claim} \label{lem:correct_kecss_focus}
 For every iteration $q$ in \Cref{alg:fast_focus_new}, \textsc{kECSSFocus} computes a core-sequence in \Cref{alg:mwu_focus}.
\end{claim}
\begin{proof}
    It is enough to verify that every free cut $(S,F)$ in the support of the update vector $g$,     \[\frac{w(\delta(S)\setminus F))}{k-F} < \rho.\] Indeed, by design, $F = E^{\geq \rho}_w(\delta(S))$. By applying the range mapping theorem (\Cref{thm:range_mapping}) with the fact that $w_{\rho}(\delta(S)) < k\rho$ (c.f. the definition of $\tilde C$), we have $\frac{w(\delta(S)\setminus F))}{k-F} < \rho.$ Finally, the update vector is guaranteed not to overshoot the step size as in $\|Ag \|_{\infty} = \eps.$
\end{proof}

Next, we present the iteration complexity, whose proof can be found in \Cref{sec:itr complex}. Let $c_{\max} := \max_e c_e$ and $c_{\min} := \min_e c_e.$
\begin{lemma}\label{lem:kECSSFocus}
  For every iteration $q$ in \Cref{alg:fast_focus_new}, \textsc{kECSSFocus} (\Cref{alg:kECSSFocus}) terminates in $  O\bigg(\frac{\log m}{\varepsilon^2}\log \big ( \frac{|\tilde B| \cdot  \log m}{\varepsilon^2} \cdot \frac{c_{\max}}{c_{\min}}\big)\bigg)$  iterations. Therefore,  \Cref{alg:fast_focus_new} terminates in\\ $O\bigg(\log k\cdot (\frac{\log m}{\varepsilon^2}\log \big ( \frac{|\tilde C| \cdot  \log m}{\varepsilon^2} \cdot \frac{c_{\max}}{c_{\min}}\big) \bigg)$ iterations.
\end{lemma}

Given \Cref{lem:correct_kecss_focus} and \Cref{lem:kECSSFocus}, implementing the MWU Free Cut Oracle that supports the new \textsc{FastFocus} operation is a straightforward modification of \Cref{lem: free cut oracle ds}, and implementing the overall algorithm is similar to the implementation in the proof of \Cref{thm:parallel_kECSS_weak}.

\subsection{Iteration Complexity} \label{sec:itr complex}

We prove \Cref{lem:kECSSFocus} in this section. Throughout this section, we fix an iteration $q$ in \Cref{alg:fast_focus_new}. By induction on iteration $q$, we have the following invariant for every iteration in the loop of $\textsc{kECSSFocus}(\tilde B,\tilde C, w,\congestion,q)$:
\begin{align} \label{eq:forall_it_k_in_range}
    \forall s \in \tilde C, \text{if } f_w(s) \in \tilde B \text{ then }  k_w(s) \in [\frac{k}{2^q},\frac{k}{2^{q-1}}].
\end{align}

We focus on analyzing the iteration complexity of \Cref{alg:kECSSFocus}. Consider the final iteration. Let $s$ be a free cut in the final iteration, and let $\ell_1,\ldots,\ell_{b} = s$ be the longest sequence of free cuts that involved in the update procedure such that $\forall j,\shore_T(\ell_j) = \shore_T(s) := S$.  For each $j$, we denote free cut $\ell_j := (S,F_j)$. For all $j \leq b$, let $k_j := k_w(\ell_j)$ where $w$ is the weight at the iteration $\ell_j$ was updated. Observe that $ \sum_{\ell\in B_{\text{prev}}: f_w(s)\subseteq \ell}x_{\ell} = \sum_{j\leq b}x_{\ell_j}$ and $k_1 \geq k_2 \geq \ldots \geq k_b$. We prove lower and upper bounds of this sum. We start with a lower bound.
\begin{align} \label{eq:lowersum} \sum_{j\leq b}x_{\ell_j} \geq  x_{\ell_0} \geq \frac{\varepsilon}{|\tilde B|\max_{e \in \delta(S)\setminus F_1} \frac{1}{k_1c_e}}.
\end{align}
Next, we prove an upper bound.
\begin{lemma}
\begin{align} \label{eq:uppersum}
\sum_{j \leq b}^\ell x_{\ell_j} < \frac{\eta}{ \max_{e \in \delta(S)\setminus F_b}\frac{1}{k_1c_e}}.
\end{align}
\end{lemma}
\begin{proof}
Indeed, suppose otherwise.  Let $e^*$ be an arbitrary edge in $\arg\max_{e \in \delta(S)\setminus F_{b}} (k_1c_e)^{-1}$.
    \[ (Ax)_{e^*} \geq \sum_{j\leq b} \frac{x_{\ell_j}}{k_jc_{e^*}}\geq \frac{1}{k_1c_{e^*}} \sum_{j \leq b} x_{\ell_j} \geq \eta.\]
    This contradicts  $\|Ax\|_{\infty} < \eta.$
 \end{proof}

\begin{lemma} \label{lem:forall_if_B_then_delta_geq}
For every iteration, if $\tilde B_0 = \emptyset$ then $\delta \geq \frac{\varepsilon}{2\eta}. $
\end{lemma}
Before proving \Cref{lem:forall_if_B_then_delta_geq}, let us see why this implies the desired iteration complexity bound. By \Cref{lem:forall_if_B_then_delta_geq}, the sum  $\sum_{j \leq b}x_{\ell_j}$ is increased by a factor of $1+\delta \geq 1+\varepsilon/(2\eta)$ per iteration, By the upper and lower bounds of $\sum_{j \leq b}x_{\ell_j}$, the number of iterations is at most
\begin{align*}
b \overset{(\ref{eq:lowersum},\ref{eq:uppersum})}{\leq}  \bigg \lceil \log_{1+\eps/(2\eta)}(\eta/\varepsilon \cdot \frac{|\tilde B|\max_{e \in \delta(S)\setminus F_1}\frac{1}{k_1c_e}}{\max_{e \in \delta(S)\setminus F_b}\frac{1}{k_1c_e}} ) \bigg \rceil
= O\bigg(\frac{\log m}{\varepsilon^2}\log \big ( \frac{|\tilde B| \cdot  \log m}{\varepsilon^2} \cdot \frac{c_{\max}}{c_{\min}}\big)\bigg).
\end{align*}

We finish this section by presenting the proof of \Cref{lem:forall_if_B_then_delta_geq}.
\begin{proof} [Proof of \Cref{lem:forall_if_B_then_delta_geq}]
    Let $g$ be the increment vector of the iteration.  By design, there exists $e$ s.t. $(Ag)_e = \varepsilon$. It suffices to show that
    \begin{align} \label{eq:A_delta_x_e_geq_A_g_e}
    (A(\delta x))_e \geq \frac{\varepsilon}{2}
    \end{align}
    Indeed, assuming \Cref{eq:A_delta_x_e_geq_A_g_e}, we show  $\delta > \varepsilon/2\eta$ from
    \[ (Ax)_e < \eta \Rightarrow (A(\delta x))_e < \delta\cdot \eta \overset{(\ref{eq:A_delta_x_e_geq_A_g_e})}{\Rightarrow} \frac{\varepsilon}{2} < \delta \cdot \eta.\]
    We next prove \Cref{eq:A_delta_x_e_geq_A_g_e}. For any cut $S$, let $\texttt{seq}_{S}$ be the longest sequence, indexed by $j$,  of free cuts $\{(S,F_j)\}_j$ that involved in the update procedure so far.  Let $z(S)$ be the last index of the sequence $\texttt{seq}_{S}$. Let $k_S(j) = k_w(S,F_j)$ where $w$ is the weight at the iteration that $(S,F_j)$ was updated. Let $k^*_S = k_S(z(S))$ and $\texttt{seq}^*_S := \texttt{seq}_{S}(z(S)).$ Observe that\footnote{The summation over an empty set is defined to be zero.}
    \begin{align}\label{eq:gseq_eq_sum_xseq} \forall S,g_{\texttt{seq}^*_S} = \delta \cdot \sum_{j=1}^{z(S)} x_{\texttt{seq}_S(j)}. \end{align}

    Therefore,
    \begin{align*}
      (A(\delta x))_e &= \sum_{\emptyset \neq S \subsetneq V: e \in \delta(S)} \sum_{\overset{F \subseteq \delta(S)\setminus \{e\}}{ |F| < k}} A_{(S,F)}\delta x_{(S,F)} \\
      &\geq \sum_{S : e \in \delta(S)} \sum_{\overset{F: (S,F) \in \tilde B}{e \not \in F}} A_{(S,F)}\delta x_{(S,F)} \\
      &\geq \sum_{S : e \in \delta(S)} \sum_{j=1}^{z(S)} \frac{1}{k_S(j)\cdot c_e} \delta x_{\texttt{seq}_S(j)}\\
      & \overset{(\ref{eq:forall_it_k_in_range})}{\geq}  \sum_{S : e \in \delta(S)} \frac{1}{2k^*_S \cdot c_e} \sum_{j=1}^{z(S)}  \delta x_{\texttt{seq}_S(j)}\\
      &\overset{(\ref{eq:gseq_eq_sum_xseq})}{=} \frac{1}{2}\cdot\sum_{S : e \in \delta(S)} \frac{g_{\texttt{seq}^*_S}}{k^*_Sc_e} \\
      &=  \frac{1}{2}\cdot (Ag )_e = \frac{\varepsilon}{2}.
    \end{align*}
The second inequality follows since for every cut $S$, if $e \in \delta(S)$ then $e \not \in F_j$ for all $\{(S,F_j)\}_j$ in the sequence $\texttt{seq}_S$, which follows because of the monotonicity of the weights and the fact that the heavy edges in truncated weights corresponds to free edges in the update procedure.
\end{proof}

\paragraph{Acknowledgements.}
Part of this work was done during the Simons Institute program “Data Structures and Optimization for Fast Algorithms”.
SY would like to thank Kent Quanrud for helpful suggestions and pointers to related work.

\bibliographystyle{alpha}
\bibliography{references,references_sndp}

\newcommand{\etalchar}[1]{$^{#1}$}
\begin{thebibliography}{LLRKS91}

\bibitem[ABCC03]{ApplegateBCC03}
David~L. Applegate, Robert~E. Bixby, Vasek Chv{\'{a}}tal, and William~J. Cook.
\newblock Implementing the dantzig-fulkerson-johnson algorithm for large traveling salesman problems.
\newblock {\em Math. Program.}, 97(1-2):91--153, 2003.

\bibitem[ABCC06]{ApplegateBCC06}
David~L. Applegate, Robert~E. Bixby, Vašek Chvatál, and William~J. Cook.
\newblock {\em The Traveling Salesman Problem: A Computational Study}.
\newblock Princeton University Press, 2006.

\bibitem[Adj18]{adjiashvili2018beating}
David Adjiashvili.
\newblock Beating approximation factor two for weighted tree augmentation with bounded costs.
\newblock {\em ACM Transactions on Algorithms (TALG)}, 15(2):1--26, 2018.

\bibitem[AHK12]{AHK12}
Sanjeev Arora, Elad Hazan, and Satyen Kale.
\newblock The multiplicative weights update method: a meta-algorithm and applications.
\newblock {\em Theory Comput.}, 8:121--164, 2012.

\bibitem[AK08]{AK2008}
Baruch Awerbuch and Rohit Khandekar.
\newblock Stateless distributed gradient descent for positive linear programs.
\newblock In {\em STOC}, 2008.

\bibitem[Aro98]{Arora98}
Sanjeev Arora.
\newblock Polynomial time approximation schemes for euclidean traveling salesman and other geometric problems.
\newblock {\em J. {ACM}}, 45(5):753--782, 1998.

\bibitem[AV84]{AtallahV84}
Mikhail~J. Atallah and Uzi Vishkin.
\newblock Finding euler tours in parallel.
\newblock {\em J. Comput. Syst. Sci.}, 29(3):330--337, 1984.

\bibitem[AZO]{AzO16_ParallelMWU}
Zeyuan Allen-Zhu and Lorenzo Orecchia.
\newblock Using optimization to break the epsilon barrier: A faster and simpler width-independent algorithm for solving positive linear programs in parallel.
\newblock In {\em SODA}.

\bibitem[AZO15]{AO15}
Zeyuan Allen-Zhu and Lorenzo Orecchia.
\newblock Nearly-linear time positive lp solver with faster convergence rate.
\newblock In {\em Proceedings of the Forty-Seventh Annual ACM Symposium on Theory of Computing}, page 229–236, New York, NY, USA, 2015. Association for Computing Machinery.

\bibitem[BBR97]{BBR97}
Yair Bartal, John~W. Byers, and Danny Raz.
\newblock Global optimization using local information with applications to flow control.
\newblock In {\em Proceedings 38th Annual Symposium on Foundations of Computer Science}, pages 303--312. {IEEE} Computer Society, 1997.

\bibitem[BBR05]{BB05}
Yair Bartal, John~W. Byers, and Danny Raz.
\newblock Fast, distributed approximation algorithms for positive linear programming with applications to flow control.
\newblock {\em {SIAM} Journal on Computing}, 33(6):1261--1279, January 2005.

\bibitem[BDT14]{borradaile2014polynomial}
Glencora Borradaile, Erik~D Demaine, and Siamak Tazari.
\newblock Polynomial-time approximation schemes for subset-connectivity problems in bounded-genus graphs.
\newblock {\em Algorithmica}, 68(2):287--311, 2014.

\bibitem[BG07]{berger2007minimum}
Andr{\'e} Berger and Michelangelo Grigni.
\newblock Minimum weight 2-edge-connected spanning subgraphs in planar graphs.
\newblock In {\em International Colloquium on Automata, Languages, and Programming}, pages 90--101. Springer, 2007.

\bibitem[Ble96]{Blelloch96}
Guy~E. Blelloch.
\newblock Programming parallel algorithms.
\newblock {\em Commun. {ACM}}, 39(3):85--97, 1996.

\bibitem[BLS20]{BhardwajLS20}
Nalin Bhardwaj, Antonio~Molina Lovett, and Bryce Sandlund.
\newblock A simple algorithm for minimum cuts in near-linear time.
\newblock In {\em {SWAT}}, volume 162 of {\em LIPIcs}, pages 12:1--12:18. Schloss Dagstuhl - Leibniz-Zentrum f{\"{u}}r Informatik, 2020.

\bibitem[CDE{\etalchar{+}}18]{chalermsook2018survivable}
Parinya Chalermsook, Syamantak Das, Guy Even, Bundit Laekhanukit, and Daniel Vaz.
\newblock Survivable network design for group connectivity in low-treewidth graphs.
\newblock {\em Approximation, Randomization, and Combinatorial Optimization. Algorithms and Techniques}, 2018.

\bibitem[CFLP00]{CarrFLP00}
Robert~D. Carr, Lisa Fleischer, Vitus~J. Leung, and Cynthia~A. Phillips.
\newblock Strengthening integrality gaps for capacitated network design and covering problems.
\newblock In {\em Proceedings of the Eleventh Annual {ACM-SIAM} Symposium on Discrete Algorithms, January 9-11, 2000, San Francisco, CA, {USA}}, pages 106--115. {ACM/SIAM}, 2000.

\bibitem[CGSZ04]{czumaj2004approximation}
Artur Czumaj, Michelangelo Grigni, Papa Sissokho, and Hairong Zhao.
\newblock Approximation schemes for minimum 2-edge-connected and biconnected subgraphs in planar graphs.
\newblock In {\em Proceedings of the fifteenth annual ACM-SIAM symposium on Discrete algorithms}, pages 496--505. Society for Industrial and Applied Mathematics, 2004.

\bibitem[CHN{\etalchar{+}}22]{ChalermsookHNSS22}
Parinya Chalermsook, Chien{-}Chung Huang, Danupon Nanongkai, Thatchaphol Saranurak, Pattara Sukprasert, and Sorrachai Yingchareonthawornchai.
\newblock Approximating k-edge-connected spanning subgraphs via a near-linear time {LP} solver.
\newblock In {\em 49th International Colloquium on Automata, Languages, and Programming, {ICALP} 2022, July 4-8, 2022, Paris, France}, volume 229 of {\em LIPIcs}, pages 37:1--37:20. Schloss Dagstuhl - Leibniz-Zentrum f{\"{u}}r Informatik, 2022.

\bibitem[Chr76]{Christofides76}
Nicos Christofides.
\newblock Worst-case analysis of a new heuristic for the traveling salesman.
\newblock Technical Report 388, Carnegie Mellon University, 1976.

\bibitem[CKK02]{csaba2002approximability}
B{\'e}la Csaba, Marek Karpinski, and Piotr Krysta.
\newblock Approximability of dense and sparse instances of minimum 2-connectivity, tsp and path problems.
\newblock In {\em Proceedings of the thirteenth annual ACM-SIAM symposium on Discrete algorithms}, pages 74--83. Society for Industrial and Applied Mathematics, 2002.

\bibitem[CL99]{czumaj1999approximability}
Artur Czumaj and Andrzej Lingas.
\newblock On approximability of the minimum-cost k-connected spanning subgraph problem.
\newblock In {\em Proceedings of the tenth annual ACM-SIAM symposium on Discrete algorithms}, pages 281--290. Citeseer, 1999.

\bibitem[CL00]{czumaj2000fast}
Artur Czumaj and Andrzej Lingas.
\newblock Fast approximation schemes for euclidean multi-connectivity problems.
\newblock In {\em International Colloquium on Automata, Languages, and Programming}, pages 856--868. Springer, 2000.

\bibitem[CL02]{CarrL02}
Robert~D. Carr and Giuseppe Lancia.
\newblock Compact vs. exponential-size {LP} relaxations.
\newblock {\em Oper. Res. Lett.}, 30(1):57--65, 2002.

\bibitem[Coo12]{Cook12}
William~J. Cook.
\newblock {\em In Pursuit of the Traveling Salesman: Mathematics at the Limits of Computation}.
\newblock Princeton University Press, 2012.

\bibitem[CQ17]{CQ17}
Chandra Chekuri and Kent Quanrud.
\newblock Approximating the {H}eld-{K}arp bound for metric {TSP} in nearly-linear time.
\newblock {\em 2017 IEEE 58th Annual Symposium on Foundations of Computer Science (FOCS)}, pages 789--800, 2017.

\bibitem[CQ18]{arxiv/ChekuriQ18}
Chandra Chekuri and Kent Quanrud.
\newblock Fast approximations for metric-{TSP} via linear programming.
\newblock {\em CoRR}, abs/1802.01242, 2018.

\bibitem[DFJ54]{DantzigFJ54}
George~B. Dantzig, D.~Ray Fulkerson, and Selmer~M. Johnson.
\newblock Solution of a large-scale traveling-salesman problem.
\newblock {\em Oper. Res.}, 2(4):393--410, 1954.

\bibitem[Fer98]{fernandes1998better}
Cristina~G Fernandes.
\newblock A better approximation ratio for the minimum sizek-edge-connected spanning subgraph problem.
\newblock {\em Journal of Algorithms}, 28(1):105--124, 1998.

\bibitem[FGKS18]{fiorini2018approximating}
Samuel Fiorini, Martin Gro{\ss}, Jochen K{\"o}nemann, and Laura Sanit{\`a}.
\newblock Approximating weighted tree augmentation via chv{\'a}tal-gomory cuts.
\newblock In {\em Proceedings of the Twenty-Ninth Annual ACM-SIAM Symposium on Discrete Algorithms}, pages 817--831. SIAM, 2018.

\bibitem[FH92]{FurediH92}
Zolt{\'{a}}n F{\"{u}}redi and P{\'{e}}ter Hajnal.
\newblock Davenport-schinzel theory of matrices.
\newblock {\em Discret. Math.}, 103(3):233--251, 1992.

\bibitem[FJ81]{FredericksonJ81}
Greg~N. Frederickson and Joseph J{\'{a}}J{\'{a}}.
\newblock Approximation algorithms for several graph augmentation problems.
\newblock {\em {SIAM} J. Comput.}, 10(2):270--283, 1981.

\bibitem[FJ82]{FredericksonJ82}
Greg~N. Frederickson and Joseph~F. J{\'{a}}J{\'{a}}.
\newblock On the relationship between the biconnectivity augmentation and traveling salesman problems.
\newblock {\em Theor. Comput. Sci.}, 19:189--201, 1982.

\bibitem[Fle00]{journals/siamdm/Fleischer00}
Lisa Fleischer.
\newblock Approximating fractional multicommodity flow independent of the number of commodities.
\newblock {\em {SIAM} J. Discret. Math.}, 13(4):505--520, 2000.

\bibitem[GB93]{GoemansB93}
Michel~X. Goemans and Dimitris Bertsimas.
\newblock Survivable networks, linear programming relaxations and the parsimonious property.
\newblock {\em Math. Program.}, 60:145--166, 1993.

\bibitem[GG18]{conf/spaa/GeissmannG18}
Barbara Geissmann and Lukas Gianinazzi.
\newblock Parallel minimum cuts in near-linear work and low depth.
\newblock In {\em Proceedings of the 30th on Symposium on Parallelism in Algorithms and Architectures, {SPAA} 2018, Vienna, Austria, July 16-18, 2018}, pages 1--11. {ACM}, 2018.

\bibitem[GGTW09]{gabow2009approximating}
Harold~N Gabow, Michel~X Goemans, {\'E}va Tardos, and David~P Williamson.
\newblock Approximating the smallest k-edge connected spanning subgraph by lp-rounding.
\newblock {\em Networks: An International Journal}, 53(4):345--357, 2009.

\bibitem[GK07]{journals/siamcomp/GargK07}
Naveen Garg and Jochen K{\"{o}}nemann.
\newblock Faster and simpler algorithms for multicommodity flow and other fractional packing problems.
\newblock {\em {SIAM} J. Comput.}, 37(2):630--652, 2007.

\bibitem[GKL24]{GurvitsKL24}
Leonid Gurvits, Nathan Klein, and Jonathan Leake.
\newblock From trees to polynomials and back again: New capacity bounds with applications to {TSP}.
\newblock In {\em 51st International Colloquium on Automata, Languages, and Programming, {ICALP} 2024, July 8-12, 2024, Tallinn, Estonia}, volume 297 of {\em LIPIcs}, pages 79:1--79:20. Schloss Dagstuhl - Leibniz-Zentrum f{\"{u}}r Informatik, 2024.

\bibitem[GKZ18]{grandoni2018improved}
Fabrizio Grandoni, Christos Kalaitzis, and Rico Zenklusen.
\newblock Improved approximation for tree augmentation: saving by rewiring.
\newblock In {\em Proceedings of the 50th Annual ACM SIGACT Symposium on Theory of Computing}, pages 632--645, 2018.

\bibitem[GMW21]{conf/sosa/GawrychowskiMW21}
Pawel Gawrychowski, Shay Mozes, and Oren Weimann.
\newblock A note on a recent algorithm for minimum cut.
\newblock In {\em 4th Symposium on Simplicity in Algorithms, {SOSA} 2021, Virtual Conference, January 11-12, 2021}, pages 74--79. {SIAM}, 2021.

\bibitem[Goe95]{Goemans95}
Michel~X. Goemans.
\newblock Worst-case comparison of valid inequalities for the {TSP}.
\newblock {\em Math. Program.}, 69:335--349, 1995.

\bibitem[GP07]{GuntinP07}
Gregory Gutin and Abraham~P. Punnen.
\newblock {\em The Traveling Salesman Problem and Its Variations}, volume~12 of {\em Combinatorial Optimization}.
\newblock Springer Science+Business Media, 2007.

\bibitem[HK70]{HeldK70}
Michael Held and Richard~M. Karp.
\newblock The traveling-salesman problem and minimum spanning trees.
\newblock {\em Oper. Res.}, 18(6):1138--1162, 1970.

\bibitem[HKZ24]{HershkowitzKZ24}
D.~Ellis Hershkowitz, Nathan Klein, and Rico Zenklusen.
\newblock Ghost value augmentation for k-edge-connectivity.
\newblock In Bojan Mohar, Igor Shinkar, and Ryan O'Donnell, editors, {\em Proceedings of the 56th Annual {ACM} Symposium on Theory of Computing, {STOC} 2024, Vancouver, BC, Canada, June 24-28, 2024}, pages 1853--1864, 2024.

\bibitem[HLRW24]{conf/soda/HenzingerLRW24}
Monika Henzinger, Jason Li, Satish Rao, and Di~Wang.
\newblock Deterministic near-linear time minimum cut in weighted graphs.
\newblock In {\em Proceedings of the 2024 {ACM-SIAM} Symposium on Discrete Algorithms, {SODA} 2024, Alexandria, VA, USA, January 7-10, 2024}, pages 3089--3139. {SIAM}, 2024.

\bibitem[HW96]{HenzingerW96}
Monika Henzinger and David~P. Williamson.
\newblock On the number of small cuts in a graph.
\newblock {\em Inf. Process. Lett.}, 59(1):41--44, 1996.

\bibitem[Kar00]{Karger00}
David~R. Karger.
\newblock Minimum cuts in near-linear time.
\newblock {\em J. {ACM}}, 47(1):46--76, 2000.

\bibitem[Kha04]{Khandekar04}
Rohit Khandekar.
\newblock {\em Lagrangian relaxation based algorithms for convex programming problems}.
\newblock PhD thesis, Indian Institute of Technology Delhi, 2004.

\bibitem[KKGZ22]{KarlinK0Z22}
Anna~R. Karlin, Nathan Klein, Shayan~Oveis Gharan, and Xinzhi Zhang.
\newblock An improved approximation algorithm for the minimum \emph{k}-edge connected multi-subgraph problem.
\newblock In Stefano Leonardi and Anupam Gupta, editors, {\em {STOC} '22: 54th Annual {ACM} {SIGACT} Symposium on Theory of Computing, Rome, Italy, June 20 - 24, 2022}, pages 1612--1620, 2022.

\bibitem[KKOG21]{KarlinKO21}
Anna~R. Karlin, Nathan Klein, and Shayan Oveis~Gharan.
\newblock A (slightly) improved approximation algorithm for metric {TSP}.
\newblock In {\em {STOC} '21: 53rd Annual {ACM} {SIGACT} Symposium on Theory of Computing, Virtual Event, Italy, June 21-25, 2021}, pages 32--45. {ACM}, 2021.

\bibitem[KKOG22]{KarlinKO22}
Anna~R. Karlin, Nathan Klein, and Shayan Oveis~Gharan.
\newblock A (slightly) improved bound on the integrality gap of the subtour {LP} for {TSP}.
\newblock In {\em 63rd {IEEE} Annual Symposium on Foundations of Computer Science, {FOCS} 2022, Denver, CO, USA, October 31 - November 3, 2022}, pages 832--843. {IEEE}, 2022.

\bibitem[KV94]{KhullerV94}
Samir Khuller and Uzi Vishkin.
\newblock Biconnectivity approximations and graph carvings.
\newblock {\em J. {ACM}}, 41(2):214--235, 1994.
\newblock announced at STOC'92.

\bibitem[Lam14]{Lampis14}
Michael Lampis.
\newblock Improved inapproximability for {TSP}.
\newblock {\em Theory Comput.}, 10:217--236, 2014.

\bibitem[LGS12]{laekhanukit2012rounding}
Bundit Laekhanukit, Shayan~Oveis Gharan, and Mohit Singh.
\newblock A rounding by sampling approach to the minimum size k-arc connected subgraph problem.
\newblock In {\em International Colloquium on Automata, Languages, and Programming}, pages 606--616. Springer, 2012.

\bibitem[LLRKS91]{LawlerLRS91}
E.~L. Lawler, Jan~Karel Lenstra, A.~H.~G. Rinnooy~Kan, and D.~B. Shmoys.
\newblock {\em The Traveling Salesman Problem: A Guided Tour of Combinatorial Optimization}.
\newblock John Wiley \& Sons, 1991.

\bibitem[LMN21]{conf/spaa/Lopez-MartinezM21}
Andr{\'{e}}s L{\'{o}}pez{-}Mart{\'{\i}}nez, Sagnik Mukhopadhyay, and Danupon Nanongkai.
\newblock Work-optimal parallel minimum cuts for non-sparse graphs.
\newblock In {\em {SPAA} '21: 33rd {ACM} Symposium on Parallelism in Algorithms and Architectures, Virtual Event, USA, 6-8 July, 2021}, pages 351--361. {ACM}, 2021.

\bibitem[LN93]{LN93}
Michael Luby and Noam Nisan.
\newblock A parallel approximation algorithm for positive linear programming.
\newblock In {\em Proceedings of the Twenty-Fifth Annual ACM Symposium on Theory of Computing}, STOC '93, page 448–457, New York, NY, USA, 1993. Association for Computing Machinery.

\bibitem[Mit99]{Mitchell99}
Joseph S.~B. Mitchell.
\newblock Guillotine subdivisions approximate polygonal subdivisions: {A} simple polynomial-time approximation scheme for geometric tsp, k-mst, and related problems.
\newblock {\em {SIAM} J. Comput.}, 28(4):1298--1309, 1999.

\bibitem[MMP90]{MonmaMP90}
Clyde~L. Monma, Beth~Spellman Munson, and William~R. Pulleyblank.
\newblock Minimum-weight two-connected spanning networks.
\newblock {\em Math. Program.}, 46:153--171, 1990.

\bibitem[MN20]{conf/stoc/MukhopadhyayN20}
Sagnik Mukhopadhyay and Danupon Nanongkai.
\newblock Weighted min-cut: sequential, cut-query, and streaming algorithms.
\newblock In {\em Proceedings of the 52nd Annual {ACM} {SIGACT} Symposium on Theory of Computing, {STOC} 2020, Chicago, IL, USA, June 22-26, 2020}, pages 496--509. {ACM}, 2020.

\bibitem[MRWZ]{conf/icalp/MahoneyRWZ16}
Michael~W. Mahoney, Satish Rao, Di~Wang, and Peng Zhang.
\newblock Approximating the solution to mixed packing and covering lps in parallel o{\texttildelow}(epsilon{\^{}}\{-3\}) time.
\newblock In {\em 43rd International Colloquium on Automata, Languages, and Programming, {ICALP} 2016, July 11-15, 2016, Rome, Italy}, volume~55 of {\em LIPIcs}, pages 52:1--52:14.

\bibitem[Nes05]{Nes05}
Yu. Nesterov.
\newblock Smooth minimization of non-smooth functions.
\newblock {\em Mathematical Programming}, 103(1):127--152, 2005.

\bibitem[NI00]{NagamochiI00}
Hiroshi Nagamochi and Toshihide Ibaraki.
\newblock Polyhedral structure of submodular and posi-modular systems.
\newblock {\em Discret. Appl. Math.}, 107(1-3):165--189, 2000.

\bibitem[NW61]{NW61}
C.~St.J.~A. Nash-Williams.
\newblock {Edge-Disjoint Spanning Trees of Finite Graphs}.
\newblock {\em Journal of the London Mathematical Society}, s1-36(1):445--450, 01 1961.

\bibitem[Pri10]{Pritchard10}
David Pritchard.
\newblock \emph{k}-edge-connectivity: Approximation and {LP} relaxation.
\newblock In {\em {WAOA}}, volume 6534 of {\em Lecture Notes in Computer Science}, pages 225--236. Springer, 2010.

\bibitem[PST95]{PST95}
Serge~A. Plotkin, David~B. Shmoys, and {\'{E}}va Tardos.
\newblock Fast approximation algorithms for fractional packing and covering problems.
\newblock {\em Math. Oper. Res.}, 20(2):257--301, 1995.

\bibitem[Ser78]{Serdyukov78}
A.~I. Serdyukov.
\newblock O nekotorykh ekstremal’nykh obkhodakh v grafakh.
\newblock {\em Upravlyaemye sistemy}, 17:76--79, 1978.

\bibitem[SV82]{ShiloachV82a}
Yossi Shiloach and Uzi Vishkin.
\newblock An o(n{\({^2}\)} log n) parallel {MAX-FLOW} algorithm.
\newblock {\em J. Algorithms}, 3(2):128--146, 1982.

\bibitem[SW90]{ShmoysW90}
David~B. Shmoys and David~P. Williamson.
\newblock Analyzing the held-karp {TSP} bound: {A} monotonicity property with application.
\newblock {\em Inf. Process. Lett.}, 35(6):281--285, 1990.

\bibitem[Wol80]{Wolsey1980}
Laurence~A. Wolsey.
\newblock {\em Heuristic analysis, linear programming and branch and bound}, pages 121--134.
\newblock Springer Berlin Heidelberg, 1980.

\bibitem[You01]{conf/focs/Young01}
Neal~E. Young.
\newblock Sequential and parallel algorithms for mixed packing and covering.
\newblock In {\em 42nd Annual Symposium on Foundations of Computer Science, {FOCS} 2001, 14-17 October 2001, Las Vegas, Nevada, {USA}}, pages 538--546. {IEEE} Computer Society, 2001.

\bibitem[You14]{arxiv/Young14}
Neal~E. Young.
\newblock Nearly linear-time approximation schemes for mixed packing/covering and facility-location linear programs.
\newblock {\em CoRR}, abs/1407.3015, 2014.

\end{thebibliography}

\appendix
\section{Missing Proofs from \Cref{sec:parallel MWU}}
\label{sec:missing_proofs}

\ubound*
\begin{proof}
We proceed by induction on $t$.
The base case $t=0$ is true by our initialization $w^{(0)} = \1_m$.
Next, we assume the inductive hypothesis for some $t$, and consider $t+1$.
By the weight update rule,
\begin{align*}
\ang{\1, w^{(t+1)}} &= \ang{\1, w^{(t)}\circ(\1+Ag^{(t)})} \\
&= \ang{\1,w^{(t)}} + \ang{w^{(t)},Ag^{(t)}} \\
&= \ang{\1,w^{(t)}} \left(1+ \frac{\ang{A^{\top}w^{(t)},g^{(t)}}}{\ang{\1,w^{(t)}}} \right) \\
&< \ang{\1,w^{(t)}} \left(1+ (1+\varepsilon)\frac{\lambda^{(t)}\ang{\1,g^{(t)}}}{\ang{\1,w^{(t)}}} \right) \\
&\leq m \exp \left((1+\varepsilon) \sum_{s=0}^{t-1} \frac{\ang{\1,g^{(s)}}}{\ang{\1,w^{(s)}}/\lambda^{(s)}}\right) \left(1+ (1+\varepsilon)\frac{\ang{\1,g^{(t)}}}{\ang{\1,w^{(t)}}/\lambda^{(t)}} \right) \\
&\leq m \exp \left((1+\varepsilon) \sum_{s=0}^{t} \frac{\ang{\1,g^{(s)}}}{\ang{\1,w^{(s)}}/\lambda^{(s)}}\right).
\end{align*}
The first inequality is because $\supp(g^{(t)})\subseteq B^{(t)}$, the second inequality is by the inductive hypothesis, and the third inequality is due to $1+z\leq e^z$.
\end{proof}

\lbound*

\begin{proof}
By our initialization $w^{(0)} = \1_m$ and the weight update rule,
\begin{align*}
    w^{(t)}_i  = \prod_{s=0}^{t-1} \left(1+(Ag^{(s)})_i\right) &\geq \exp\left(\sum_{s=0}^{t-1} (Ag^{(s)})_i(1-(Ag^{(s)})_i)\right) \\
    &\geq \exp\left((1-\varepsilon)\sum_{s=0}^{t-1} (Ag^{(s)})_i\right) = \exp((1-\varepsilon)(Ax^{(t)})_i).
\end{align*}
The first inequality follows from $1+z\geq e^{z(1-z)}$ for all $z\geq 0$, whereas the second inequality is due to $\|Ag^{(s)}\|_\infty\leq \varepsilon$ for all $s$.
\end{proof}

\correctness*

\begin{proof}
Suppose that the algorithm terminated in iteration $T$.
Let $i\in [m]$ be a row with $(Ax^{(T)})_i = \|Ax^{(T)}\|_\infty \geq \eta$.
By \Cref{lem:weight-upper-bound} and \Cref{lem:weight-lower-bound},
\[m \exp \left((1+\varepsilon) \sum_{t=0}^{T-1} \frac{\ang{\1,g^{(t)}}}{\ang{\1,w^{(t)}}/\lambda^{(t)}}\right)\geq \ang{\1,w^{(T)}} \geq w^{(T)}_i \geq e^{(1-\varepsilon)(Ax^{(T)})_i}\]
Taking logarithms on both sides yields
\[\ln(m) + (1+\varepsilon)\frac{\ang{\1,x^{(T)}}}{\min_{t}\ang{\1,w^{(t)}}/\lambda^{(t)}} \geq \ln(m) + (1+\varepsilon) \sum_{t=0}^{T-1} \frac{\ang{\1,g^{(t)}}}{\ang{\1,w^{(t)}}/\lambda^{(t)}} \geq (1-\varepsilon)(Ax^{(T)})_i\]
Rearranging gives
\[\frac{\ang{\1,x^{(T)}}}{(Ax^{(T)})_i} \geq \frac{(1-\varepsilon) - \ln (m)/(Ax^{(T)})_i}{1+\varepsilon} \min_t \frac{\ang{\1,w^{(t)}}}{\lambda^{(t)}}\]
Since $(Ax^{(T)})_i \geq \eta = \ln (m)/\varepsilon$, we obtain
\[\frac{\ang{\1,x^{(T)}}}{(Ax^{(T)})_i} \geq \frac{1-2\varepsilon}{1+\varepsilon} \min_t \frac{\ang{\1,w^{(t)}}}{\lambda^{(t)}}.\]
The proof is complete by observing that $x^{(T)}/(Ax^{(T)})_i$ is feasible to \eqref{lp:pack}, and $w^{(t)}/\lambda^{(t)}$ is feasible to \eqref{lp:cover} for all $t$, where the latter is due to $\lambda^{(t)}\leq \min_{j\in [N]} \ang{w^{(t)}, A_j}$.
\end{proof}

\epochs*

\begin{proof}
Let $OPT$ be the optimal value of \eqref{lp:pack} and \eqref{lp:cover}.
For every iteration $t$, $w^{(t)}/\min_{j\in [N]} (A^\top w^{(t)})_j$ is a feasible solution to \eqref{lp:cover}.
Hence, the following invariant holds throughout
\[\min_{j\in [N]} (A^\top w^{(t)})_j\leq \frac{\ang{\1,w^{(t)}}}{OPT}.\]

First, we claim that $\ang{\1,w^{(t)}} < m^{1+1/\varepsilon}$ as long as the algorithm does not terminate in iteration $t$.
From the weight update rule,
\[w^{(t)}_i = \prod_{s=0}^{t-1} (1+(Ag^{(s)})_i) \leq \exp\left(\sum_{s=0}^{t-1} (Ag^{(s)})_i\right) = \exp((Ax^{(t)})_i).\]
If $\|w^{(t)}\|_1 \geq m^{1+1/\varepsilon}$, then $\|w^{(t)}\|_\infty \geq m^{1/\varepsilon}$.
It follows that $\|Ax^{(t)}\|_\infty \geq \eta$, so the algorithm would have terminated in iteration $t$.
Next, denoting $A_j$ as the $j$th column of $A$, we have
\[\min_{j\in [N]} (A^\top w^{(0)})_j = \min_{j\in [N]} \|A_j\|_1 \geq \min_{j\in [N]} \|A_j\|_\infty \geq \frac{1}{OPT}.\]
The last inequality follows from the observation that $e_j/\|A_j\|_\infty$ is a feasible solution to \eqref{lp:pack} for all $j\in [N]$.

Since we initialized $\lambda^{(0)} := \min_{j\in [N]} (A^\top w^{(0)})_j$ and the invariant $\lambda^{(t)} \leq \min_{j\in [N]} (A^\top w^{(t)})_j$ holds throughout, $\lambda^{(t)}$ can only increase by a factor of at most $m^{1+1/\varepsilon}$. Thus, the number of epochs is at most $\log_{1+\varepsilon}(m^{1+1/\varepsilon}) = O(\log m/\varepsilon^2)$.
\end{proof}

\iterations*

\begin{proof}
By \Cref{lem:epochs}, it suffices to show that every epoch has $O(\log m \log(|B^{(t)}|\eta/\varepsilon)/\varepsilon^2)$ iterations.
Fix an epoch, and let $t_0$ and $t_1$ be the first and last iteration of this epoch respectively.
We claim that $\delta>\eta/\varepsilon$ during iterations $t_0 < t < t_1$.
Fix such a $t$.
There exists an $i\in [m]$ such that
\[\varepsilon  = (Ag^{(t)})_i \leq (A(\delta x^{(t)}))_i = \delta (Ax^{(t)})_i < \delta\eta.\]
The first inequality is due to the nonnegativity of $A$ and $g^{(t)}\leq \delta x^{(t)}$, while the second inequality is because the algorithm did not terminate in iteration $t$.
Hence, $\delta > \varepsilon/\eta$.

Now, let $j'\in B^{(t_1-1)}$.
Since the weights are nondecreasing, it follows that $j'\in B^{(t)}$ for all $t_0 \leq t < t_1$.
By our initialization, we have $x^{(t_0+1)}_{j'} \geq \varepsilon/(|B^{(t_0)}|\max_{i\in [m]} A_{i,j'})$.
We also know that $x^{(t_1-1)}_{j'}< \eta/\max_{i\in [m]} A_{i,j'}$ by our termination condition.
Since $x^{(t)}_{j'}$ increases by a factor of at least $1+\varepsilon/\eta$ for all $t_0< t < t_1$, the number of iterations in this epoch is at most
\[O\left(\log_{1+\varepsilon/\eta}\left(\frac{\eta |B^{(t_0)}|}{\varepsilon}\right)\right) = O\left(\frac{\eta\log(\eta |B^{(t_0)}|/\varepsilon)}{\varepsilon}\right) = O\left(\frac{\log (m)\log(\eta |B^{(t_0)}|/\varepsilon)}{\varepsilon^2}\right). \qedhere\]
\end{proof}

\end{document}